\documentclass[conference,compsoc]{IEEEtran}

\ifCLASSOPTIONcompsoc
  \usepackage[nocompress]{cite}
\else
  \usepackage{cite}
\fi

\ifCLASSINFOpdf
\else
\fi

\usepackage{amsmath}

\usepackage{url}

\usepackage{xcolor}
\usepackage{enumitem}
\usepackage{comment}
\usepackage{amsfonts}
\usepackage{amsthm}
\usepackage{bbm}
\usepackage{booktabs}
\usepackage{graphicx}
\usepackage{tabularx}
\usepackage[differential]{boldtensors}
\usepackage{multicol}
\usepackage{multirow}
\usepackage{adjustbox}
\usepackage{float}
\usepackage[normalem]{ulem}
\usepackage{hyperref}

\theoremstyle{definition}
\newtheorem{definition}{Definition}[section]

\newtheorem{proposition}{Proposition}
\newtheorem{theorem}{Theorem}
\newtheorem{corollary}{Corollary}
\newtheorem{assumption}{Assumption}
\newtheorem{remark}{Remark}

\newcommand{\topic}[1]{\noindent\textbf{#1}}
\newcommand{\lr}[1]{\left(#1\right)}
\newcommand{\nn}[1]{\left|#1\right|}
\newcommand{\IP}[2]{\left\langle#1,#2\right\rangle}

\usepackage{tikz}

\begin{document}
\title{%
    Modeling Neural Networks with Privacy Using \\
    Neural Stochastic Differential Equations}

\author{
\IEEEauthorblockN{%
Sanghyun Hong\IEEEauthorrefmark{1}, 
Fan Wu\IEEEauthorrefmark{2}, 
Anthony Gruber\IEEEauthorrefmark{3}, 
Kookjin Lee\IEEEauthorrefmark{2},
}
\IEEEauthorblockA{\IEEEauthorrefmark{1}Oregon State University}
\IEEEauthorblockA{\IEEEauthorrefmark{2}Arizona State University}
\IEEEauthorblockA{\IEEEauthorrefmark{3}Sandia National Laboratories}
}

\maketitle

\begin{abstract}
In this work, we study the feasibility of using
neural ordinary differential equations (NODEs) 
to model systems with intrinsic privacy properties.
Unlike conventional feedforward neural networks,
which have unlimited expressivity and
can represent arbitrary mappings between inputs and outputs,
NODEs constrain their learning to the solution of a system of differential equations.
We first examine whether this constraint reduces memorization and,
consequently, the membership inference risks associated with NODEs.
We conduct a comprehensive evaluation of NODEs
under membership inference attacks and show that
they exhibit twice the resistance compared to 
conventional models such as ResNets.
By analyzing the variance in membership risks across different NODE models,
we find that their limited expressivity leads to 
reduced overfitting to the training data.
We then demonstrate, both theoretically and empirically,
that membership inference risks can be further mitigated by
utilizing a stochastic variant of NODEs:
neural stochastic differential equations (NSDEs).
We show that NSDEs are differentially-private (DP) learners
that provide the same provable privacy guarantees as DP-SGD,
the de-facto mechanism for training private models.
NSDEs are also effective in mitigating membership inference attacks,
achieving risk levels comparable to private models trained with DP-SGD
while offering an improved privacy-utility trade-off.
Moreover, we propose a drop-in-replacement strategy 
that efficiently integrates NSDEs into 
conventional feedforward architectures to enhance their privacy.
\end{abstract}

\IEEEpeerreviewmaketitle

\section{Introduction}
\label{sec:intro}

Neural networks are ``universal function approximators":
they can learn arbitrary continuous mappings 
from training data to target outputs~\cite{hornik1989multilayer}.
This expressivity allows highly parameterized networks
to model complex, non-linear relationships and 
capture hidden patterns that traditional methods may struggle to detect.
However, this strength also introduces an undesirable property: \emph{memorization}~\cite{feldman2020neural}.
Memorization occurs when the learned mappings are overly specific to 
individual training instances---regardless of whether 
these mappings generalize to unseen data---%
potentially leading to the leakage of their membership
through model queries~\cite{yeom, shokri, songattack, lira, zhang2024does}.
Given that neural networks are often deployed in sensitive domains,
such as healthcare, where training data may include confidential medical records,
this membership leakage naturally raises privacy concerns.

Recent work has introduced a new class of neural networks:
neural ordinary differential equations (NODEs)~\cite{node}.
Unlike conventional neural networks, 
NODEs constrain the learned mappings to
solutions of parameterized differential equations%
---a fundamental mathematical tool 
for describing processes in the physical world.
By shifting from manually designing governing equations to data-driven modeling, 
NODEs have emerged as a powerful paradigm in scientific computing
for modeling dynamical systems across various domains,
such as time-series forecasting%
~\cite{kidger2020neuralcde, 
       jeon2022gtgan, wu2024identifying}, %
scientific machine learning (ML)%
       ~\cite{greydanus2019hamiltonian, vermaclimode,
       lee2021machine, gruber2024reversible},
and computer vision tasks%
~\cite{dupont2019augmented, yildiz2019ode2vae, xia2021heavy, cho2022adamnodes}. %
These works demonstrate that NODEs,
despite their constrained expressivity to learnable differential equations,
can achieve performance comparable to that of conventional neural networks
across a range of tasks.

In this work, we study the privacy implications of employing
differential equation-based neural networks.
Specifically, we ask the following research question:
\emph{Does this emerging class of neural networks (and their variants)
offer inherent privacy advantages, such as reduced membership inference risks, 
compared to conventional neural networks?}
Because NODEs restrict model expressivity to the system of differential equations, 
they may influence memorization and overfitting%
---key factors attributing to membership inference risks.
Their stochastic variants, 
neural stochastic differential equations (NSDEs),
introduce randomness into the model dynamics and thus,
may satisfy formal privacy guarantees.
Moreover, the modular design of these models
could enable their integration into conventional feedforward architectures 
as privacy-enhancing components.

We first conduct a systematic evaluation of 
the membership inference risks associated with NODEs.
To this end, we develop a risk analysis framework 
that runs six off-the-shelf membership inference attacks
across a variety of neural networks
including NODEs and baseline feedforward models.
In our evaluation on four image recognition benchmarks,
commonly used in prior work on assessing membership risks,
we find that NODEs exhibit up to $2\times$ \emph{lower} membership risks
compared to conventional feedforward networks, like ResNets~\cite{he2016deep}.
We also find that this reduced risk is partly 
attributable to the lower degree of overfitting in NODEs,
while achieving accuracy comparable to baselines.

Moreover, we analyze factors that may %
affect their membership risks.
We first examine configurations
one can control at the block-level: solvers, tolerance, and step size.
Our analysis indicates that these factors,
which change the complexity of learned dynamics
while keeping the same model capacity,
can increase/decrease the likelihood of overfitting.
We also test the model-level configurations:
the number of blocks and advanced variants of NODEs%
~\cite{dupont2019augmented, xia2021heavy, norcliffe2020second}.
They increase the model capacity 
or the complexity of learned dynamics, 
and therefore, increase membership risks.

Our previous evaluation shows that,
while NODEs are less vulnerable, 
the risks to membership inference still exist.
We address this problem
by employing a stochastic adaptation of NODEs: NSDEs~\cite{liu2019neural}.
NSDEs also model underlying dynamics using differential equations,
but augment them with a diffusion term that introduces Gaussian noise.

We first formally show that
this diffusion term acts as a differentially-private (DP) mechanism 
during training, and turns NSDEs into DP-learners.
We establish a theoretical bound on privacy leakage 
per mini-batch SGD iteration and 
present a mechanism to account for total privacy leakage 
over the entire training process.
Surprisingly, we find that the total privacy leakage ($\varepsilon$)
guaranteed by NSDE training matches that of DP-SGD~\cite{dpsgd},
the de-facto standard for training private models.
The leakage $\varepsilon$ is proportional to 
the number of training iterations and the noise level ($\sigma$).

We then empirically evaluate the effectiveness of NSDEs 
in further reducing the membership inference risks.
Our results show that NSDEs reduce the effectiveness of
existing membership inference attacks by 1.8--10$\times$ 
compared to NODEs.
Despite offering the same privacy guarantee,
NSDE models maintain utility that is comparable to
(or slightly higher than) ResNet models trained with DP-SGD.
This shows NSDEs as a promising countermeasure 
that achieves an improved privacy-utility trade-off.
We also compare NSDEs with prior heuristic defenses~\cite{DBLP:journals/corr/abs-2006-05336}
designed to counter membership inference attacks 
but lacking formal privacy guarantees.
We show that NSDEs are $2$--$5\times$ more effective 
in mitigating the membership inference risks.

Implementing our privacy mechanism
requires training an NSDE model from scratch,
which can be a computationally demanding option 
for %
model providers.
To address this practical challenge, 
we propose an effective strategy to apply NSDEs 
to existing pre-trained models,
particularly conventional model architectures like ResNets.
We leverage the common practice of \emph{transfer learning}
where the last layer of a pre-trained model 
is replaced with a classification head,
followed by fine-tuning on data for a few iterations.
Our proposed strategy, \emph{replace-then-finetune}, 
substitutes the last few layers of a pre-trained model 
with an NSDE block.
We test our strategy on ResNet14 pre-trained on %
the four benchmark tasks.
Across all benchmarks,
the fine-tuned models maintain accuracy 
comparable to the pre-trained models,
while reducing the membership risks %
substantially (e.g., by 10$\times$ in CIFAR-10).
We hope our work encourages the adoption of NSDEs
as a viable mechanism for further reducing privacy risks
while maintaining model utility.
\smallskip

\topic{Contributions.}
In summary, our contributions are:
\begin{itemize}[topsep=0.2em, itemsep=0.2em, leftmargin=1.4em]
    \item We are the first to study the membership inference risks of NODEs.
    Our comprehensive evaluation shows that NODEs exhibit 2$\times$ lower membership risks compared to baseline models. We attribute this to NODEs' reduced overfitting, which is from their restricted expressivity---learning within the constraints of a system of ODEs.
    \item We propose NSDEs as a defense mechanism to mitigate membership inference attacks. We formally show that NSDEs provide the same provable privacy guarantee as DP-SGD. We also present an accounting mechanism to track the total privacy leakage during NSDE training.
    \item In our evaluation, we show that NSDEs reduce membership inference risks by 2--5$\times$ compared to heuristic defenses without provable guarantees. Compared to private models trained with DP-SGD, NSDEs offer comparable (provable) protection against membership inference attacks while achieving slightly higher accuracy.
    \item We present a practical strategy to leverage these privacy benefits 
    without fully replacing existing models with NSDEs: \emph{replace-then-finetune}. 
    This strategy reduces membership inference risks by 10$\times$ 
    while achieving a higher utility than private models traind with DP-SGD.
\end{itemize}

\section{Background}
\label{sec:prelim}

\subsection{Neural Ordinary Differential Equations}
\label{subsec:neural-odes}

NODEs model the dynamics of system states as %
a system of ODEs such that:
\begin{align*}
    \frac{\mathrm d ~h(t)}{\mathrm d t} = ~f_{\Theta}(~h(t), t), 
\end{align*}
where $~h(t)$ are the state variables and
$~f_{\Theta}$ is a neural network,
parameterized by its parameters $\Theta$,
defining the rate of change in $~h(t)$.
The scalar variable $t$ %
refers to a time variable 
and the system states correspond to observables in time-dependent physical processes~\cite{strogatz2018nonlinear}
or measurements in time-series data (e.g., \cite{rubanova2019latent}). 
However, in the context of hidden representations 
within a deep feedforward network (FFN),
$t$ instead refers to %
a continuous depth parameter. More precisely, recall that the flow map $~\varphi = ~\varphi(t,~h)$ of an ODE system associates to each state $~h$ its ``flowed out'' state after evolving along the ODE for $t$ time, i.e., $~\varphi(0,~h)=~h$ and $(d/dt)~\varphi(t,~h) = ~f_{\Theta}(~h(t), t)$ in the case of the NODE above.  Then, the number of ``layers'' in the NODE architecture defined by $~f_{\Theta}$ is analogous to the number of ``time steps'' used to generate the flowed-out state $~h(t)$, and hence each network prediction based on this quantity.  
Our focus is on this use case, which allows for depth control through direct manipulation of the ODE time-stepping. %
\smallskip

\topic{Training and inference.}
In NODEs, the forward pass involves solving an initial value problem (IVP) describing the evolution of some given initial states under the flow of the learnable ODE in question.
Given an initial condition $~h(0)=~h_0$ for the state variables (or hidden state representations),
the states at any time index~$\{~h(t_i)\}_{i=1}^{n}$ 
can be obtained by solving the associated IVP, lending NODEs their interpretation as continuous-depth FFNs.
To solve the necessary IVP, a black-box ODE time-integrator with solver options \textsc{opt} can be employed: 
\begin{equation*}
    ~h(t_1),\ldots, ~h(t_n) = \text{ODESolve}(~f_{\Theta}, ~h(0), [t_1,\ldots,t_n], \textsc{opt}).
\end{equation*}
Different choices of numerical time-integration yield different realizations of FNNs, e.g., under the forward Euler method with step size 1, autonomous ($t$-independent) NODEs degenerate to ResNets~\cite{he2016deep}, $~h(t_{i+1}) = ~h(t_i) + ~f_{\Theta}(~h(t_i))$. With this setup, the forward pass for both training and inference is equivalent to recurrently applying a nonlinear transformation, defined by $~f_{\Theta}$, to a hidden state, followed by a skip connection. In classification tasks, the last hidden representation $~h(t_n)$ is fed into a classifier network, producing logits, and the loss of the model is obtained via computing the cross-entropy. The continuous-depth network $~M(t, ~h) = ~M_t~(~h) \approx ~\varphi(t,~h)$ can then be thought of as discretizing the ODE flow map, i.e., $~M_{t_n}(~h_0) = ~h(t_n)$, where $~h(t_n)$ is the solution of the IVP defined above.  As mentioned before, any intermediate time steps $t_1,...,t_{n-1}$ used in the production of the final state  $~h(t_n)$ are analogous to the labels on hidden layers in an FFN. 

For the backward process required during training, 
NODEs use either standard backpropagation through the rolled-out computational path, 
or the adjoint sensitivity method~\cite{pontryagin2018mathematical}.
The primary challenge lies in achieving 
stable and memory-efficient gradient computation for $~f_\Theta$ 
with respect to its parameters, 
as backpropagating directly through an iterative numerical solver can be memory intensive and exacerbate vanishing or exploding gradient issues.  
\smallskip

\topic{NODE variants.}
Many variants have been developed 
to improve their expressivity and efficiency.
Augmented NODEs (ANODEs)~\cite{dupont2019augmented} 
propose an additional dimension to the state variable $[~h(t), ~a(t)]$,
to increase expressivity.
ANODEs were further extended into second-order NODEs  (SONODEs)~\cite{norcliffe2020second},
which formulates a system of coupled first-order NODEs, 
implementing the broader class of augmented models.
Other variants incorporate concepts from gradient descent optimization,
such as momentum, resulting in heavy ball NODEs (HBNODEs)~\cite{xia2021heavy} 
and Nesterov NODEs~\cite{nguyenimproving}, 
both of which improve computational efficiency 
by reducing the number of function evaluations in the forward pass. 
\smallskip

NSDEs~\cite{liu2019neural, li2020scalable, kidger2021neural} extend NODEs to the modeling of stochastic processes, typically introducing randomness
in the form of additive Gaussian noise.  These systems can be formulated in terms of deterministic, vector-valued ``drift'' and stochastic, matrix-valued ``diffusion'' terms as follows:
\begin{align*}
    \mathrm d ~h = ~f_{\Theta_1}(~h,t)\, \mathrm d t + ~G_{\Theta_2}(~h,t)\,\mathrm d ~B_t,
\end{align*}
where $d~B_t$ denotes the Wiener increment\textemdash a zero mean, unit variance, delta-correlated Gaussian process describing ``white noise''.
A simple but popular discretization of NSDEs is given by the Euler--Maruyama method~\cite{platen2010numerical} with step size 1,
expressed as $~h(t_{i+1}) = ~h(t_i) + ~f_{\Theta_1}(~h, t_i)\Delta t + ~G_{\Theta_2}(~h,t_i)~w_{t_i}$, 
where $\Delta t$ is the time step size and $~w_{t_i} \sim \mathcal N(0,\Delta t)$ is randomly sampled.
In this formulation, Gaussian noise is injected alongside each residual connection. 
Because of this property, NSDEs have shown enhanced robustness 
against adversarial examples~\cite{liu2020how} compared to standard NODEs.
Our study further explores this property
as a formal mechanism for mitigating privacy risks.

\subsection{Membership Inference Attacks}
\label{subsec:mia}

Membership inference attacks aim to determine
whether a specific sample is a \emph{member} of the training data.
To do this, the attacker exploits the difference
in the target model's response to the specific sample
when it is a member versus a non-member.
Membership inference can be considered a threat on its own,
but a model's vulnerability to inference attacks
also implies its potential to leak other private information outside of this context,
an idea which aligns closely with the definition of differential privacy~\cite{dwork2006differential}. %
\smallskip

\topic{Existing attacks.}
Prior work has developed various attacks 
to exploit differences and identify \emph{membership}.
In our work, we evaluate NODEs
against each of the representative attacks described below:

\begin{itemize}[itemsep=0.1em, leftmargin=1.4em]
    \item Yeom \textit{et al.}~\cite{yeom} %
    focus on differences in \emph{loss} values:
    for a specific instance $z = (x, y)$,
    a model will be more accurate in its predictions
    when it has seen $z$ during training.
    The attack predicts $z$ as a member 
    if and only if, for a threshold $\tau$, 
    the model's loss on $z$ is below $\tau$; 
    otherwise, it is classified as a non-member.
    \item Shokri \textit{et al.}~\cite{shokri} %
    leverage \emph{shadow models} which better approximate 
    the differences in a model's responses
    to members versus non-members.
    When constructing shadow models,
    the attacker artificially generates datasets
    for training and testing %
    so that the members and non-members are known in advance.
    The attacker collects responses from shadow models
    for both members and non-members and 
    trains a classifier to predict membership of $z$
    based on the target model's response.
    \item Song and Mittal~\cite{songattack}
    leverage the \emph{prediction correctness}
    to compute the threshold $\tau$ for identifying membership. 
    Instead of training classifiers to perform attacks,
    $\tau$ is designed such that
    correct predictions with high confidence yield the lowest values
    while confident but incorrect predictions achieve the highest.
    \item Watson~\textit{et al.}~\cite{watson2022on}
    proposed {\it per-example difficulty} calibration
    in which an attacker leverages shadow models 
    trained without a particular example 
    to compute the average confidence level of a model.
    The average is then subtracted 
    from the example's confidence obtained from 
    the target model to calibrate the difficulty of the sample.
    \item Carlini~\textit{et al.}~\cite{lira}
    introduced Likelihood Ratio Attack (LiRA),
    which refines how shadow training data is used by the adversary.
    For each sample in the shadow training data,
    LiRA trains two sets of shadow models:
    one set with the sample included (``in") and another without (``out").
    The attacker collects logits
    from both members and non-members across these shadow models,
    computes a membership score for a sample $z$
    as the ratio of the likelihoods of observing its logit
    in ``in" versus ``out" models, and 
    uses this score to identify membership.
    
    \item The latest work by Zarifzade~\textit{et al.}~\cite{zarifzadeh2023low} 
    introduced a Robust Membership Inference Attack (RMIA), 
    which models membership inference as a fine-grained likelihood-ratio test.
    RMIA operates by contrasting the target model's output distribution 
    against a reference distribution, 
    constructed from a small set of reference models and population data.
    This approach enables more robust differentiation between members and non-members. 
\end{itemize}
\smallskip 

\topic{Metrics.}
Initial work~\cite{yeom} %
uses inference accuracy as a metric for evaluating attack success.
Subsequent works use the area under the ROC curve (AUC)
as an additional metric to report the balance between
true-positive (TPR) and false-positive rates (FPR).
More recently, Carlini~\cite{lira} proposes evaluating attack success 
in the worst-case \textit{low-FPR regime}.
Our work reports all these metrics in evaluation.

\subsection{Defenses against Membership Inference}
\label{subsec:mia-defenses}

Prior work has developed defenses %
to obscure the difference in model output between members and non-members.
\\ \vspace{-0.8em}

\topic{Reducing overfitting.}
One key difference is the degree of overfitting 
or the extent to which a model memorizes its training data. 
Yeom~\textit{et al.}~\cite{yeom} found that 
defending against membership inference attacks
and reducing overfitting go hand-in-hand,
so that the goals of privacy and performance are aligned.
Many defenses~\cite{shokri, songattack, 10.1145/3319535.3363201} 
therefore focus on a regularization effect,
and existing regularization methods are also being used 
to reduce membership risk~\cite{DBLP:journals/corr/abs-2006-05336}.

\begin{itemize}[itemsep=0.1em, leftmargin=1.4em]
    \item \emph{$\ell^p$-regularization and dropout}
    are two of the most commonly used regularization techniques in practice,
    shown effective against membership inference attacks~\cite{shokri}.
    \item \emph{Early stopping} is a straightforward way to prevent overfitting~\cite{songattack}.
    One halts training when the model's performance stops improving.
    This has the effect of stopping the model before it can be begin to memorize.
    \item \emph{MMD with Mix-up}~\cite{10.1145/3422337.3447836} 
    combines two different mechanisms: 
    (1) a maximum mean discrepancy (MMD) loss %
    to reduce the training accuracy to match validation accuracy, and
    (2) training on mix-up augmented data,
    where a linear interpolation between two randomly drawn examples
    (and their labels) is used instead of individual samples.
    Mix-up has been shown to increase overall accuracy 
    and further reduce overfitting.
    \item \emph{MemGuard}~\cite{10.1145/3319535.3363201} 
    uses an ``attack-as-a-defense" approach:
    It adds carefully-calibrated noise to each confidence score vector 
    predicted by the target model. Once added, 
    the resulting adversarial examples mislead 
    the attacker's membership identification (e.g., the attack classifiers).
\end{itemize}
\smallskip

\topic{Differential privacy}~\cite{dwork2006differential} (DP)
is a mathematical framework that provides 
a probabilistic privacy guarantee.  Precisely, the following definition implies that,  as a random variable, the privacy loss is bounded in the worst case by a constant $\varepsilon>0$ with probability $0 < 1-\delta < 1$.

\begin{definition}[($\varepsilon,\delta$)-DP]
Given any two datasets $D$ and $D'$ that differ by only a single record, 
a mechanism $M$ satisfies ($\varepsilon,\delta$)-DP if, 
for any subset $S$ of the image of $M$,
\begin{align*}
    Pr[M(D) \in S] \leq e^{\varepsilon} Pr[M(D') \in S] + \delta
\end{align*}
where $\varepsilon$ is the privacy budget (or leakage) and 
$\delta$ is the failure probability.
\end{definition}

A primary advantage of DP is its immunity to post-processing: 
an $(\varepsilon,\delta)$-DP algorithm cannot be weakened through manipulation of its output.
Based on this property,
Abadi~\textit{et al.} 
proposed an effective, differentially-private adaptation of 
stochastic gradient descent: DP-SGD~\cite{dpsgd}.
The key in DP-SGD is the \emph{moment accountant},
a mechanism that tracks total privacy leakage 
possible under a worst-case adversary.
This leakage is typically denoted as $\varepsilon$,
a convention that  will be followed throughout the rest of the paper.
Because of the formal privacy guarantee that DP-SGD offers,
it has become a de-facto standard for constructing private models.
However, this approach inherently leads to a utility loss
due to the addition of Gaussian noise to the gradients
required for accurate learning.
We evaluate how effective those defenses are
in mitigating the membership risks of NODEs
and formally demonstrate that NSDEs are DP learners,
offering an improved privacy-utility tradeoff.

\section{Membership Risks of NODEs}
\label{sec:membership-risks}

\begin{table*}[t]
\centering
\caption{%
    \textbf{Comparison of membership inference risks.}
    We evaluate two different types of neural networks against six existing attacks on four benchmarks. We compare their risks based on four different metrics used in the prior work.
}
\label{tbl:main-mia}
\adjustbox{max width=\linewidth}{
\begin{tabular}{@{}ll | cccc | cccc | cccc | cccc@{}}
\toprule
 &  %
 & \multicolumn{4}{c|}{\textbf{TPR @ 0.1\% FPR}} & \multicolumn{4}{c|}{\textbf{TPR @ 1\% FPR}} & \multicolumn{4}{c|}{\textbf{AUC}} & \multicolumn{4}{c}{\textbf{Inference acc.}} \\ \midrule
\begin{tabular}{l}\textbf{Model}\end{tabular} & 
    \textbf{Method} &
    \textbf{F-M} & \textbf{C-10} & \textbf{C-100} & \textbf{T-I} & 
    \textbf{F-M} & \textbf{C-10} & \textbf{C-100} & \textbf{T-I} & 
    \textbf{F-M} & \textbf{C-10} & \textbf{C-100} & \textbf{T-I} & 
    \textbf{F-M} & \textbf{C-10} & \textbf{C-100} & \textbf{T-I} \\ \midrule \midrule
\multirow{6}{*}{\begin{tabular}{l}ResNet-14\end{tabular}} & 
    Yeom et al.~\cite{yeom} & 
    0.00\% & 0.00\% & 0.00\% & 0.02\% & 
    0.00\% & 0.00\% & 0.69\% & 1.11\% & 
    0.563 & 0.590 & 0.767 & 0.683 &
    59.91\%& 58.92\% & 73.81\% & 64.27\% \\
    & Shokri et al.~\cite{shokri} & 
    0.05\% & 0.01\% & 0.01\% & 0.10\% & 
    0.50\% & 0.14\% & 0.06\% & 1.01\% & 
    0.468 & 0.411 & 0.336 & 0.501 & 
    50.00\% & 50.00\% & 50.00\% & 50.42\% \\
    & Song and Mittal~\cite{song2021scorebased} & 
    0.00\% & 0.00\% & 0.25\% & 0.13\% & 
    0.23\% & 0.00\% & 1.76\% & 1.23\% & 
    0.471 & 0.590 & 0.494 & 0.499 & 
    51.29\% & 58.92\% & 51.19\% & 50.28\% \\
    & Watson et al.~\cite{watson2022on} & 
    0.00\% & 0.00\% & 0.00\% & 0.12\% &
    0.00\% & 0.00\% & 0.62\% & 1.21\% &
    0.554 & 0.590 & 0.767 & 0.548 &
    56.21\% & 58.92\% & 73.81\% & 53.43\% \\
    & Carlini et al.~\cite{lira} & 
    1.92\% & 3.96\% & 10.52\% & 0.28\% & 
    5.87\% & 10.57\% & 25.95\% & 1.97\% &  
    0.599 & 0.679 & 0.865 & 0.554 & 
    55.70\% & 61.15\% & 76.46\% & 53.72\% \\ 
    & Zarifzadeh et al.~\cite{zarifzadeh2023low} & 
    1.13\% & 3.73\% & 13.98\% & 0.32\% & 
    6.54\% & 8.78\% & 28.25\% & 2.25\% & 
    0.497 & 0.581 & 0.858 & 0.578 & 
    56.83\% & 59.13\% & 75.53\% & 55.39\% \\
    \midrule
\multirow{6}{*}{\begin{tabular}{l}NODE\end{tabular}} & 
    Yeom et al.~\cite{yeom} & 
    0.00\% & 0.00\% & 0.01\% & 0.09\% &
    0.06\% & 0.00\% & 0.71\% & 1.12\% &
    0.518 & 0.558 & 0.701 & 0.645 & 
    52.55\% & 56.59\% & 70.75\% & 61.23\% \\
    & Shokri et al.~\cite{shokri} & 
    0.09\% & 0.04\% & 0.04\% & 0.10\% &
    0.90\% & 0.43\% & 0.37\% & 1.01\% &
    0.493 & 0.448 & 0.398 & 0.502 & 
    50.15\% & 50.01\% & 50.00\% & 50.21\% \\
    & Song and Mittal~\cite{song2021scorebased} & 
    0.00\% & 0.00\% & 0.15\% & 0.12\% &
    0.53\% & 0.00\% & 1.66\% & 1.19\% &
    0.507 & 0.529 & 0.498 & 0.500 &
    50.98\% & 54.22\% & 52.48\% & 50.45\% \\
    & Watson et al.~\cite{watson2022on} & 
    0.00\% & 0.00\% & 0.01\% & 0.12\% &
    0.04\% & 0.00\% & 0.71\% & 1.17\% &
    0.517 & 0.558 & 0.701 & 0.542 & 
    52.04\% & 56.69\% & 70.75\% & 53.24\% \\
    & Carlini et al.~\cite{lira} & 
    0.19\% & 1.01\% & 3.30\% & 0.24\% &
    1.48\% & 4.22\% & 12.11\% & 1.74\% & 
    0.529 & 0.616 & 0.782 & 0.543 & 
    52.81\% & 57.65\% & 72.88\% & 52.95\% \\ 
    & Zarifzadeh et al.~\cite{zarifzadeh2023low} & 
    0.67\% & 1.93\% & 5.41\% & 0.26\% &
    3.16\% & 6.57\% & 17.58\% & 2.04\% &
    0.523 & 0.592 & 0.797 & 0.570 &
    53.05\% & 58.17\% & 73.37\% & 55.04\% \\ 
    \bottomrule
\end{tabular}
}
\end{table*}

We first evaluate the membership risks of NODEs.
Our hypothesis is that these models
are vulnerable to existing membership inference attacks
as their performance on popular benchmarking tasks
is comparable to conventional feedfoward networks.
However, because their expressivity is limited to modeling a system of ODEs,
the risk could be lower than that of conventional networks.

\subsection{Threat Model}
\label{subsec:threat-model}

The attacker aims to determine whether a specific sample $z=(x, y)$ 
is included in the data used for training the target model $f$.
The adversary only has \emph{black-box} access to $f$.%
\footnote{Nasr~\textit{et al.}~\cite{nasr2019comprehensive} 
showed that a \emph{white-box} membership adversary
with full access to the model and its parameters, 
does not particularly perform better than black-box attacks.}
We assume that the attacker can train shadow models $f_s$
with known membership for all shadow training samples
and use them to infer membership from the model output. %
Because the attacker aims for $f_s$ to closely resemble the behaviors of $f$, 
we assume the attacker's shadow training dataset
and the original training data come from the same underlying data distribution 
and may overlap partially.
We also assume a worst-case black-box adversary
with full knowledge of $f$'s architecture and training configurations
such that they employ this knowledge to construct and train $f_s$.
When evaluating defenses in \S\ref{sec:eval},
we consider an \emph{adaptive} adversary
who knows the defense mechanisms deployed on $f$ 
and their hyper-parameters.

\subsection{Experimental Setup}
\label{sec:experimental-setup}

\topic{Datasets.}
We use four object recognition benchmarks:
FashionMNIST~\cite{xiao2017fashion},
CIFAR-10 and CIFAR-100~\cite{Krizhevsky09learningmultiple}, 
and Tiny-ImageNet~\cite{imagenet_cvpr09}.
These datasets are commonly employed 
in prior work~\cite{yeom, shokri, song2021scorebased, lira}. 
We employ them to evaluate and compare the membership risks of NODEs 
with those reported in prior studies on feedforward neural networks.
\smallskip

\topic{Models.}
For each dataset, we train models with architectures 
based on those studied by Oganesyan~\textit{et al.}~\cite{oganesyan2020stochasticity}.
They have a first few convolutional layer for downsampling, 
followed by three residual or ODE blocks 
(see Appendix for details).
Our baseline is a ResNet architecture (\textbf{ResNet-14})
with 3 residual blocks where each block within a group has the same size 
and each group has 16, 32, and 64 filters respectively,
widened by a factor of 2.
To construct our NODE models,
we substitute each group of residual blocks 
with an ODE block. %
\textbf{NODE} (ODENet-16-32-64) uses the same dimensions.
We also examine \textbf{ODENet-64}, which contains only one ODE block
after the downsampling layers---a typical structure used in prior work~\cite{node, dupont2019augmented}.
\smallskip

\topic{Metrics.}
Following best practices from prior work~\cite{lira}, 
we measure the membership risks by computing 
the true-positive rate at a low false-positive rate (\emph{TPR@0.1--1\% FPR}). 
We also compare \emph{AUC} and inference accuracy (or \emph{accuracy}).
Because this empirical membership risk is
highly connected to a model's generalization gap, 
we report both the best train and test accuracies of target models.

\subsection{Quantifying the Membership Risks}
\label{sec:quantifying}

We evaluate our models against six membership inference attacks outlined in \S\ref{subsec:mia}. 
For the attacks employing shadow models, 
we train 16 models on 16 different subsets of the original training dataset.
Following the strategy of prior work~\cite{lira},
we construct these 16 different shadow training datasets
such that each training sample is included in half of the sets 
and excluded from the other half.
We then choose one of them as the target model $f$ 
and the remaining 15 as shadow models $f_s$. 
We run 16-fold cross-validation by choosing a different target model
from the 16 models in each round, and we report the averaged metrics.
For attacks like those by Yeom~\textit{et al.},
which do not use shadow models. %

\topic{Results.}
Table~\ref{tbl:main-mia} summarizes our results.
Due to the page limit, 
we include the ODENet-64 results in Appendix and focus on ResNet-14 and NODE models.
We contrast the membership risks of NODEs and ResNet-14 models.
Overall, we find that \emph{NODEs are 2$\times$ less vulnerable 
to membership inference attacks}, compared to ResNet14 in most cases. 
In CIFAR-10, attacks except those by
Carlini~\textit{et al.} and Zarifzadeh~\textit{et al.}
show a TPR of $\sim$0.0\% at 0.1--1.0\% FPR.
Against the two effective attacks, 
NODEs reduce their effectiveness by 1.3--4$\times$ than ResNet-14.
On CIFAR-100, other attacks also show some effectiveness,
with TPRs of up to 1.5\%, but the LiRA and RMIA attacks
remain the most effective.
Both attacks are 1.6--2.6$\times$ less effective 
on NODEs at the same low-FPR regime.
On FashionMNIST and Tiny-ImageNet,
we also observe a reduction in membership inference success.
However, as the baseline attacks 
were already less effective on these datasets,
the relative reduction is limited at most 2$\times$.
We find the reduced effectiveness in AUC and inference accuracy.
Note that our results do not imply that 
NODEs are free from membership risks.
The attacks are still effective,
with TPRs ranging from 1.1\% to 10.6\% at low FPRs
in CIFAR-10 and 100.
\smallskip

\topic{Generalization gap}
is a known factor contributing to membership risks.
Prior work~\cite{yeom, shokri} shows that 
risks increase when a model \emph{overfits}%
---it shows a large disparity between the training and testing accuracy.
Hence, we analyze whether the reduced membership risks
observed in NODEs are because of lower overfitting.
Table~\ref{tbl:main-mia-overfit} summarizes the gaps.

\begin{table}[ht]
\centering
\caption{%
    \textbf{Generalization gap},
    measured for both ResNet-14 and NODE models 
    across four benchmark datasets.
}
\label{tbl:main-mia-overfit}
\vspace{-0.6em}
\adjustbox{max width=\linewidth}{
\begin{tabular}{@{}r|rr|rr|rr|rr@{}}
    \toprule
    \textbf{Task} & \multicolumn{2}{c|}{\textbf{F-M}} & \multicolumn{2}{c|}{\textbf{C-10}} & \multicolumn{2}{c|}{\textbf{C-100}} & \multicolumn{2}{c}{\textbf{T-I}} \\ \midrule
    \textbf{Model} & \textbf{R-14} & \textbf{NODE} & \textbf{R-14} & \textbf{NODE} & \textbf{R-14} & \textbf{NODE} & \textbf{R-14} & \textbf{NODE}\\ \midrule \midrule
    \textbf{Train} & 100.0\% & 95.5\% & 98.1\% & 94.3\% & 95.0\% & 82.1\% & 64.7\% & 55.9\% \\
    \textbf{Test} & 88.4\% & 90.0\% & 86.0\% & 84.4\% & 54.3\% & 52.2\% & 42.4\% & 40.5\% \\ \midrule
    \textbf{$\Delta$} & 11.6\% & 5.5\% & 12.1\% & 9.9\% & 40.7\% & 29.9\% & 22.3\% & 15.4\% \\ \bottomrule
\end{tabular}
}
\end{table}

\begin{figure}[h]
    \centering
    \vspace{-0.8em}
    \includegraphics[width=0.9\linewidth]{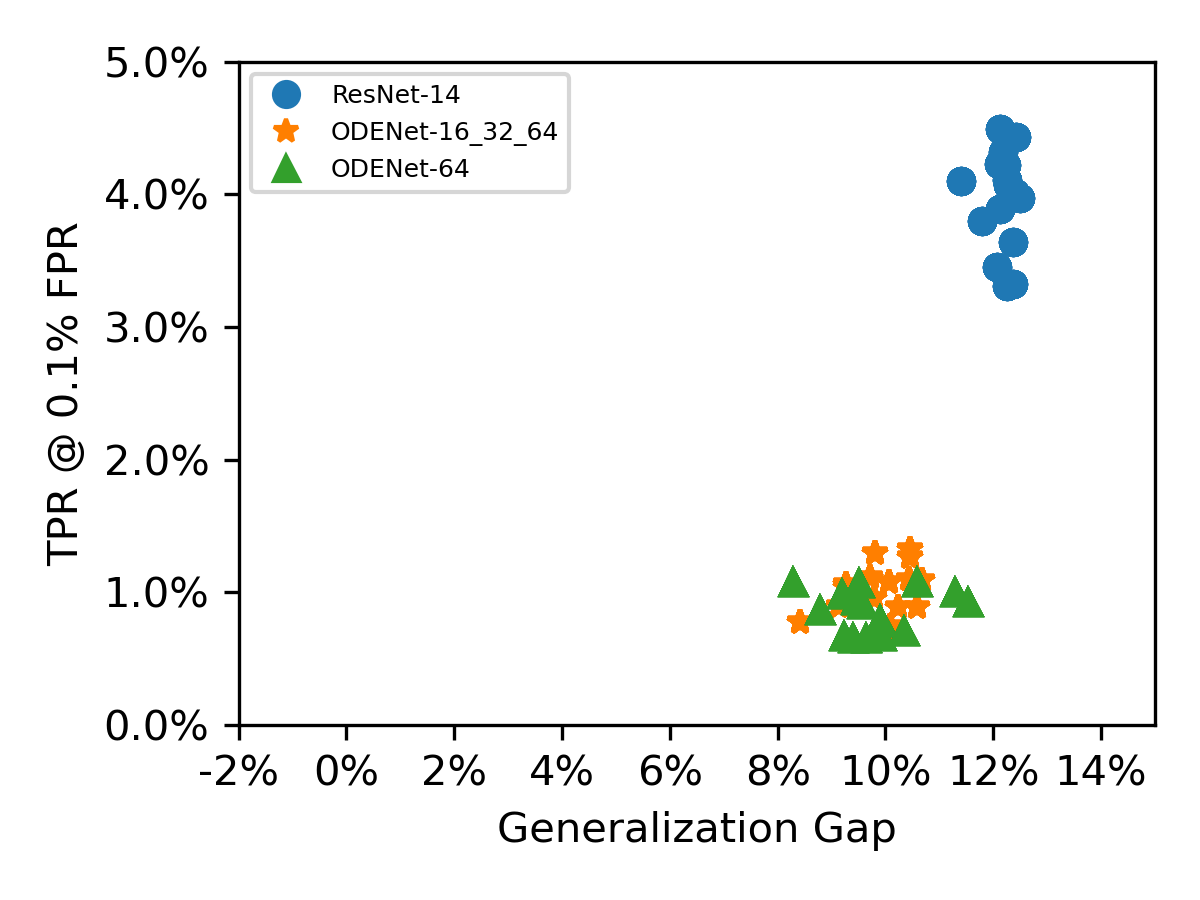}
    \vspace{-1.6em}
    \caption{\textbf{Membership risks and overfitting} in CIFAR-10.}
    \label{fig:generalization_gap}
\end{figure}

Figure~\ref{fig:generalization_gap} analyzes the interaction 
between membership risks and the generalization gap in CIFAR-10.
We plot the attack success and the generalization gap
observed for each of the 16 target models,
focusing on LiRA~\cite{lira}, 
where we observe the largest reduction.
We first find that overfitting correlates with
the success of membership inference attacks, 
as shown in Table~\ref{tbl:main-mia-overfit}.
In ResNets, the generalization gap is $>$12\%,
and the TPR @ 0.1\% FPR is 3--5\%.
In contrast, NODE models (ODENet-16\_32\_64 and -64), 
with 0.5--1.5\% TPR @ 0.1\% FPR, 
show the generalization gap in 8--12\%.
Please refer to Appendix for results in CIFAR-100.

\subsection{Characterizing the Membership Risks}
\label{subsec:characterization}

We now analyze factors that 
attribute to the membership inference risks.
For this analysis, we use LiRA,
where NODEs demonstrate the largest reduction in their risk.
We examine block-level configurations,
such as the choice of ODE solvers, and model-level factors,  
including the number of ODE blocks and 
the type of differential equations used to model a system.
Figure~\ref{fig:node-risk-ablation} summarizes our results.

\smallskip

\begin{figure*}[t]
    \centering
    \includegraphics[width=.32\linewidth]{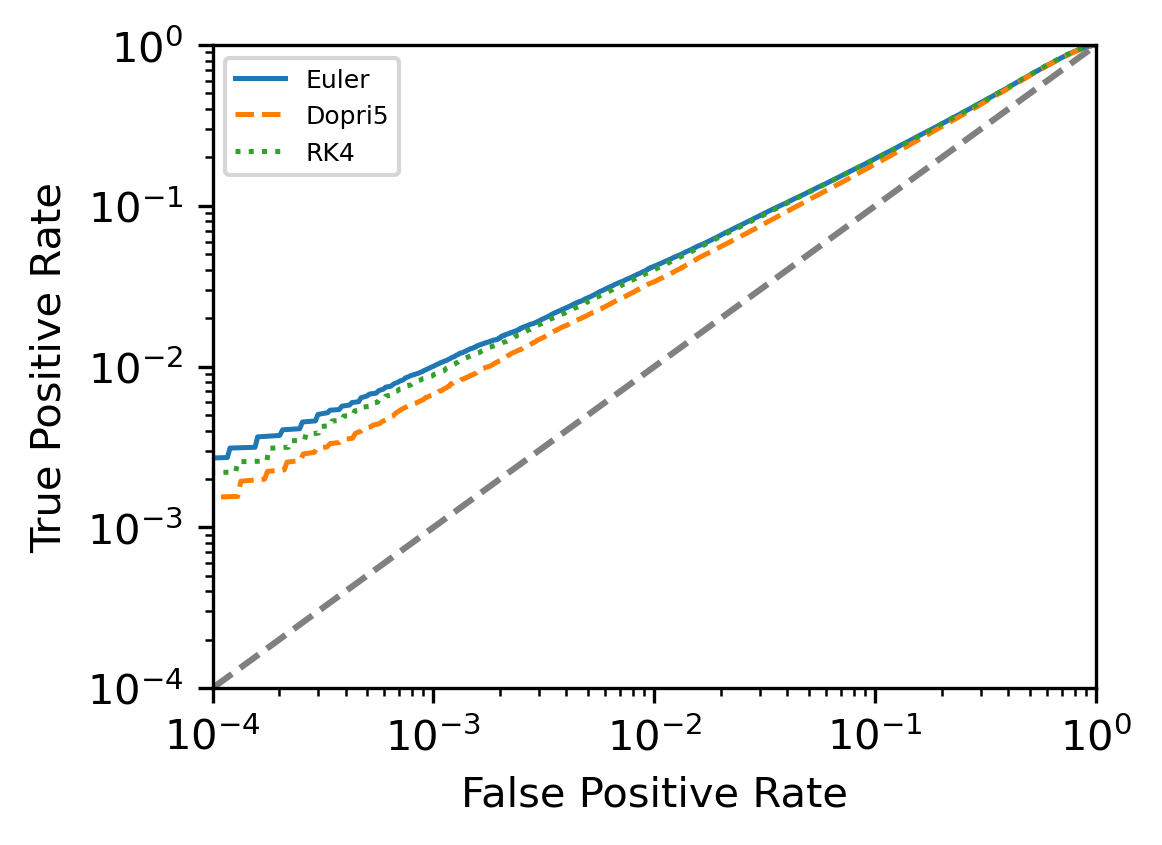}
    \hfill
    \includegraphics[width=.32\linewidth]{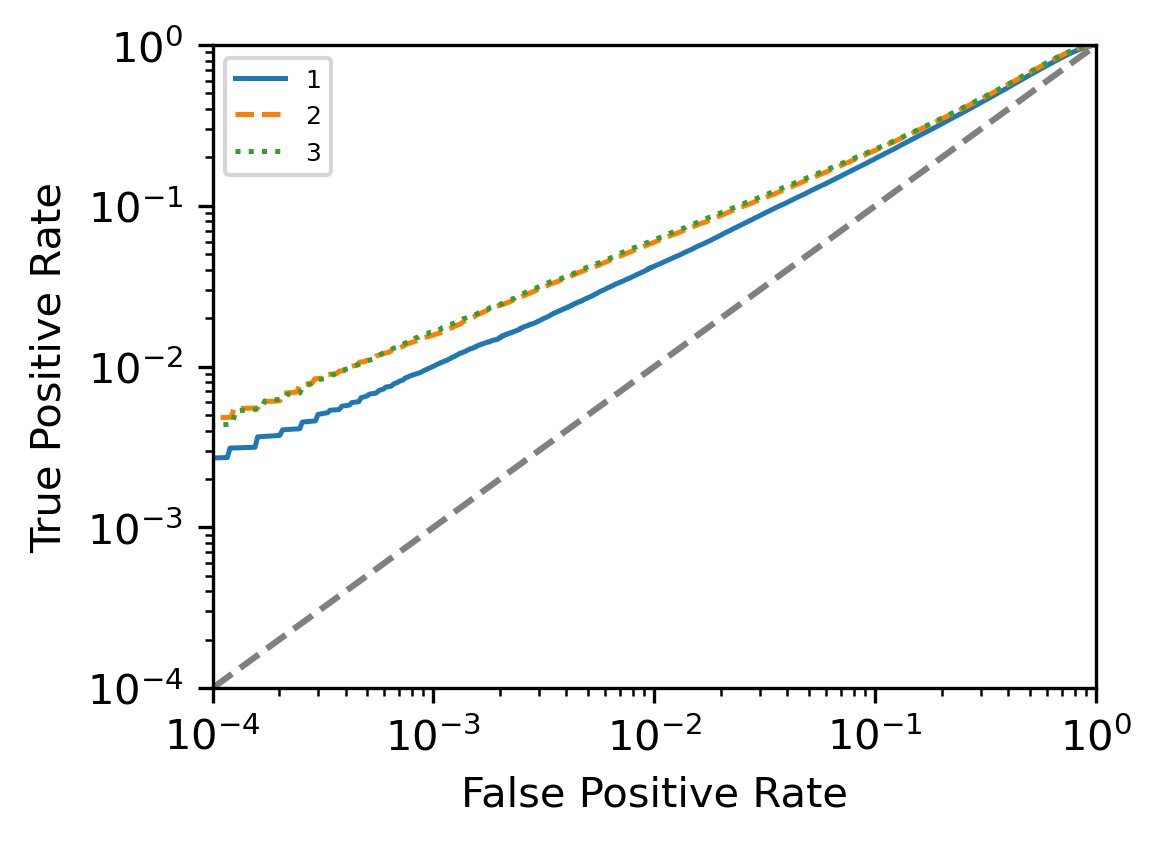}
    \hfill
    \includegraphics[width=.32\linewidth]{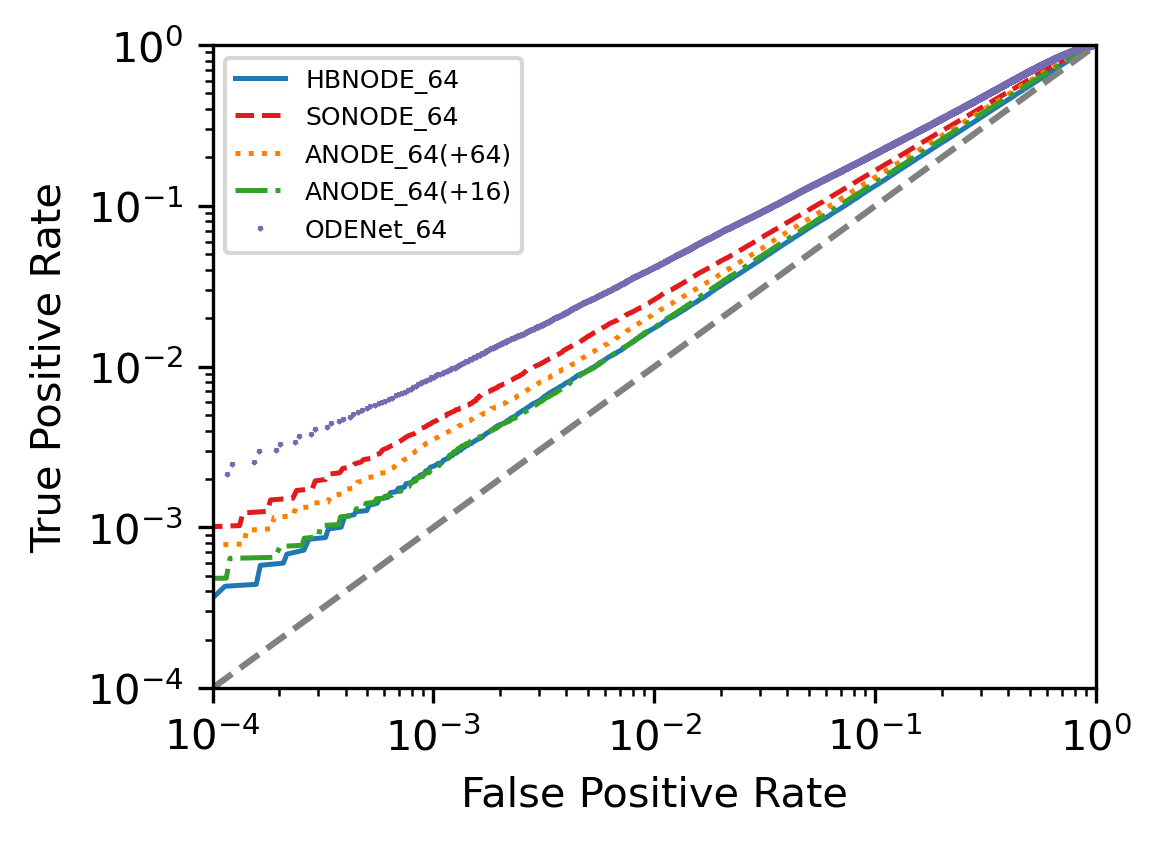}
    \vspace{-0.8em}
    \caption{%
        \textbf{Characterization of membership risks in NODEs with different configurations.}
        We vary the choice of ODE solvers (left), the number of ODE blocks (middle), 
        and the equations used for modeling a system (right). 
        We show the results in CIFAR-10 against LiRA~\cite{lira}, 
        the strongest membership inference attack.
        Please refer to Appendix for details and results in CIFAR-100.
    }
    \label{fig:node-risk-ablation}
\end{figure*}

\topic{Impact of a solver.}
In all our experiments, 
we use Euler's method as the default solver;
but, other popular choices exist in the literature.
We therefore analyze how the choice of solver 
impacts the membership risks of NODEs.
We train and test the same NODEs with two other solvers:
a fixed-step solver, Runge--Kutta (RK4) %
and
an adaptive-step solver, Dormand--Prince (Dopri5).%

The leftmost figure %
illustrates the success rate of LiRA on NODEs 
trained and tested with three different solvers.
Overall, we observe a marginal difference in attack success across these models.
In both CIFAR-10 and CIFAR-100 
(shown in Appendix),
the TPR remains 0.7--1.0\% and 3.0--3.7\% %
@ 0.1\% FPR, respectively.
Upon closer examination,
Euler and RK4 yield similar risks,
while Dopri5 shows 0.3\% lower TPR at the same FPR.
We attribute this to Dopri5's use of adaptive step-sizes.
During the training of a NODE model with Dopri5, 
it stops processing of an input 
when the error rate becomes sufficiently small%
--within the predefined tolerance value.
It acts as an early-stopping mechanism~\cite{songattack}, 
preventing the model from becoming overconfident on specific training samples.
RK4, by modeling data with higher-order polynomials,
may achieve an increased training (and testing) accuracy.
However, the ability to perform complex modeling
can introduce a risk of overfitting to specific training samples,
which can also increase membership inference risks.
In CIFAR-100, a more complex task than CIFAR-10,
we observe 0.3\% and 1.0\% higher TPR at 0.1\% and 1\% FPR, respectively,
compared to the model using Euler.
\smallskip

\topic{Impact of the number of ODE blocks.}
Our NODEs replace a group of residual blocks with an ODE block,
a configuration chosen for its comparable accuracy to ResNet models.
One can also increase the number of ODE blocks further.
We hypothesize that, because additional ODE blocks enhance
the model's ability to represent complex systems,
this may lead to an increase in membership risks.

The middle figure compares the membership risks of NODEs 
as the number of ODE blocks increases from 1 to 3.
We find that adding more ODE blocks 
leads to higher membership risks,
likely due to the increased number of parameters,
which raises the likelihood of \emph{overfitting}.
However, this increase does not yield
significant performance gains,
suggesting that simply adding more ODE blocks 
is not an effective strategy for improving the performance.
\smallskip

\topic{Non-stochastic NODE variants.}
In \S\ref{subsec:neural-odes}, we review NODE variants 
designed to enhance expressivity and stability, 
ultimately improving generalization.
We examine whether these benefits 
come with a reduction in membership risks. 
We evaluate ANODE~\cite{dupont2019augmented}, SONODE~\cite{norcliffe2020second} 
and HBNODE~\cite{xia2021heavy}. 
We select ODENet-64 as our baseline model
because it is composed of a single ODE block,
making it easy to extend to NODE variants.
We implement SONODE and HBNODE using an ODE block 
identical to that in ODENet-64.
We augment ODENet-64 with an additional 16 and 64 dimensions (ANODE).
These variants are expected to enhance the model's generalization capabilities, 
allowing us to assess the impact of better system modeling on membership risks.

The rightmost figure compares the membership risks of 
these models with those of the baseline ODENet-64 model.
We first observe that NODE variants, except for the HBNODE model,
achieves higher test accuracy (84.6--86.5\%) than the baseline (84.4\%).
SONODE performs the best with 86.5\%,
followed by ANODE (+64) at 85.1\% and ANODE (+16) at 84.6\%.
But HBNODE achieves 82.1\%.
Interestingly, we find that 
\emph{all the NODE variants are less vulnerable} to %
LiRA.
SONODE, ANODE (+64), ANODE (+16), and HBNODE
achieve TPRs of 0.44\%, 0.35\%, 0.24\%, and 0.24\% at 0.1\% FPR, respectively,
compared to ODENet-64 with a TPR of 0.86\% at 0.1\% FPR.
We attribute this to the reduced overfitting shown by NODE variants.
The differences between training and testing accuracy are
5.9\%, 4.0\%, 3.0\%, and 2.8\% for
SONODE, ANODE (+64), ANODE (+16), and HBNODE, respectively,
compared to 9.7\% gap in ODENet-64.
However, for HBNODE, there is a possibility that the lower risk 
is attributed to the lack of generalization.
Modeling NODEs with momentum (HBNODE) enhances the stability, 
but does not increase a model's capacity.
In our evaluation, the training accuracy of HBNODE is 84.9\%,
2.8--7\% lower than that of the other variants.

\section{NSDEs Are DP-Learners}
\label{sec:our-defense}

In the previous section, 
we demonstrated that NODEs and their non-stochastic variants 
are still vulnerable to membership inference attacks.
We now shift our focus to the stochastic variants of NODEs, 
particularly neural stochastic differential equations 
(NSDEs)~\cite{liu2020how}.
In NSDEs, which contain ``drift'' and ``diffusion'' terms, 
stochasticity arises from the \emph{randomness in the diffusion term},
typically modeled as white or colored Gaussian noise.
We hypothesize that this diffusion term leads NSDEs to act as a differentially-private (DP) mechanism~\cite{dwork2006differential}, preserving the privacy of its inputs along its evolution.
Here we first formally show that NSDEs are DP-learners.
Because of the diffusion term, 
those models learn ``differentially-private'' systems of ODEs, 
achieving privacy without the need of additional mechanisms,  
such as DP-SGD~\cite{dpsgd}, which %
alters the training process.
Next in \S\ref{sec:eval},
we empirically assess the membership risks of NSDEs
in comparison to non-stochastic NODE variants
and ``private" feedforward models trained with DP-SGD.

\subsection{Theoretical Analysis}
\label{subsec:theoretical-analysis}

NSDEs model the dynamics of system states as:
\begin{align*}
    d~h = ~f(~h,t)\,dt + ~G(~h,t)\,d~B_t,
\end{align*}
in terms of a drift function $~f:\Omega \to \mathbb{R}^n$ 
and a diffusion term $~G:\Omega\to\mathbb{R}^{n\times n}$, where the explicit dependence on parameters $\Theta$ has been dropped for notational convenience.
Our privacy bounds will make use of the 
\textit{sensitivity} of a function $~g:\Omega\to\mathbb{R}^n$ (c.f. \cite{dwork2006differential}):
\begin{align*}
    S_{~g} = \max_{\nn{~h-~h'}\leq 1}\nn{~g(~h)-~g(~h')}.
\end{align*}
Here, we adopt the following natural assumptions
from the work by Liu~\textit{et al.}~\cite{liu2020how},
which guarantees the well-posedness of the SDE formulation under consideration.
\begin{assumption}[Sublinear Growth]
    $~f$ and $~G$ grow at most linearly, i.e., there exists a constant $c>0$ such that $\nn{~f(~h,t)}+\nn{~G(~h,t)} \leq c(1+\nn{~h})$ for all $~h\in\Omega$ and $t\geq 0$. 
\end{assumption}
\begin{assumption}[Lipschitz Continuity]
    $~f$ and $~G$ are $L$-Lipschitz, i.e., there exists a constant $L>0$ 
    such that $\nn{~f(~h,t)-~f(~h',t)} + \nn{~G(~h,t)-~G(~h',t)}\allowbreak 
    \leq L\nn{~h-~h'}$ for all $~h,~h'\in\Omega$ and $t\geq 0$. 
\end{assumption}
With these conditions in place, the following theorem establishes 
the differential privacy of the SDE mechanism.

\begin{theorem}[SDEs Are Differentially-Private]
\label{thm:SDEDP}
    Suppose $\sigma$, $T>0$ and $~f(~h,t):\Omega\to\mathbb{R}^n$ 
    is Lipschitz-continuous with constant $L_{~f}$. 
    Consider the well-posed SDE:
    \[d~h = ~f(~h,t)\,dt + \frac{\sigma}{\sqrt{T}}\,d~B_t.\]
    For any privacy budget $0<\varepsilon<1$ and failure probability $\delta>0$,
    the mechanism $~M(~h) = ~h(T)$, defined by the flow map of the SDE at time $T$,
    is $(\varepsilon, \delta)$-differentially private, provided that
    $\sigma \geq \sqrt{2\ln(1.25/\delta)}(TL_{~f}/\varepsilon)$%
    \footnote{Note that $\varepsilon$ and $\delta$ here
    are different from those used in \S\ref{subsec:mia-defenses}.}.
\end{theorem}

\begin{proof}
    Recall that the Wiener increment $d~B_t$ satisfies 
    $\int_s^t d~B_\tau = ~B_t-~B_s\sim\mathcal{N}(~0,(t-s)~I)$. 
    It follows that the solution to the hypothesized SDE can be formally expressed as:
    \[~h(T) = \int_0^T~f(~h,t)\,dt + \frac{\sigma}{\sqrt{T}}\int_0^T d~B_t := ~F(~h) + ~B_T,\]
    where $~B_T\sim\mathcal{N}\lr{~0, \sigma^2~I}$ represents additive Gaussian noise
    and $~F$ denotes deterministic part of the SDE evolution.
    Proposition~\ref{prop:GaussianDP} in Appendix shows that 
    this mechanism obeys $(\varepsilon,\delta)$-differential privacy 
    whenever $\varepsilon<1$ and $\sigma\geq\sqrt{2\ln(1.25/\delta)}(S_{~F}/\varepsilon)$.
    The claimed bound on the variance $\sigma$ follows 
    since the sensitivity $S_{~F}$ is bounded above due to Lipschitz continuity in $~f$:
    \begin{align*}
        S_{~F} &= \max_{\nn{~h-~h'}\leq 1}\nn{~F(~h)-~F(~h')} \\
        &\leq \max_{\nn{~h-~h'}\leq 1}\int_0^T\nn{~f(~h,t)-~f(~h',t)}dt \\ 
        &\leq \max_{\nn{~h-~h'}\leq 1}\int_0^T L_{~f}\nn{~h-~h'}dt \leq TL_{~f}. \qedhere
    \end{align*}
\end{proof}

\begin{remark}
    Importantly, Theorem~\ref{thm:SDEDP}, while stated continuous time,
    also holds in discrete time---for instance, with the Euler--Maruyama integration method%
    ---as long as the function $~F(~h) = \int_0^T ~f(~h,t)\,dt + ~\epsilon(~h)$ appropriately accounts for the quadrature error introduced by this discretization.
    Assuming Lipschitz continuity in the error term $~\epsilon$, 
    the claimed variance bound holds with $TL_{~f}$ replaced by $TL_{~f} + L_{~\varepsilon}\!\geq\!S_{~F}$.
    While this quadrature error 
    may affect solution accuracy, it does not affect the privacy guarantee, 
    since each Euler sub-step adds a statistically independent fraction of 
    the total required Gaussian noise, and therefore the total amount of noise added is equivalent even if the exact and approximate Gaussian contributions differ.
    Notably, only the approximate output $~h(T)$ at the final time is private and appropriate for release in this case.
\end{remark}

Theorem~\ref{thm:SDEDP} provides an additive noise threshold which guarantees differential privacy in the NSDE forward pass.  Although Theorem~\ref{thm:SDEDP} is stated for a constant diffusion term 
$~G(~h,t) = (\sigma/\sqrt{T})~I$,
it continues to hold for a general Lipschitz-continuous,
at most linear $~G(~h,t)$, provided the Gaussian random variable $\int_0^T ~G(~h,t)\,dt$ satisfies the stated variance bound for each $~h$.
Moreover, since computing gradients of the NSDE model 
with respect to its parameters accesses the private input data 
only through the SDE output $~h(T)$, 
the following %
Corollary %
demonstrating that differential privacy extends to the entire NSDE training process.

\begin{corollary}[NSDE Training is Differentially Private]\label{cor:DPSDEtrain}
    Consider target privacy parameters $0\!<\!\varepsilon'\!<\!1$ and $\delta'\!>\!0$ 
    with a maximum iteration count of $K$.
    Then, NSDE training of the drift $~f$ in Theorem~\ref{thm:SDEDP} 
    will be $(\varepsilon', K\delta+\delta')$-differentially private 
    as long as the variance satisfies $\sigma\!\geq\! 4\sqrt{K\ln(1.25/\delta)\ln(1/\delta')}(TL/\varepsilon')$.
    Here, $\delta\!>\!0$ is arbitrary %
    and $L = \max_{~f} L_{~f}$ denotes the maximum Lipschitz constant 
    attainable by $~f$ during training.
\end{corollary}

\begin{proof}
    Note that, by the parallel composition principle,
    mini-batching into disjoint subsets at each iteration 
    does not increase the total privacy leakage, 
    as each individual's data in the state $~h$
    is accessed only once by the SDE mechanism.
    Similarly, the composition of the NSDE output 
    with a downstream classifier does not degrade 
    $(\varepsilon,\delta)$-differential privacy guarantees, 
    as these guarantees are immune to post-processing by construction.
    Consequently,
    the total privacy leakage in NSDE training is directly proportional to 
    the number of training iterations $K$, 
    corresponding to the number of accesses to the private input $~h$.
    The result follows from the Strong Composition 
    Theorem in Dwork~\textit{et al.}~\cite{dwork2014algorithmic} 
    [Theorem 3.20, Corollary 3.21], %
    in combination with the variance bound established in Theorem~\ref{thm:SDEDP}.
\end{proof}

\begin{remark}
    Note that determining the Lipschitz constant $L$ in Corollary~\ref{cor:DPSDEtrain} is most easily accomplished by restricting the NSDE architecture to guarantee a particular Lipschitz bound, see, e.g.,~\cite{wang2023direct}. Moreover, it should be noted that the privacy guarantees presented here are not necessarily tight in every case 
    and may be improved with more sophisticated privacy accounting mechanisms, 
    such as those in~\cite{dpsgd, dong2022gaussian}.
\end{remark}

Corollary~\ref{cor:DPSDEtrain} guarantees the privacy of NSDE training provided that the diffusion term is constructed to add enough stochastic noise along NSDE solution trajectories.  This has the remarkable advantage of privatizing learning without requiring direct intervention in the training process: the network dynamics are already private, therefore no further privacy-preserving modifications are required during the backward pass. 
In contrast, DP-SGD must use inexact gradient information 
in order to retain privacy, which can potentially hinder the training process.

\begin{table*}[t]
\centering
\caption{%
    \textbf{Contrasting membership inference risks of NSDEs and NODEs.}
    We evaluate both models with six %
    attacks.
}
\label{tbl:main-sde}
\adjustbox{max width=\linewidth}{
\begin{tabular}{@{}ll | cccc | cccc | cccc | cccc@{}}
\toprule
 &  %
 & \multicolumn{4}{c|}{\textbf{TPR @ 0.1\% FPR}} & \multicolumn{4}{c|}{\textbf{TPR @ 1\% FPR}} & \multicolumn{4}{c|}{\textbf{AUC}} & \multicolumn{4}{c}{\textbf{Inference acc.}} \\ \midrule
\begin{tabular}{l}\textbf{Model}\end{tabular} & 
    \textbf{Method} &
    \textbf{F-M} & \textbf{C-10} & \textbf{C-100} & \textbf{T-I} & 
    \textbf{F-M} & \textbf{C-10} & \textbf{C-100} & \textbf{T-I} & 
    \textbf{F-M} & \textbf{C-10} & \textbf{C-100} & \textbf{T-I} & 
    \textbf{F-M} & \textbf{C-10} & \textbf{C-100} & \textbf{T-I} \\ \midrule \midrule
\multirow{6}{*}{\begin{tabular}{l}NODE\end{tabular}} & 
    Yeom et al.~\cite{yeom} & 
    0.00\% & 0.00\% & 0.01\% & 0.09\% &
    0.06\% & 0.00\% & 0.71\% & 1.12\% &
    0.518 & 0.558 & 0.701 & 0.645 & 
    52.55\% & 56.59\% & 70.75\% & 61.23\% \\
    & Shokri et al.~\cite{shokri} & 
    0.09\% & 0.04\% & 0.04\% & 0.10\% &
    0.90\% & 0.43\% & 0.37\% & 1.01\% &
    0.493 & 0.448 & 0.398 & 0.502 & 
    50.15\% & 50.01\% & 50.00\% & 50.21\% \\
    & Song and Mittal~\cite{song2021scorebased} & 
    0.00\% & 0.00\% & 0.15\% & 0.12\% &
    0.53\% & 0.00\% & 1.66\% & 1.19\% &
    0.507 & 0.529 & 0.498 & 0.500 &
    50.98\% & 54.22\% & 52.48\% & 50.45\% \\
    & Watson et al.~\cite{watson2022on} & 
    0.00\% & 0.00\% & 0.01\% & 0.12\% &
    0.04\% & 0.00\% & 0.71\% & 1.17\% &
    0.517 & 0.558 & 0.701 & 0.542 & 
    52.04\% & 56.69\% & 70.75\% & 53.24\% \\
    & Carlini et al.~\cite{lira} & 
    0.19\% & 1.01\% & 3.30\% & 0.24\% &
    1.48\% & 4.22\% & 12.11\% & 1.74\% & 
    0.529 & 0.616 & 0.782 & 0.543 & 
    52.81\% & 57.65\% & 72.88\% & 52.95\% \\ 
    & Zarifzadeh et al.~\cite{zarifzadeh2023low} & 
    0.67\% & 1.93\% & 5.41\% & 0.26\% &
    3.16\% & 6.57\% & 17.58\% & 2.04\% &
    0.523 & 0.592 & 0.797 & 0.570 &
    53.05\% & 58.17\% & 73.37\% & 55.04\% \\ 
    \midrule
\multirow{6}{*}{\begin{tabular}{l}SDENet ($\sigma\!=\!2$)\end{tabular}} & 
    Yeom et al.~\cite{yeom} & 
    0.00\% & 0.02\% & 0.05\% & 0.11\% & 
    0.06\% & 0.92\% & 1.09\% & 1.13\% & 
    0.526 & 0.525 & 0.588 & 0.557 &
    54.17\%& 52.44\% & 56.96\% & 54.72\% \\
    & Shokri et al.~\cite{shokri} & 
    0.09\% & 0.09\% & 0.08\% & 0.11\% & 
    0.94\% & 0.86\% & 0.85\% & 1.03\% & 
    0.498 & 0.484 & 0.473 & 0.502 & 
    50.14\% & 50.04\% & 50.00\% & 50.40\% \\
    & Song and Mittal~\cite{song2021scorebased} & 
    0.00\% & 0.00\% & 0.12\% & 0.12\% & 
    0.18\% & 0.20\% & 1.34\% & 1.05\% & 
    0.497 & 0.509 & 0.498 & 0.499 & 
    50.12\% & 50.95\% & 50.84\% & 50.12\% \\
    & Watson et al.~\cite{watson2022on} & 
    0.01\% & 0.02\% & 0.05\% & 0.11\% &
    0.18\% & 0.91\% & 1.09\% & 1.12\% &
    0.515 & 0.525 & 0.588 & 0.519 &
    51.02\% & 52.44\% & 52.44\% & 51.91\% \\
    & Carlini et al.~\cite{lira} & 
    0.10\% & 0.14\% & 0.43\% & 0.12\% & 
    0.97\% & 1.24\% & 3.00\% & 1.14\% &  
    0.500 & 0.625 & 0.612 & 0.509 & 
    50.01\% & 51.89\% & 57.64\% & 50.77\% \\ 
    & Zarifzadeh et al.~\cite{zarifzadeh2023low} & 
    0.27\% & 0.38\% & 0.83\% & 0.16\% & 
    1.82\% & 2.39\% & 4.60\% & 1.45\% & 
    0.527 & 0.554 & 0.646 & 0.535 & 
    51.12\% & 54.11\% & 59.25\% & 52.67\% \\
    \bottomrule
\end{tabular}
}
\end{table*}

\section{Empirical Evaluation}
\label{sec:eval}

We now empirically evaluate the effectiveness of NSDEs, 
as a DP-learner, in mitigating membership inference attacks.
We run our evaluation on the CIFAR-10 and CIFAR-100 datasets,
where membership inference attacks are shown to be most effective.
Our NSDE models (\textbf{SDENet}) adopt architectural configurations 
similar to ODENet, comprising three groups of ODEBlocks
with 16, 32, and 64 dimensional filters, respectively, 
widened by a factor of 2.
In SDENet, stochasticity is introduced by employing an SDE solver within the ODEBlocks, 
utilizing a stochastic fourth-order Runge-Kutta (RK4) integration method.
SDENet also incorporates a hyper-parameter ($\sigma$)
to control noise intensity,
influenced by factors such as time, step size, and the desired level of stochasticity.

\subsection{Effectiveness of NSDEs}
\label{subsec:effectiveness}

We conduct three comparisons:
(1) NSDEs with ResNet-14 to assess
whether stochasticity reduces the empirical membership risks.
(2) NSDEs with heuristic defenses discussed in \S\ref{sec:prelim},
which lack provable privacy guarantees.
(3) NSDEs with ``private" ResNet-14 trained with DP-SGD.
In (2), we consider four defenses studied in prior work%
~\cite{DBLP:journals/corr/abs-2006-05336, 10.1145/3319535.3363201, 
       10.1145/3422337.3447836, 10.1145/3580305.3599561}:
$\ell^1$/$\ell^2$-regularization, MMD-Mixup, and MemGuard.
For $\ell^1$/$\ell^2$-regularization, 
we set the penalty parameter \(\lambda\) to \(10^{-5}\).
We use the most robust version of MemGuard 
as described in Choquette~\textit{et al.}~\cite{choquette2021label},
which allows for arbitrary adjustments to the confidence vector 
while preserving the model's predicted label.
For MMD-Mixup, we adopt the implementation 
from the study by Li~\textit{et al.}~\cite{10.1145/3422337.3447836}, 
and set the penalty parameter \(\lambda\) to 5$\times10^{-4}$. 
In (3), when training with DP-SGD,
we set the privacy budget $\epsilon$ to 8,
a common choice for CIFAR-10 and 100 in prior work~\cite{10.1145/3580305.3599561}.
We use the popular Opacus library%
\footnote{Opacus: \url{https://github.com/pytorch/opacus/}}
to implement privacy accounting in NSDE training 
and to compute the total privacy leakage ($\varepsilon$).
\smallskip

\topic{(1) Mitigating membership inference attacks.}
We first compare the membership inference risks of 
NSDEs and NODEs against six existing attacks. 
Table~\ref{tbl:main-sde} shows our results.
Across four evaluation metrics,
NSDE reduce attack success by 4--10$\times$ compared to NODEs.
Particularly against the LiRA attack, 
NSDEs achieve a TPR of 0.14--1.24\% at both 0.1\% and 1\% FPRs,
while NODEs exhibit a higher TPR of 1.0--4.22\% 
under the same conditions (FPRs).
\smallskip

\begin{figure}[ht]
    \vspace{-0.8em}
    \centering
    \includegraphics[width=0.9\linewidth]{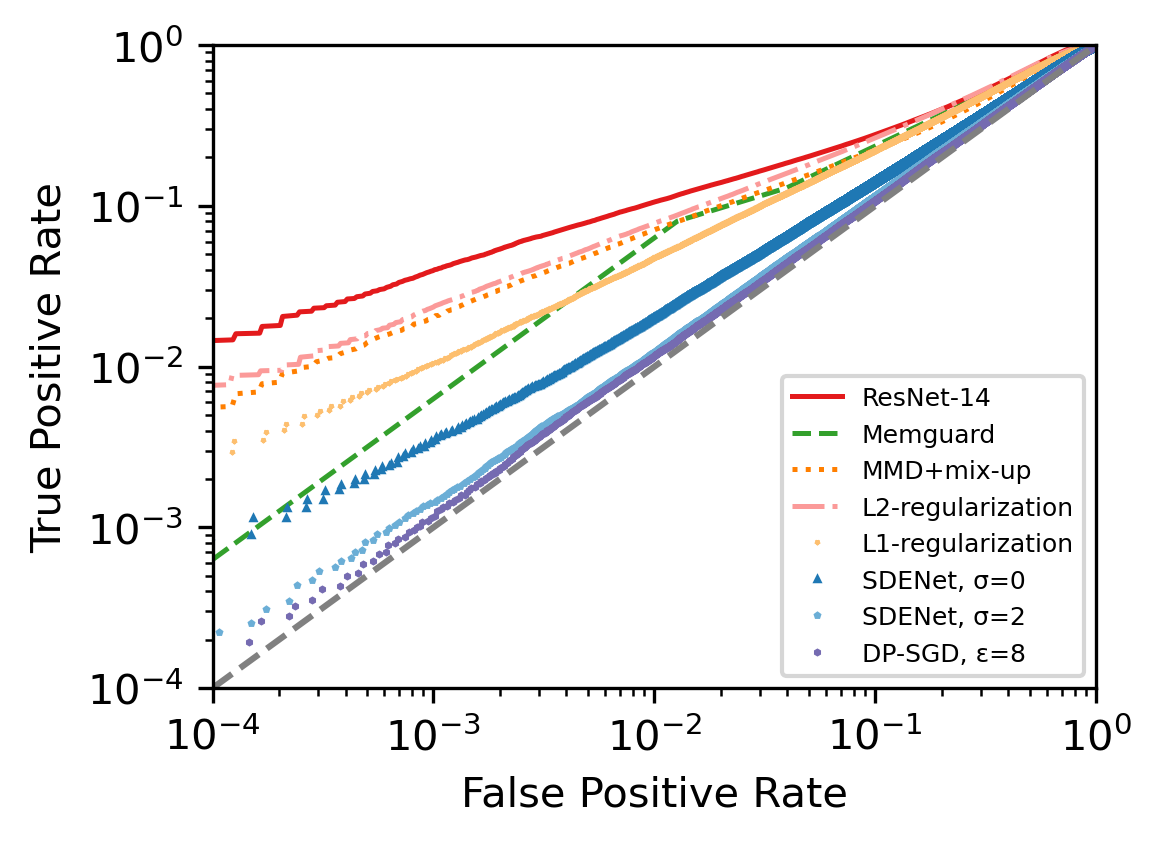}
    \vspace{-1.0em}
    \caption{%
        \textbf{Comparing the effectiveness of different defenses,}
        in CIFAR-10.
        We use LiRA to assess their membership risks. 
        All defenses, except for SDENets, are applied to ResNet-14.
    }
    \label{fig:nsde-vs-others}
\end{figure}

\topic{(2) Comparing with heuristic defenses.}
Figure~\ref{fig:nsde-vs-others} summarizes our results against LiRA,
the membership inference attack
where we observe the most significant reduction in attack success.
Overall, we first demonstrate that 
\emph{NSDEs effectively mitigate membership inference risks}.
Compared to ResNet-14, a feedforward model 
that yields 3.96--10.57\% %
TPR at both 0.1\% and 1\% FPR,
NSDEs with $\sigma\!=\!2$ achieve approximately 
two orders of magnitude lower TPRs (0.14--1.24\%). 
We also find that NSDEs significantly outperform 
heuristic defenses in reducing membership inference risks.
Compared to the four empirical defenses%
---$\ell^1$/$\ell^2$-regularization, MMD-Mixup and MemGuard%
---NSDEs with $\sigma\!=\!2$ achieves 2--8$\times$ lower TPRs at 0.1--1\% FPRs.
\smallskip

\topic{(3) Comparison with private models constructed via DP-SGD.}
NSDEs match the effectiveness of
ResNet-14 trained with DP-SGD at $\varepsilon\!=\!8$.
It is interesting to observe that
even the non-private variant of NSDEs (i.e., $\sigma\!=\!0$)
achieve lower TPRs than these defended models.
NSDEs with $\sigma\!=\!0$ achieve a test accuracy of 83.9\%, 
which, though slightly lower than ResNet-14's 86.0\%,
significantly narrows the generalization gap from 12.2\% to 5.2\%.
This narrowing leads to a markedly lower TPR at 0.1\% FPR: 
0.4\% for NSDEs versus 4.0\% for ResNet-14.
NSDEs with $\sigma\!=\!2$ attain a test accuracy of 81.9\% 
and nearly eliminate the generalization gap (-0.02\%),
further reducing the TPR to 0.1\% at 0.1\% FPR.
Compared to $\ell^1$ and $\ell^2$-regularization, 
NSDEs with $\sigma\!=\!2$ not only provide a modest reduction in test accuracy by 4.6\%--4.7\%, 
but also enhance defenses, achieving a TPR at 0.1\% FPR of 0.14\% 
versus 1.1\% for $\ell^1$ and 2.4\% for $\ell^2$).
While DP-SGD with $\varepsilon\!=\!8$ realizes a smallest generalization gap,
it significant compromises test accuracy,
dropping 7.5--9.5\% compared to NSDEs ($\sigma\!=\!0$ and $\sigma\!=\!2$).
Consequently, NSDEs demonstrates an improved overall performance 
in terms of both model training (higher test accuracy) 
and defense against membership inference attacks, 
making it a more balanced and practical approach.
Please refer to Appendix for further details.

\begin{figure*}[t]
    \centering
    \includegraphics[width=.242\linewidth]{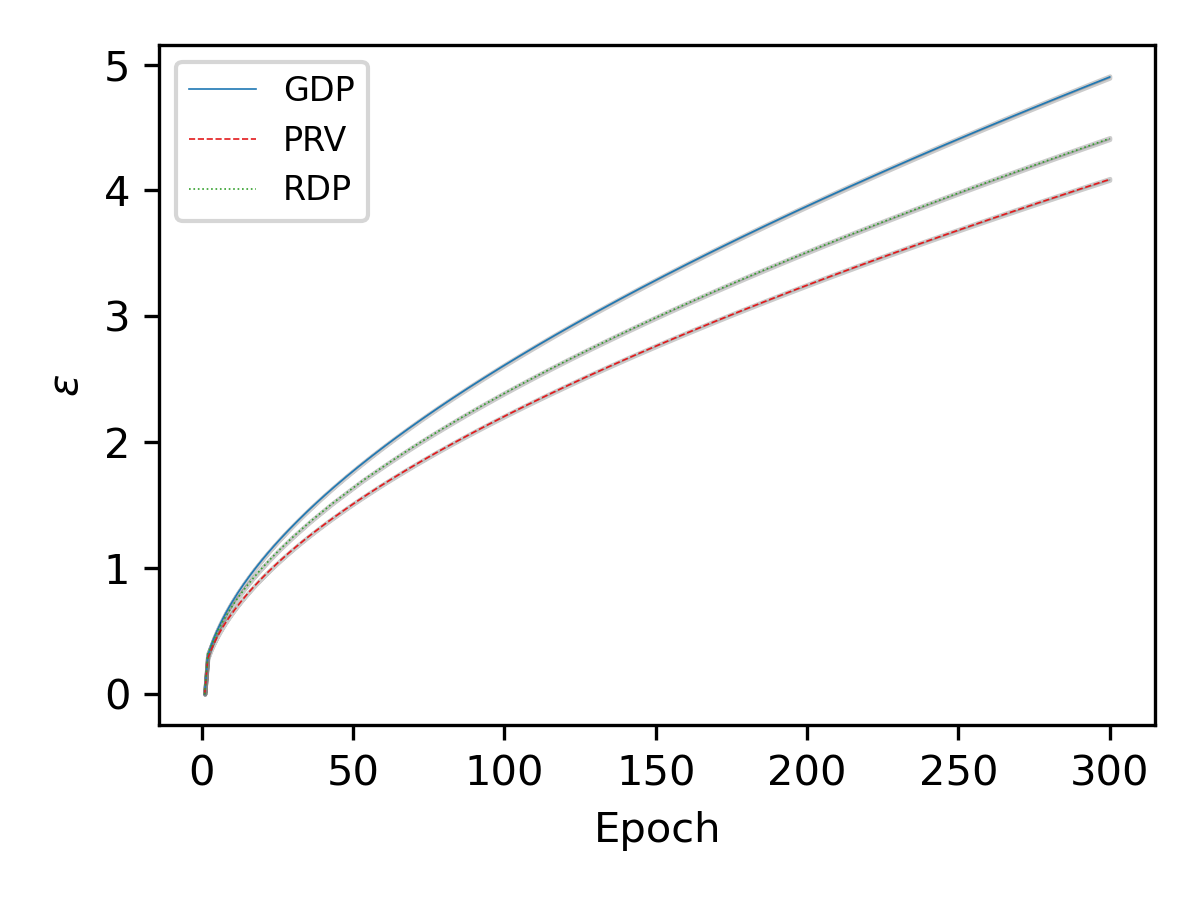}
    \hfill
    \includegraphics[width=.242\linewidth]{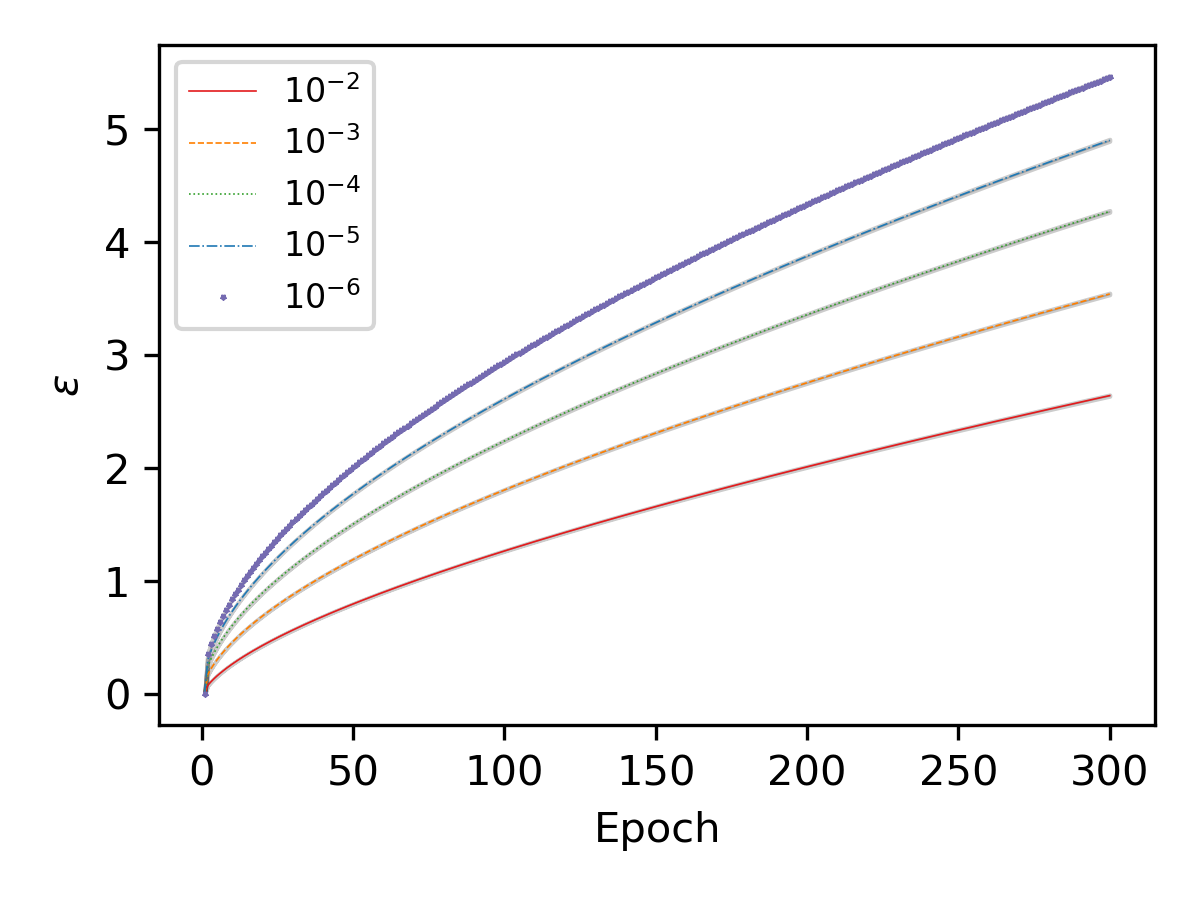}
    \hfill
    \includegraphics[width=.242\linewidth]{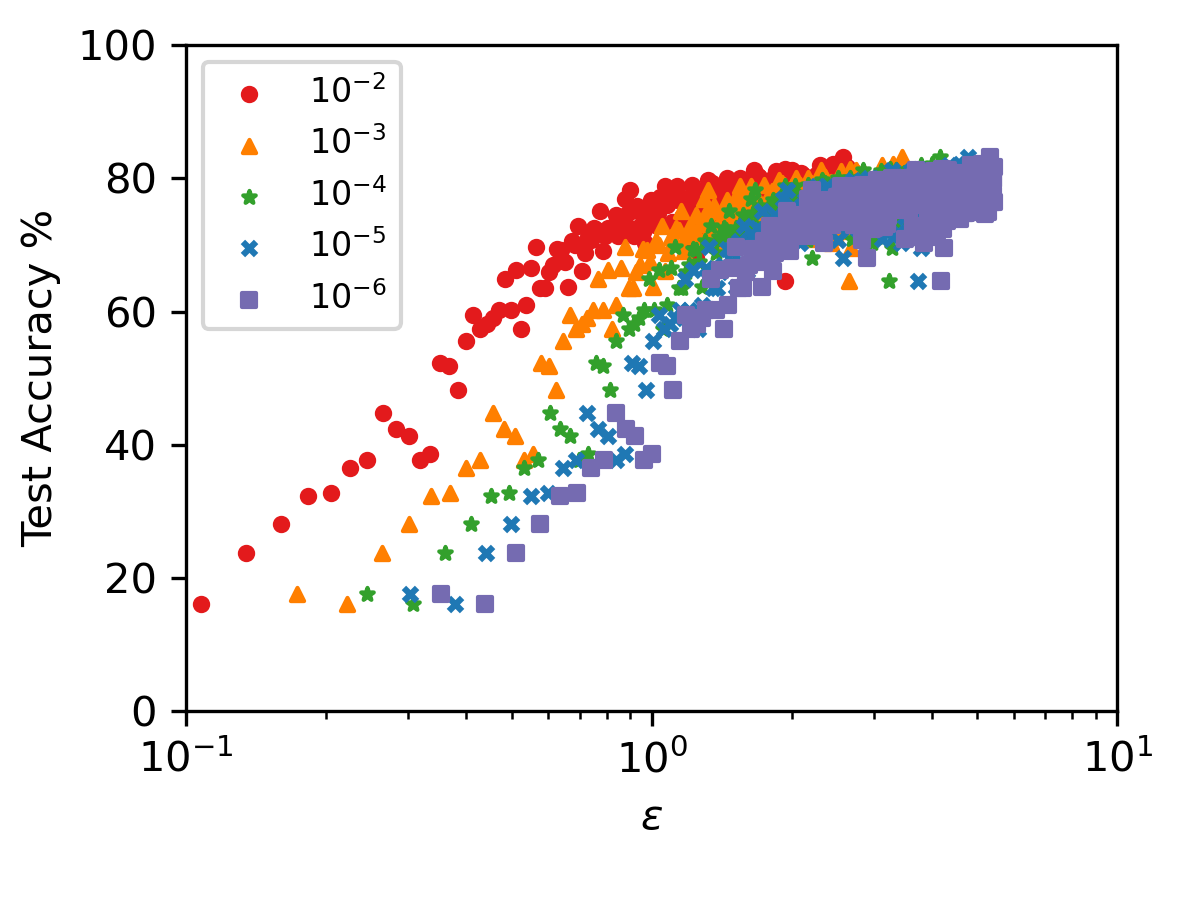}
    \hfill
    \includegraphics[width=.242\linewidth]{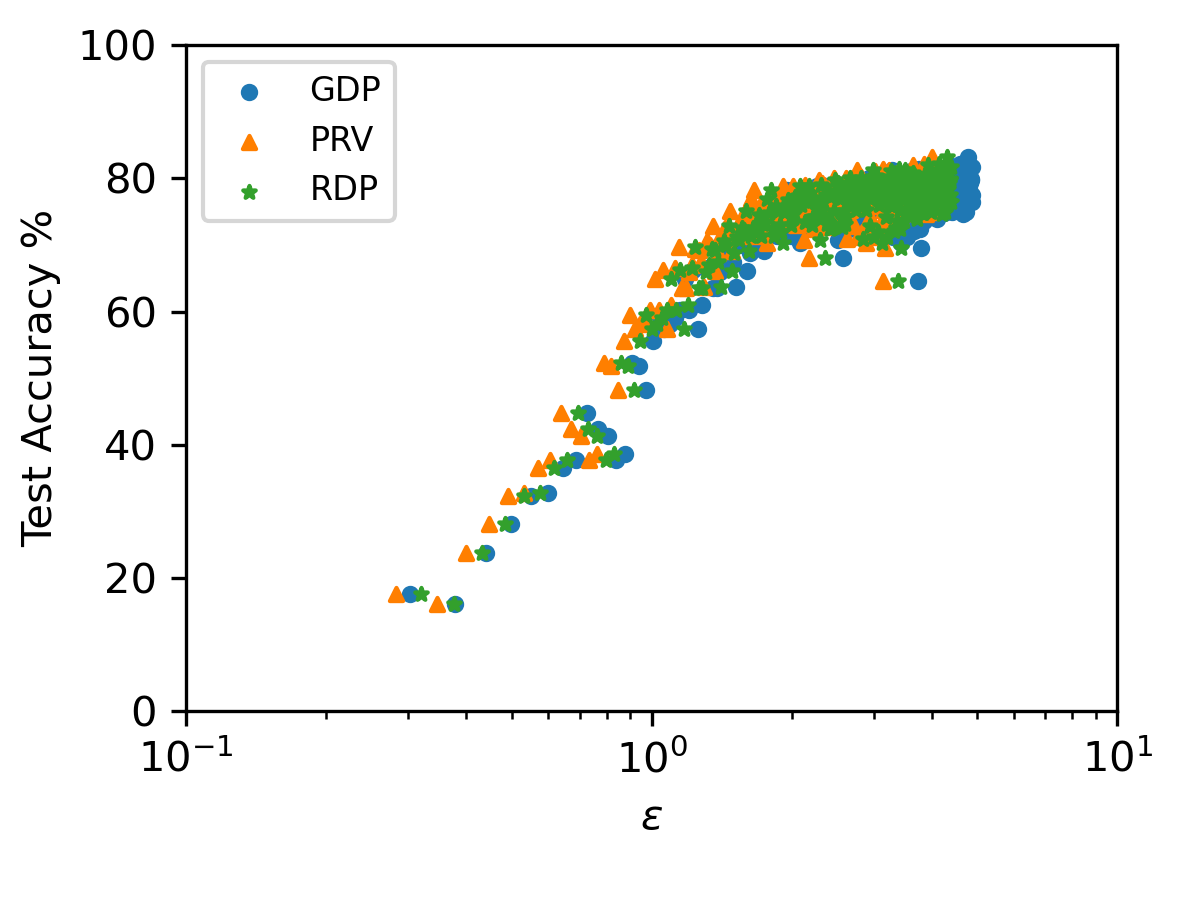}
    \vspace{-1.em}
    \caption{%
        \textbf{Impact of privacy hyper-parameters on NSDE training} in CIFAR-10.
        We analyze the total privacy leakage ($\varepsilon$) in the left two plots
        and the privacy-utility tradeoff in the right two plots,
        while varying key hyper-parameters.
    }
    \label{fig:nsde-privacy-analysis}
\end{figure*}

Moreover, to understand where the resilience of NSDEs against these attacks comes,
we compare the generalization gap in Table~\ref{tbl:main-sde-overfit}.
In NODEs, the gap is 9.9--29.89\%
for CIFAR-10 and CIFAR-100,
while in NSDEs, it is 0.0--5.9\%.
This result implies that 
\emph{the reduction in membership risks in NSDEs
is partly due to a significant reduction in overfitting}.

\begin{table}[ht]
\centering
\caption{%
    \textbf{Generalization gap},
    measured for both NODE models and NSDE models
    across four benchmark datasets.
}
\label{tbl:main-sde-overfit}
\vspace{-0.6em}
\adjustbox{max width=\linewidth}{
\begin{tabular}{@{}r|rr|rr|rr|rr@{}}
    \toprule
    \textbf{Task} & \multicolumn{2}{c|}{\textbf{F-M}} & \multicolumn{2}{c|}{\textbf{C-10}} & \multicolumn{2}{c|}{\textbf{C-100}} & \multicolumn{2}{c}{\textbf{T-I}} \\ \midrule
    \textbf{Model} & \textbf{NODE} & \textbf{NSDE} & \textbf{NODE} & \textbf{NSDE} & \textbf{NODE} & \textbf{NSDE} & \textbf{NODE} & \textbf{NSDE}\\ \midrule \midrule
    \textbf{Train} & 95.5\% & 96.2\% & 94.3\% & 81.9\% & 82.1\% & 56.5\% & 55.9\% & 39.9\%\\
    \textbf{Test} & 90.0\% & 89.2\% & 84.4\% & 81.9\% & 52.2\% & 50.6\% & 40.5\% & 38.1\%\\ \midrule
    \textbf{$\Delta$}& 5.5\% & 7.0\% & 9.9\% & 0.0\% & 29.9\% & 5.9\% & 15.4\% & 1.8\%\\ \bottomrule
\end{tabular}
}
\end{table}

\subsection{In-depth Analysis of Privacy of NSDEs}
\label{subsec:nsde-privacy-analysis}

We next analyze the privacy guarantee ($\varepsilon$) provided by NSDEs.
Because we theoretically show that NSDEs are DP-learners
(with the privacy guarantee equivalent to those of DP-SGD),
it is important to understand empirically 
whether their privacy interactions behave similarly to
observations from prior work~\cite{236254}.
We first compare the privacy leakage of NSDEs 
with the worst-case leakage bound established by DP-SGD.
We also examine the impact of privacy hyper-parameters---%
such as $\delta$ and the choice of privacy accounting algorithms---%
on the guarantee $\varepsilon$ and the model utility.
\smallskip

\begin{figure}[h]
\centering
\includegraphics[width=0.9\linewidth]{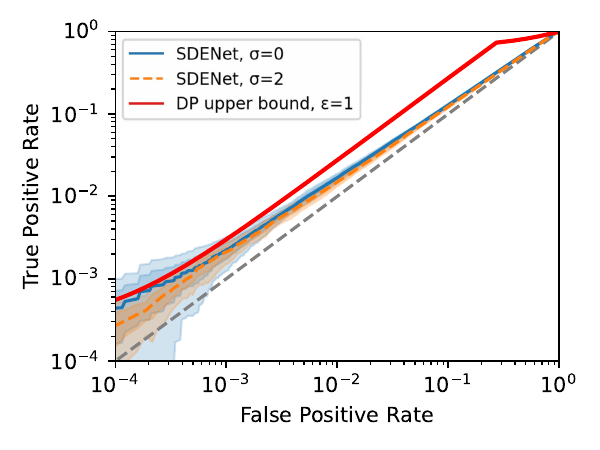}
\vspace{-1.0em}
\caption{%
    \textbf{Comparing the empirical membership risks of NSDEs 
    with the theoretical bound by DP-SGD.}
    We use LiRA to measure the membership risks 
    and the theoretical bound when $\varepsilon\!=\!1$
    is computed by the method presented in~\cite{lira}.
    NSDE with $\sigma\!=\!2$ is used to compare,
    and that with $\sigma\!=\!0$ is its non-private version.
    The shaded areas give the 68\% (darker) and 95\% (lighter) confidence intervals.}
\label{fig:compare-with-the-bound}
\vspace{-1.0em}
\end{figure}

\topic{Comparison to the theoretical DP-SGD bound.}
Figure~\ref{fig:compare-with-the-bound} compares 
the empirical success of LiRA against NSDEs with $\sigma\!=\!2$
to the theoretical upper-bound of DP-SGD (when $\varepsilon\!=\!1$).
Note that training NSDEs with $\sigma\!=\!2$ 
while targeting $\varepsilon\!=\!1$ 
completely destroys the model utility,
making any meaningful LiRA success unattainable.
To maintain utility, 
we target $\varepsilon\!=\!8$ 
in our experiments with NSDEs. %
This does not compromise the generality of our results,
because even against NSDEs with weaker privacy guarantees,
the empirical attack (LiRA) does not exceed the DP-SGD bound.
We first observe that the attack success of LiRA on NSDEs with $\sigma\!=\!2$
remains within the worst-case leakage bound established by DP-SGD.
However, when the diffusion mechanism is removed ($\sigma\!=\!0$),
LiRA's attack success often goes beyond the DP-SGD bound,
especially at lower FPR regions (e.g., at 0.1--1\% FPR),
highlighting the absence of formal privacy guarantees
in non-stochastic variants.
\smallskip

\topic{Privacy leakage $\varepsilon$ over training.}
The left two plots in Figure~\ref{fig:nsde-privacy-analysis}
illustrate the total privacy leakage ($\varepsilon$) during NSDE training.
As expected, both plots show that privacy leakage 
increases with training epochs.
However, the growth rate decreases over time---%
consistent with empirical findings from prior work
on training models with DP-SGD~\cite{dpsgd}.

The leftmost plot compares privacy leakage estimates
produced by three different accounting methods: 
Gaussian Differential Privacy (GDP), 
Renyi Differential Privacy (RDP), 
and Privacy loss Random Variables (PRVs). 
Among them, GDP consistently yield 
higher estimates of privacy leakage 
compared to the other two methods.
This is partly because GDP relies on Gaussian approximations,
which often deviate from the true privacy loss distribution.
In contrast, precise accounting methods like RDP or PRVs,
which enable tracking the exact or near-exact cumulative privacy loss,
provide tighter privacy bounds.
Moreover, because privacy loss accumulates over time, 
the discrepancy between GDP and the more precise methods 
progressively widens as the number of training epochs increases.

The second plot from the left 
shows the total privacy leakage computed 
using the GDP accountant across varying $\delta$ values,
from \(1 \times 10^{-2}\) to \(1 \times 10^{-6}\).
As $\delta$ increases, 
so does the total privacy leakage ($\varepsilon$), 
consistent with theoretical expectations.
This is because $\delta$ represents the probability of 
rare but catastrophic privacy failures 
(e.g., data leakage from physical breaches).
Allowing for a higher $\delta$ weakens the privacy guarantee
by tolerating a greater failure risk.
\smallskip

\begin{figure*}[t]
    \centering
    \includegraphics[width=.32\linewidth]{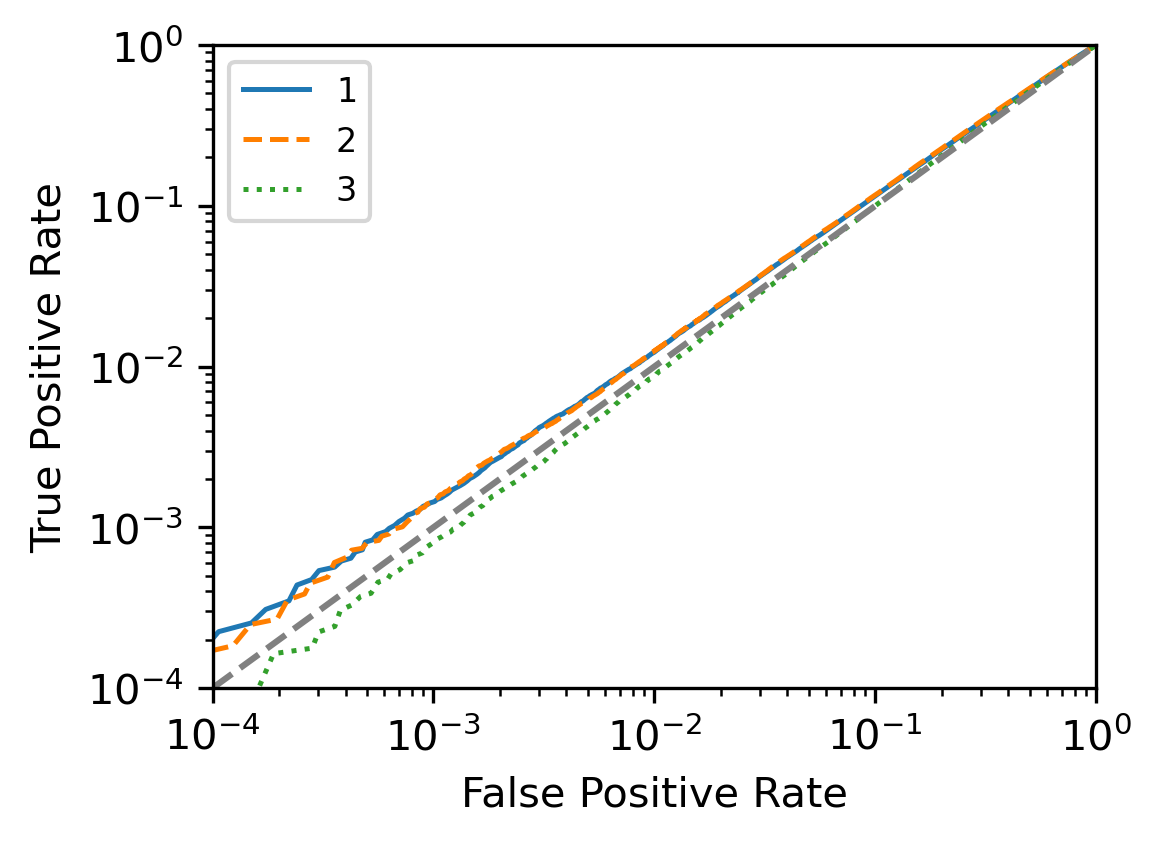}
    \includegraphics[width=.32\linewidth]{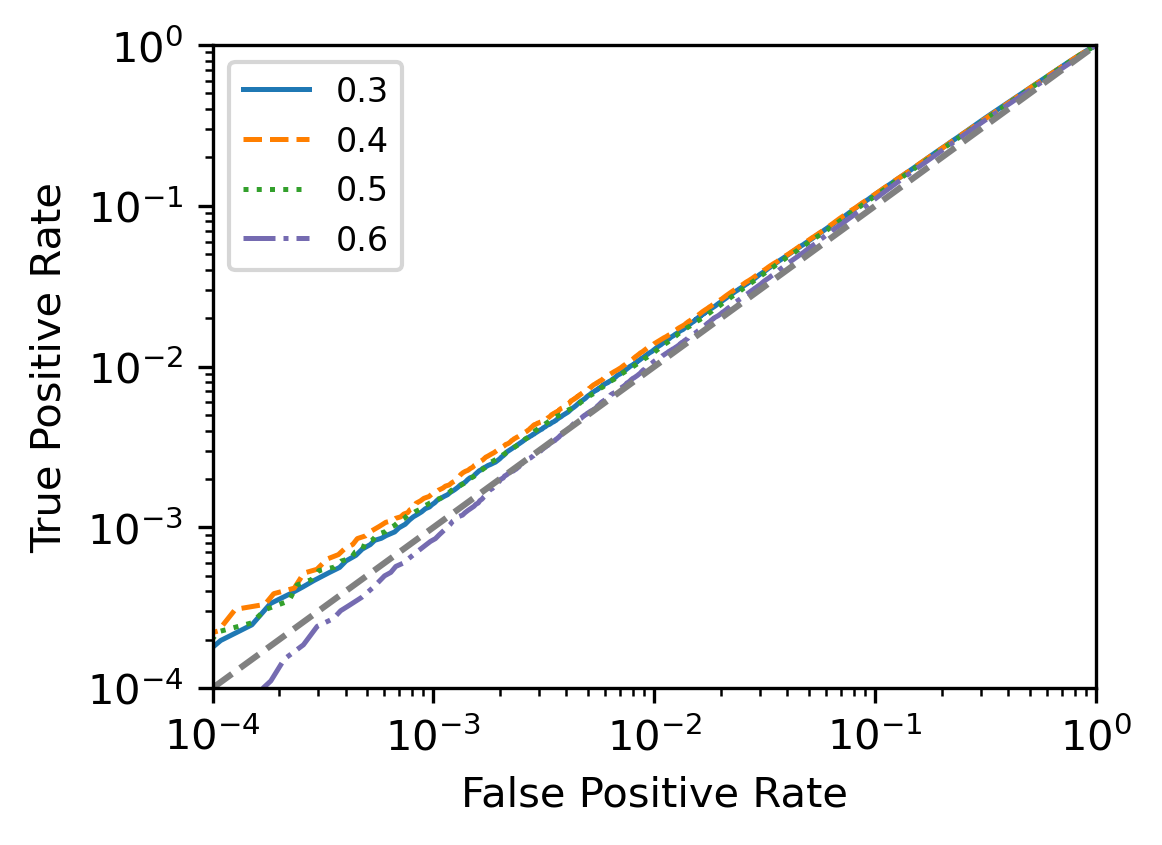}
    \includegraphics[width=.32\linewidth]{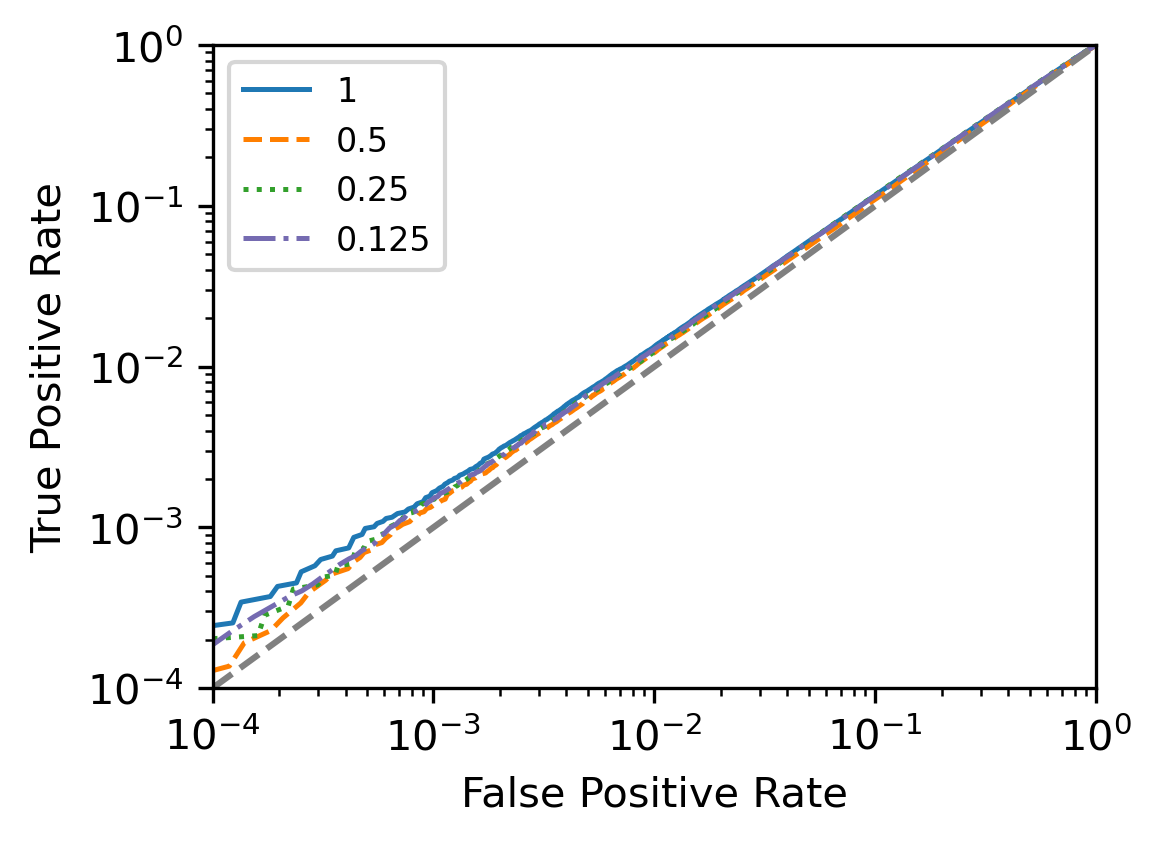}
    \includegraphics[width=.32\linewidth]{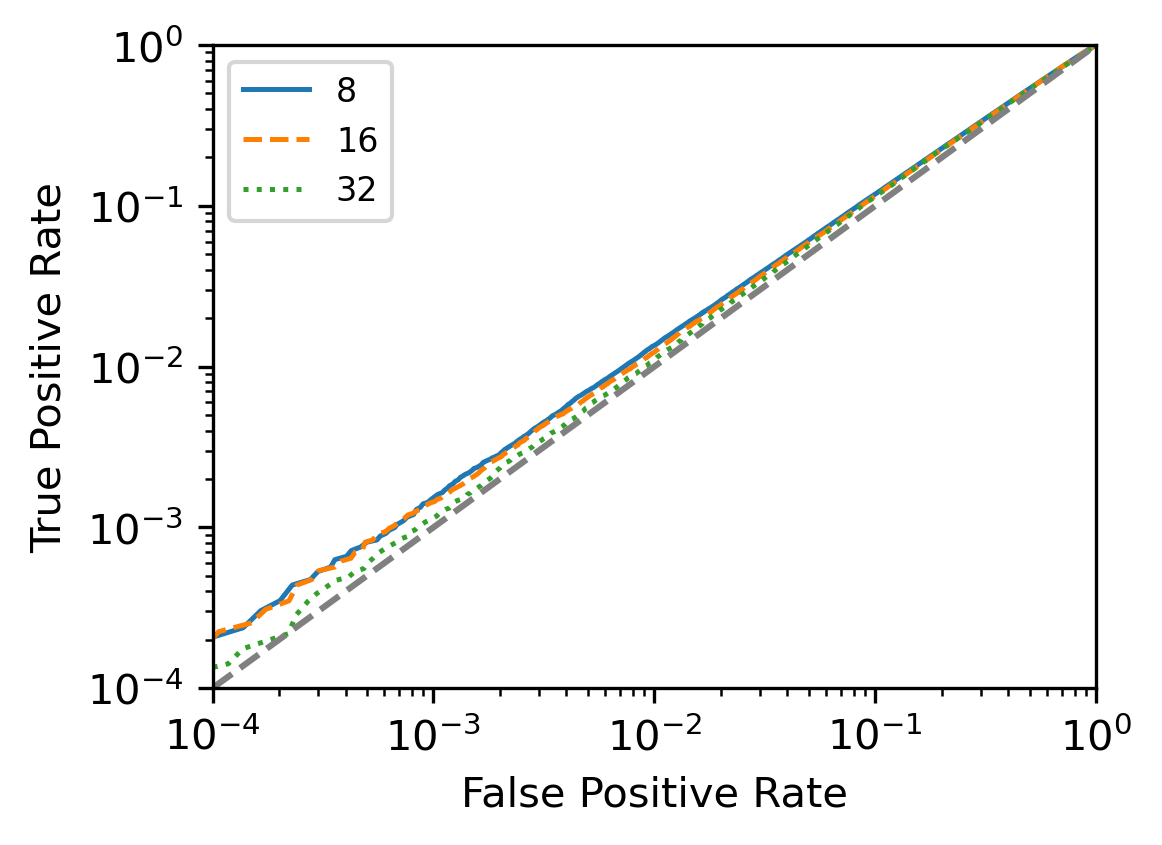}
    \includegraphics[width=.32\linewidth]{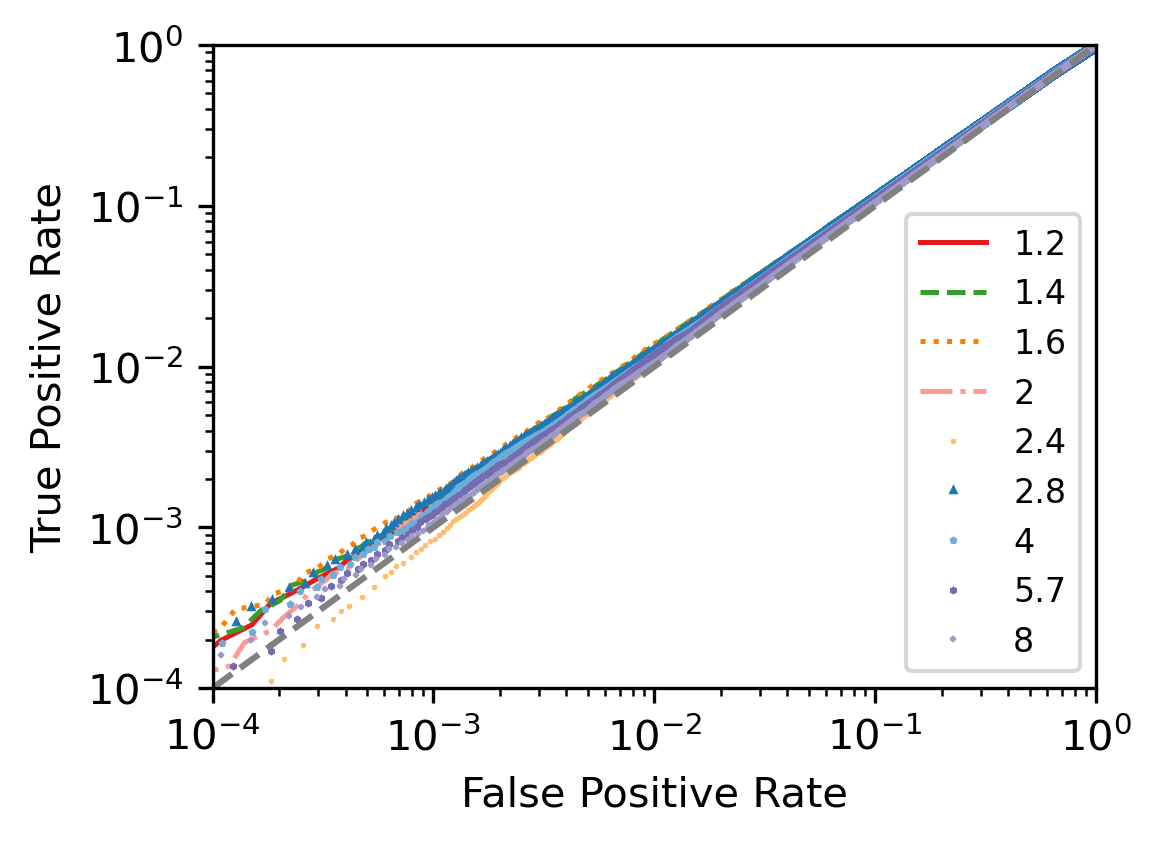}
    \vspace{-0.8em}
    \caption{%
        \textbf{Membership risks of NSDEs trained with different configurations.}
        (Top row, left to right) We illustrate the impact of the number of blocks ($n$), 
        the level of stochasticity ($k$), and the step-size ($s$). 
        (Bottom row, left to right) We show the impact of 
        the integration time $T$ and the noise intensity $\sigma$. 
        We use LiRA to assess the membership risks. The results are in CIFAR-10.
    }
    \label{fig:nsde-risk-ablation}
    \vspace{-0.8em}
\end{figure*}

\topic{Privacy-utility trade-off.}
The right two plots in Figure~\ref{fig:nsde-privacy-analysis}
illustrate the test accuracy of NSDEs trained 
with different privacy budgets ($\varepsilon$).
We record the test accuracy and privacy leakage 
at each training epoch for the 16 models 
used in our evaluation, and we plot them.
The second figure from the right shows 
the trade-offs across different failure probabilities ($\delta$).
As $\delta$ increases, the test accuracy of NSDEs 
increases under the same privacy budget.
This is because a larger $\delta$ permits 
a higher likelihood of privacy leakage.
Surprisingly, NSDEs achieve a test accuracy of 80\%
with $\varepsilon$ = 2 and $\delta$ = \(10^{-5}\) under GDP,
which demonstrates their ability to maintain high utility 
even at relatively low privacy budgets.
The rightmost plot demonstrates that 
all three accountants produce similar estimates of privacy leakage and test accuracy.
However it is important to note the variability in test accuracy:
even under the same privacy budget, e.g., $\varepsilon\!=\!2$,
test accuracy ranges from 65--82\%.
We attribute this variance to the factors like 
randomness introduced by the diffusion term 
or the training data stochasticity.
In practice, to train NSDEs with high utility 
we may need to train the same model multiple times,
as shown in prior work~\cite{10.1145/3580305.3599561}.

\subsection{Impact of NSDE Configurations}
\label{subsec:sde-ablation-study}

We now turn our attention to 
the impact of NSDE configurations on their membership risks.
In this analysis, we maintain consistent hyper-parameters 
that affect privacy accounting when varying,
such as the number of training iterations, batch-size and learning rate.
The only exception is the noise intensity ($\sigma$)
which directly influences the total privacy leakage.
We examine five configurations:
the number of NSDE blocks $n$,
the level of stochasticity $k$,
the integration time interval $T$,
the step-size {$1/s$}, 
and the noise intensity $\sigma$.
It is important to note that 
the noise intensity depends on the first three variables,
following the formulation $\frac{k}{\sqrt{T/s}}$.
\smallskip

\topic{The number of SDE blocks ($n$).}
We investigate the impact of adding more SDE blocks on privacy.
Due to the post-processing property of DP,
connecting multiple SDE blocks sequentially---%
each with the privacy guarantee $\varepsilon$---%
still results in an overally privacy leakage bounded by $\varepsilon$.
However, we expect that the diffusion terms in these blocks 
will introduce random noise multiple times, 
leading to a substantial decrease in model utility.
We vary $n \in \{1, 2, 3\}$.

The top-left plot in Figure~\ref{fig:nsde-risk-ablation}
illustrates our results.
Membership risks associated with NSDE models 
are consistent when $n\!=\!1$ or $n\!=\!2$.
But employing three SDE blocks reduces the risks by a factor of 2,
as measured by TPRs at 0.1--1\% FPR.
We attribute this reduction to poor generalization when $n\!=\!3$.
For instance, in CIFAR-10, 
while NSDEs with $n \in \{1, 2\}$ achieve test accuracy of 81.9--82.4\%,
models with $n\!=\!3$ show significantly lower test accuracy of 74.4\%.
Interestingly, unlike NODEs, 
augmenting the number of SDE blocks in NSDEs 
leads to a corresponding reduction in membership risks.
\smallskip

\topic{The level of stochasticity ($k$)}
directly increases the novel intensity ($\sigma$)
for the Gaussian diffusion term, following the formula $\sigma = \frac{k}{\sqrt{T/s}}$).
Increasing $k$ adds more noise to the modeling dynamics,
and therefore, a reduction in membership risks.
We vary $k$ across \{0.3, 0.4, 0.5, 0.6\},
with a default value of 0.5 for all our experiments.

The top-middle plot %
illustrates our findings.
We show that NSDEs with varying $k$ values
effectively mitigate membership risks, as measured by LiRA.
With $k$ values up to 0.5, NSDEs exhibit the same level of effectiveness.
But when $k$ reaches to 0.6,
it completely mitigates the membership risks,
bringing the ROC curve below that of a random classifier (a diagonal line).
We also examine how each choice of $k$ affects model performance. 
Increasing the level of stochasticity 
results in a decrease in model utility.
On CIFAR-10, the test accuracy drops from 82.6\% to 57.1\% 
as $k$ increases from 0.3 to 0.6.
This result implies that the state-of-the-art attack (LiRA) 
is yet weak so that one can completely mitigate with $k$ 
that does not significantly degrade model utility.
\smallskip

\topic{The integration time interval $T$.}
We evaluate the impact of varying integration time intervals by
adjusting $T$ and measuring its impact on membership risks.
We set $T$ across \{0.125, 0.25, 0.5, 1\},
with the default value of $T\!=\!1$ used in all our experiments.

The top-right plot in the same figure %
illustrates our findings.
Overall, we find no significant interaction between
the integration time interval ($T$) and membership risks.
All NSDE models consistently show low membership risks.
But theoretically, decreasing $T$ will increase the noise intensity ($\sigma$),
leading to a smaller privacy leakage ($\varepsilon$)
with reduced model utility and membership risks.
\smallskip

\topic{The step-size ($1/s$),}
which controls the granularity 
at which the integration over time $T$ is computed.
The increase in $s$ results in a model 
less capable of learning details of dynamics in the training data.
We vary $s$ across \{8, 16, 32\},
with the default value of $16$.

The bottom-left plot in Figure~\ref{fig:nsde-risk-ablation} summarizes our findings.
Overall, we observe only marginal variations in membership risks
caused by changes in step size ($1/s$).
On CIFAR-10, the membership risks observed at $s\!=\!8$ or $16$ are comparable.
However, we find a reduction of the risks with $s\!=\!32$,
likely because the model lacks learning detailed dynamics from the training data.
Our analysis indicates that
this reduction comes at the cost of accuracy.
A step size of $32$ limits NSDEs effectiveness in 
modeling underlying dynamics, and in consequence, 
they experience substantial performance degradation.
NSDEs at this step size achieve an average accuracy of 35.9\%,
which is 46.5--46.1\% lower than that achieved with $s$ of 8 or 16.
\smallskip

\topic{The noise intensity ($\sigma$).}
We finally examine how the noise intensity $\sigma$ affect membership risks,
though this parameter being influenced by our choice of $\sigma, T$ and $s$.
To evaluate this, we collect results
across all combinations of these three configurations from our previous experiments.
This results in unique $\sigma$ values of 
\{1.2, 1.4, 1.6, 2.0, 2.8, 4.0, 5.7, 8.0\}, with duplicates removed.
We use LiRA to quantify (empirical) membership risks.

The bottom-right plot %
illustrates our results.
Overall, we observe that all the noise intensity values 
we examine effectively mitigate membership risks,
consistent with prior work~\cite{lira}, 
which shows that even the strongest membership inference attacks
can be easily countered by a weak privacy guarantee provided by DP-SGD.
We also demonstrate that as the noise intensity gets stronger, NSDE's effectiveness increases.
However, it is important how we configure the noise intensity.
For instance, setting a lower noise intensity, such as 1.6,
by employing a larger step-size (e.g., 32),
may appear to reduce membership riusks in the figure.
But this reduction is primarily due to a substantial drop in model utility,
as shown in our analysis of different step-sizes.
\smallskip

\subsection{NSDE as a Drop-in Replacement Module}
\label{subsec:blocks}

All our evaluations so far assume that 
a defender trains NSDE models from scratch.
This approach often requires more computational resources than expected,
and could potentially face pushback, 
especially when well-performing pre-trained models are available.
To address this issue, 
we present an efficient \emph{replace-then-finetune} strategy.
Instead of training NSDEs from scratch, 
we propose \emph{replacing} the last few layers of an existing pre-trained model 
with an NSDE block and \emph{then finetune}.
This approach reduces computational demands 
by re-using the majority of pre-trained parameters,
as it requires fewer training iterations, 
while also benefiting from NSDEs’ privacy advantages 
(based on DP's composition theorem).
It also potentially improves model utility, as fine-tuning leverages the features 
(i.e., the latent representations/activations obtained at the penultimate layer), 
a strategy demonstrated to be effective in prior work 
on training models with DP-SGD~\cite{tramer2021differentially}.
In the following, we present results on CIFAR-10.
We refer the readers to Appendix 
for results on other datasets.
\smallskip

\begin{figure}[ht]
    \centering
    \vspace{-0.4em}
    \includegraphics[width=0.9\linewidth]{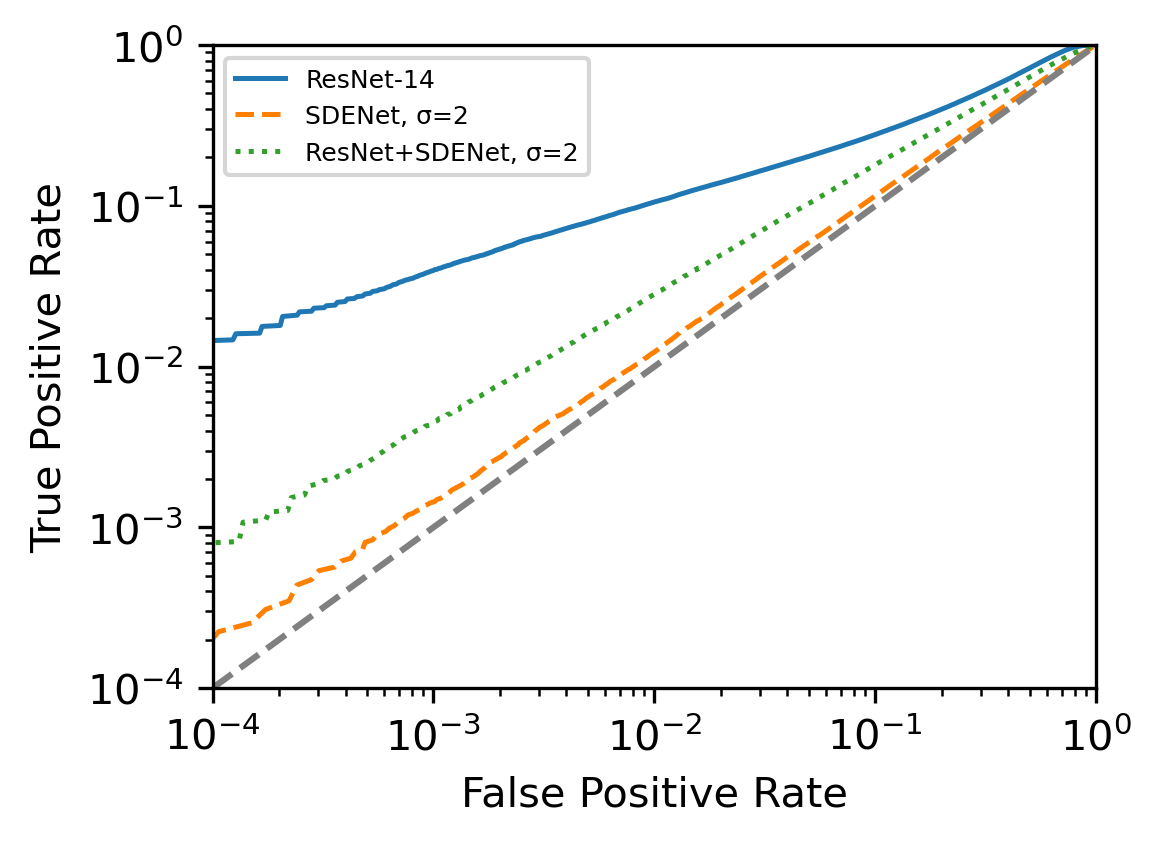}
    \vspace{-0.8em}
    \caption{\textbf{Effectiveness of \emph{replace-and-finetune} strategy.}
    We replace the last residual block of ResNet14 with a NSDE block 
    and then finetune the adapted model on CIFAR-10.}
    \label{fig:replace-and-finetune}
\end{figure}

\topic{Impact on privacy.}
Figure~\ref{fig:replace-and-finetune} demonstrates 
the effectiveness of the replace-then-finetune strategy.
When the last group of the final residual block is replaced with an NSDE block,
membership risks are reduced by more than an order of magnitude,
particularly in low FPR regions, compared to ResNet-14.
However, when compared to the NSDE models trained from scratch,
membership risks are higher.
This does not imply a difference in the privacy guarantees 
($\varepsilon$) between the two models.
Their guarantees remain the same due to the post-processing property.%
\footnote{Note that the last NSDE block is randomly initialized and fine-tuned.}
LiRA may hit the worst-case privacy leakage
against the replace-then-finetuned models,
as indicated by the bound shown in Figure~\ref{fig:compare-with-the-bound}.
\smallskip

\topic{Impact on model utility.}
We also observe that the test accuracy of 
the replace-then-finetuned models 
is 84.0\%, a modest decrease from the 86.0\% 
achieved by ResNet-14.
However, this reduction is comparatively smaller 
than that observed in NSDEs trained from scratch,
which achieves a test accuracy of 81.9\%.
We emphasize that this approach also offers 
potentially reductions in  
the computational demands of NSDE training.
Notably, we observe that test accuracy stabilizes after 150 epochs,
which is half the total number of epochs we use for training NSDEs.
As shown in our theoretical analysis in \S\ref{subsec:theoretical-analysis},
a reduction in the number of training iterations ($K$) 
also decreases the total privacy leakage of the trained models.
Our results demonstrate 
the effectiveness of the replace-then-finetune strategy 
in reducing membership inference risks 
while maintaining high utility and reducing computational overhead.

\section{Related Work}
\label{sec:related}

\topic{Privacy attacks on machine learning.}
Membership inference attacks%
---where an adversary, given a particular instance,
attempts to infer whether it was part of a model's training data%
---represent a worst-case form of privacy leakage 
and have been widely studied as a key mechanism 
for privacy auditing in machine learning.
However, other types of privacy attacks 
have also been explored in the literature.

Fredrikson \textit{et al.}~\cite{fredrikson2014privacy} 
introduced the concept of \emph{model inversion},
a class of reconstruction attacks 
aimed at reverse-engineering training samples
or recovering sensitive information 
in the training data that has not exposed to public.
Initial work~\cite{fredrikson2014privacy} 
focused on attacks targeting linear models,
while subsequent work~\cite{10.1145/2810103.2813677, 7536387} 
extended these attacks to non-linear models, 
including decision trees and neural networks.
In particular, these attacks against neural networks
typically involve iterative gradient-based optimization
to refine inputs until they resemble the original training data.

When a model inversion is used to infer sensitive attributes---%
such as gender or race---of specific training instances,
it is referred to as \emph{attribute inference attacks}~\cite{yeom}.
These attacks exploit query access to a target model 
to extract unknown attributes from partially observed training records. 
Various attack methodologies~\cite{yeom, 280036, fredrikson2014privacy} have been proposed, 
and attribute inference has been exploited across domains, 
such as social media~\cite{%
    10.1145/1526709.1526781, chaabane2012you, 10.1145/2594455, 
    197122, 10.1145/3038912.3052695, 10.1145/3154793} 
or recommender systems~\cite{10.1145/2365952.2365989, 10.1145/3397271.3401144}.

The recent emergence of advanced neural networks,
such as language models, diffusion models, 
and differential equation-based networks like NODEs,
necessitates the adaptation and evaluation of 
existing privacy attacks on these models.
While much attention has been given to 
language models and diffusion models, 
with prior studies~\cite{236216, 274574, 291199} 
demonstrating their susceptibility to 
membership inference and model inversion 
(or data extraction) attacks,
no prior work has systematically studied 
the privacy risks of differential equation-based 
neural networks such as NODEs and NSDEs.
Our work addresses this gap by presenting 
the first study on understanding and mitigating 
the membership risks associated 
with differential equation-based models.
\smallskip

\topic{Robustness of differential equation-based networks.}
In the existing literature, 
the security and privacy aspects of NODEs and NSDEs 
have mainly focused on countering 
\emph{adversarial examples}~\cite{43405}
and ensuring the stability of learned dynamics.
Time-invariant steady NODEs (TisODEs)~\cite{hanshurobustness},
for instance, empirically demonstrated that 
NODEs are more resilient to random Gaussian noise
and adversarial examples than 
traditional convolutional neural networks.
TisODEs introduced a steady-state regularizer 
that minimizes the deviation of hidden states $~h(t)$
beyond a terminal time $T$ (i.e., steady-state after $T$), 
effectively encouraging the network to learn stable dynamics.
Similarly, Kang~\textit{et al.}~\cite{kang2021stable}
proposed guiding trajectory evolution toward 
Lyapunov-stable equilibrium points,
enabling NODEs to extract features around stable attractors.
This approach led to improved robustness against adversarial examples.
The idea has been further extended to non-autonomous dynamical systems, 
guiding trajectories to asymptotically stable equilibrium points~\cite{li2022defending}. 
In~\cite{zeqiri2023efficient},
robustness analysis has also been extended to 
high-dimensional NODEs using graph-based verification tools.

Other techniques to improve NODE robustness 
include using skew-symmetric ODE blocks~\cite{huang2022adversarial} 
and integrating Guassian process~\cite{anumasa2021improving},
both of which empirically enhance robustness.
Though not explicitly designed for adversarial robustness, 
some efforts contribute similarly by regularizing NODE dynamics.
Kinetic and Jacobian norm regularizers~\cite{finlay2020train}, 
along with higher-order derivative regularizers~\cite{kelly2020learning}, 
have been proposed to enforce smoothness in NODE dynamics 
and have been shown to be effective. %
The method proposed in STEER~\cite{ghosh2020steer}
mitigates the learning of stiff ODEs 
by randomizing the terminal times during training.

In contrast, there are only handful of works 
that have investigated security and privacy aspects of NSDEs.
A few notable works include Jia~\textit{et al.}~\cite{jha2021smoother},
who demonstrated that attributions generated from NSDEs 
are less noisy, visually sharper, and quantitatively more robust 
than those computed from deterministic NODEs.
Liu \textit{et al.}~\cite{liu2020how} 
also showed formally and empirically that NSDEs improve 
robustness against adversarial examples.
More broadly, NSDEs have been widely adopted 
as a tool to estimate uncertainty~\cite{kong2020sde}
and generative modeling~\cite{kidger2021neural}.%

\section{Conclusion}
\label{sec:conclusion}

This work studies the privacy implications of emerging neural networks,
specifically NODEs, through the lens of membership inference attacks.
To our knowledge, this is the first work that
conducts a \emph{comprehensive analysis} of the membership risks
associated with NODEs.
Our key findings are:
(1) NODEs are twice as resilient to 
existing membership inference attacks
compared to conventional feedforward models like ResNets,
while maintaining comparable accuracy.
(2) This privacy benefit stems from reduced model overfitting.
NODEs exhibit a smaller generalization gap 
compared to conventional models,
which we attribute to their constrained expressivity%
---restricting the modeling of underlying dynamics 
to a system of differential equations.
(3) Moreover, advanced, non-stochastic variants of NODEs
can also reduce overfitting, leading to both improved accuracy 
and a further reduction of membership inference risks.

Our analysis does \emph{not} imply 
that NODEs are free from membership inference risks.
To further mitigate these risks,
we study NSDEs, a stochastic extension of NODEs that
incorporates an additional diffusion term into system modeling.
Most importantly, we \emph{formally show} that 
this diffusion term functions as 
a differentially-private (DP) mechanism,
thereby establishing NSDEs as DP-learners.
Our theoretical analysis shows that
the privacy guarantee ($\varepsilon$) offered by NSDEs 
is equivalent to that of DP-SGD when using identical mini-batch training configurations.
We demonstrate empirically that 
NSDEs are more effective at reducing membership risks
compared to existing defenses that lack provable privacy guarantees.
Unlike DP-SGD, NSDEs achieve this protection 
with a smaller impact on model utility.

Moreover, we present an effective strategy,
\emph{replace-then-finetune}, for using NSDEs 
as a drop-in privacy-enhancing module.
We replace the last few layers of 
a pre-trained conventional feedforward model 
with an NSDE block and finetune the adapted model for a few epochs.
In evaluation, this strategy reduces the membership inference risks 
of the original pre-trained model by an order of magnitude,
while achieving higher model utility than DP-SGD
and requiring only half the training cost 
(see Appendix for comparisons).

Our work demonstrates that 
systems learned by neural networks 
can be effectively \emph{modeled with privacy}.
Our findings imply that advanced neural networks
incorporating diffusion mechanisms, 
such as stable diffusion models, 
may, with careful modeling, serve as inherently private models
while preserving its strong performance.
As future work, investigating these private systems
may offer valuable insights into 
how DP shapes what learned by neural networks.
We also hope our findings inspire the broader adoption 
of the replace-then-finetune strategy
as a practical and effective method 
for enhancing model privacy.

\ifCLASSOPTIONcompsoc
  \section*{Acknowledgments}
\else
  \section*{Acknowledgment}
\fi
\noindent
S.H. is partially supported by the Google Faculty Research Award (2023). K.L. acknowledges support from the U.S. National Science Foundation under grant IIS 2338909. 
A.G. acknowledges support from the U.S. Department of Energy, Office of Advanced Scientific Computing Research under the Scalable, Efficient, and Accelerated Causal Reasoning for Earth and Embedded Systems (SEA-CROGS) project. Sandia National Laboratories is a multimission laboratory managed and operated by National Technology \& Engineering Solutions of Sandia, LLC, a wholly owned subsidiary of Honeywell International Inc., for the U.S. Department of Energy’s National Nuclear Security Administration under contract DE-NA0003525. This paper describes objective technical results and analysis. Any subjective views or opinions that might be expressed in the paper do not necessarily represent the views of the U.S. Department of Energy or the United States Government. This article has been co-authored by an employee of National Technology \& Engineering Solutions of Sandia, LLC under Contract No. DE-NA0003525 with the U.S. Department of Energy (DOE). The employee owns all right, title and interest in and to the article and is solely responsible for its contents. The United States Government retains and the publisher, by accepting the article for publication, acknowledges that the United States Government retains a non-exclusive, paid-up, irrevocable, world-wide license to publish or reproduce the published form of this article or allow others to do so, for United States Government purposes. The DOE will provide public access to these results of federally sponsored research in accordance with the DOE Public Access Plan https://www.energy.gov/downloads/doe-public-access-plan.

{
    \bibliographystyle{IEEEtran}
    \bibliography{bib/this_work}
}

\appendix

\section*{A. Model Architectures}
\label{appendix:model-architectures}

We adopt architectures examined in prior work~\cite{oganesyan2020stochasticity}.
\smallskip

\topic{ResNet-14} is defined as
\[
\begin{split}
    [&\text{Conv2d} \rightarrow \text{BN} 
    \rightarrow [\text{ResBlock} \times 3 \rightarrow \text{BN} \rightarrow \text{Conv2d}]\times 2 \\
    &\rightarrow \text{ResBlock} \times 3 \rightarrow \text{AdaptiveAvgPool2d} \rightarrow \text{Linear}],
\end{split}
\]
where \text{ResBlock} is composed as 
$
    [\text{BN} \rightarrow \text{Conv2d} \rightarrow \text{BN} \rightarrow \text{Conv2d}].
$

\topic{ODENet-16-32-64} is defined same as ResNet-14 except that ResBlock is replaced with ODEBlock, which is defined as
$
    [\text{ConcatConv2d} \rightarrow \text{ConcatConv2d}],
$
where \text{ConcatConv2d} augments $x$ by concatenating an additional feature channel derived from $t$. 
\smallskip

Both models begin with Conv2d that processes a 3-channel input image to produce 32 feature maps using a 3$\times$3 filter with a stride of 2 and padding of 1. Following the initial convolutional layer, the features pass through a group of either three ResBlocks or a single ODEBlock, depending on the configuration. Subsequently, the network transitions through a Downsampling layer, which effectively reduces the spatial dimensions while adjust the channel depth, preparing the features for subsequent layers. After the third group of three ResBlocks or an ODEBlock, the network utilizes an adaptive average pooling layer to condenses the feature maps into a single dimensional vector per feature map. This vector is then flattened and passed through a linear layer which acts as the classifier of the network, outputting the final predictions across the categories.
\smallskip

\smallskip

\begin{table*}[t]
\centering
\caption{%
    \textbf{Membership inference risks of ODENet-64 models and SDENet $\sigma\!=\!0$ models.}
    We evaluate ODENet-64 models and and SDENet $\sigma\!=\!0$ models against six existing attacks on CIFAR-10 and CIFAR-100.
}
\label{tbl:mia-odenet64}
\adjustbox{max width=\linewidth}{
\begin{tabular}{@{}llcccccccccccc@{}}
\toprule
 &  & \multicolumn{2}{c}{\textbf{Acc. (C-10)}} & \multicolumn{2}{c}{\textbf{Acc. (C-100)}} & \multicolumn{2}{c}{\textbf{TPR @ 0.1\% FPR}} & \multicolumn{2}{c}{\textbf{TPR @ 1\% FPR}} & \multicolumn{2}{c}{\textbf{AUC}} & \multicolumn{2}{c}{\textbf{Inference acc.}} \\ \midrule
\textbf{Model} & 
    \textbf{Method} &
    \textbf{Train} & \textbf{Test} & \textbf{Train} & \textbf{Test} &
    \textbf{C-10} & \textbf{C-100} & 
    \textbf{C-10} & \textbf{C-100} & 
    \textbf{C-10} & \textbf{C-100} & 
    \textbf{C-10} & \textbf{C-100} \\ \midrule \midrule
\multirow{6}{*}{ODENet-64} & 
    Yeom et al.~\cite{yeom} & 
    \multirow{6}{*}{94.14\%} & \multirow{6}{*}{84.36\%} & \multirow{6}{*}{70.93\%} & \multirow{6}{*}{51.85\%} & 
    0.00\% & 0.03\% & 0.28\% & 1.17\% & 0.536 & 0.659 & 53.57\% & 66.09\% \\
    & Shokri et al.~\cite{shokri} & 
    &  &  &  &  
    0.08\% & 0.09\% & 0.75\% & 0.85\% & 0.473 & 0.480 & 50.03\% & 50.01\% \\
    & Song and Mittal~\cite{song2021scorebased} & 
    &  &  &  &  
    0.01\% & 0.13\% & 0.53\% & 1.58\% & 0.511 & 0.497 & 51.17\% & 51.63\% \\
    & Watson et al.~\cite{watson2022on} & 
    &  &  &  &  
    0.00\% & 0.03\% & 0.28\% & 1.17\% & 0.536 & 0.659 & 53.57\% & 66.09\% \\
    & Carlini et al.~\cite{lira} & 
    &  &  &  &  
    0.86\% & 1.72\% & 4.17\% & 7.70\% & 0.639 & 0.741 & 59.57\% & 69.49\% \\ 
    & Zarifzadeh et al.~\cite{zarifzadeh2023low} & 
    &  &  &  &  
    2.72\% & 2.91\% & 8.40\% & 11.47\% & 0.617 & 0.757 & 59.96\% & 70.65\% \\
    \midrule
\multirow{6}{*}{SDENet $\sigma\!=\!0$} & 
    Yeom et al.~\cite{yeom} &
    \multirow{6}{*}{89.09\%} & \multirow{6}{*}{83.93\%} & \multirow{6}{*}{68.01\%} & \multirow{6}{*}{52.49\%} &
    0.00\% & 0.09\% & 0.39\% & 1.12\% & 0.538 & 0.625 & 54.00\% & 60.52\% \\
    & Shokri et al.~\cite{shokri} &
    &  &  &  & 0.07\% & 0.05\% & 0.68\% & 0.63\% & 0.468 & 0.449 & 50.02\% & 50.00\% \\
    & Song and Mittal~\cite{song2021scorebased} &
    &  &  &  & 0.00\% & 0.14\% & 0.50\% & 1.40\% & 0.516 & 0.496 & 51.65\% & 51.13\% \\
    & Watson et al.~\cite{watson2022on} &
    &  &  &  & 0.00\% & 0.09\% & 0.44\% & 1.12\% & 0.538 & 0.625 & 54.00\% & 60.52\% \\
    & Carlini et al.~\cite{lira} &
    &  &  &  & 0.35\% & 1.03\% & 2.02\% & 5.36\% & 0.559 & 0.676 & 54.09\% & 62.26\% \\
    & Zarifzadeh et al.~\cite{zarifzadeh2023low} & 
    &  &  &  &  
    0.91\% & 2.04\% & 4.01\% & 8.04\% & 0.579 & 0.701 & 55.76\% & 63.58\% \\
    \bottomrule
\end{tabular}
}
\end{table*}

\section*{B. Differential Privacy of Gaussian Mechanism}
\label{app:proof}

The following result is proved in \cite{dwork2014algorithmic}. Since it is crucial to the results in the body, a slightly different proof is provided here for completeness.

\begin{proposition}[Theorem A.1 in \cite{dwork2014algorithmic}]\label{prop:GaussianDP}
    Let $\varepsilon \in (0,1)$.  The Gaussian mechanism, $~M(~q) = ~f(~q) + \mathcal{N}\lr{0,\sigma^2~I}$ with variance parameter $\sigma \geq \sqrt{2\ln\lr{1.25/\delta}}(S_{~f}/\varepsilon)$ is $(\epsilon,\delta)$-differentially private.
\end{proposition}
\begin{proof}
    Consider two adjacent datasets $~q, ~q'$, i.e., $|~q-~q'|<1$. Without loss of generality, write $~f(~q)=~f(~q')+~v$ for some $~v\in\mathbb{R}^n$.  Drawing $~X\sim\mathcal{N}\lr{~0,\sigma^2~I}$, the proof is based on analyzing the privacy loss random variable,
    \begin{align*}
        \ln&\lr{\frac{\Pr[~M(~q)=~f(~q)+~X]}{\Pr[~M(~q')=~f(~q)+~X]}} = \ln\lr{\frac{\exp\lr{\frac{-1}{2\sigma^2}\nn{~X}^2}}{\exp\lr{\frac{-1}{2\sigma^2}\nn{~X+~v}^2}}} \\
        &= \frac{1}{2\sigma^2}\lr{\nn{~X+~v}^2-\nn{~X}^2} = \frac{\IP{~v}{~X}}{\sigma^2} + \frac{\nn{~v}^2}{2\sigma^2},
    \end{align*}
    which is normally distributed with mean $\nn{~v}^2/\lr{2\sigma^2}$ and variance $\nn{~v}^2/\sigma^2$.  Choosing $~Z\sim\mathcal{N}\lr{~0,~I}$, the privacy loss random variable is bounded by $\varepsilon$ when 
    \begin{align*}
        \nn{\frac{\nn{~v}}{\sigma}~Z + \frac{\nn{~v}^2}{2\sigma^2}} \leq \frac{\nn{~v}}{\sigma}\lr{\nn{~Z}+\frac{\nn{~v}}{2\sigma}} \leq \varepsilon.
    \end{align*}
    In view of this, it suffices to bound the tail probability
    \begin{align*}
        \Pr\left[\nn{~Z}>\frac{\varepsilon\sigma}{\nn{~v}}-\frac{\nn{~v}}{2\sigma}\right] = 2\Pr\left[~Z>\frac{\varepsilon\sigma}{\nn{~v}}-\frac{\nn{~v}}{2\sigma}\right]<\delta.
    \end{align*}
    Choosing the standard deviation $\sigma=t\nn{~v}/\varepsilon$ for some $t>0$ to be determined, Mill's inequality guarantees that it is enough to require 
    \begin{equation*}
        \Pr\!\left[Z > t - \tfrac{\varepsilon}{2t}\right] \leq \tfrac{1}{\sqrt{2\pi}(t - \varepsilon/2t)} \exp\!\left(-\tfrac{1}{2}(t - \varepsilon/2t)^2\right) < \delta/2
    \end{equation*}
    or, said differently, it is enough to choose $t$ such that 
    \begin{align*}
        \ln\lr{t-\frac{\varepsilon}{2t}} + \frac{1}{2}\lr{t-\frac{\varepsilon}{2t}}^2 > \ln\lr{\frac{1}{\delta}\sqrt{\frac{2}{\pi}}}.
    \end{align*}
    Note that, by enforcing $t\geq 1$ and $0<\varepsilon <1$, the first term on the left-hand side is nonnegative when 
        $t-\frac{\varepsilon}{2t} \geq t - \frac{1}{2} \geq 1$.
    Therefore, it suffices to choose $t\geq 3/2$
    to ensure positivity.  On the other hand, the second term is increasing in $t$ when $t>\sqrt{\varepsilon/2}$ and decreasing in $\varepsilon$ when $\varepsilon<2t^2$, so it follows that we may ensure
    {\small
\begin{equation*}
(t\!-\!\tfrac{\varepsilon}{2t})^2 = t^2\!-\!\varepsilon\!+\!(\tfrac{\varepsilon}{2t})^2 \geq t^2\!-\!1\!+\!1\cdot 3^{-2} = t^2\!-\!\tfrac{8}{9} > 2\ln(\tfrac{1}{\delta}\sqrt{\tfrac{2}{\pi}}).
\end{equation*}
}
This implies that 
{\small
    \begin{equation*}
        t^2 > 2\ln\left(\frac{1}{\delta}\sqrt{\frac{2}{\pi}}\right) + \ln\left(e^{8/9}\right) = 2\ln\left(\frac{1}{\delta}\right) + \ln\left(e^{8/9}\frac{2}{\pi}\right)
    \end{equation*}}
    and since $(2/\pi)e^{8/9}<1.25^2$, this condition is satisfied when 
        $t^2 > 2\ln\lr{\frac{1.25}{\delta}}$.
    This shows that the privacy loss random variable is bounded above by $\varepsilon$ with probability $1-\delta$ on $~q,~q'$ whenever $\sigma \geq \sqrt{2\ln\lr{1.25/\delta}}\allowbreak (\nn{~v}/\varepsilon)$.  To remove the dependence on $~q,~q'$, it suffices to choose $\sigma \geq \sqrt{2\ln\lr{1.25/\delta}}(S_{~f}/\varepsilon)$.
\end{proof}

\section*{C. Additional Evaluation Results}
\label{appendix:additional-results}

\topic{Membership risks of NODEs, based on ODENet-64 architecture.}
Table~\ref{tbl:mia-odenet64} summarizes our additional results on NODE (ODENet-64) models and NSDEs with $\sigma\!=\!0$, respectively. It is notable that ODENet-64 demonstrates comparable results to ODENet-16\_32\_64 but with a decrease in test accuracy of 1.64\% and 2.41\% on CIFAR-10 and CIFAR-100, respectively. It is 2.5-6 $\times$ less susceptible to LiRA attacks compared to ResNet-14. On CIFAR-100, ODENet exhibits superior defense capabilities against LiRA attacks than ODENet-16\_32\_64, demonstrating a lower TPR@0.1\%FPR of 1.72\%, thereby demonstrating enhanced security against membership inference attacks compared to ResNet-14. 
For NSDEs with $\sigma\!=\!0$, despite being the non-privacy variant, it exhibits robust defense capabilities. NSDEs with $\sigma\!=\!0$ maintains competitive test accuracy and generalization gaps on both CIFAR datasets. Notably, NSDEs $\sigma\!=\!0$ showcases it strengths in defending against membership inference attacks, with TPR@0.1\%FPR of 0.35\% and 1.03\% for CIFAR-10 and CIFAR-100, respectively, demonstrating a stronger defense compared to ResNet-14 and NODE models. This superior performance in TPR@0.1\%FPR, coupled with robust accuracy metrics, underscores NSDE's ability to enhance security against advanced threats.
\smallskip

\begin{figure}[h]
    \centering
    \includegraphics[width=0.9\linewidth]{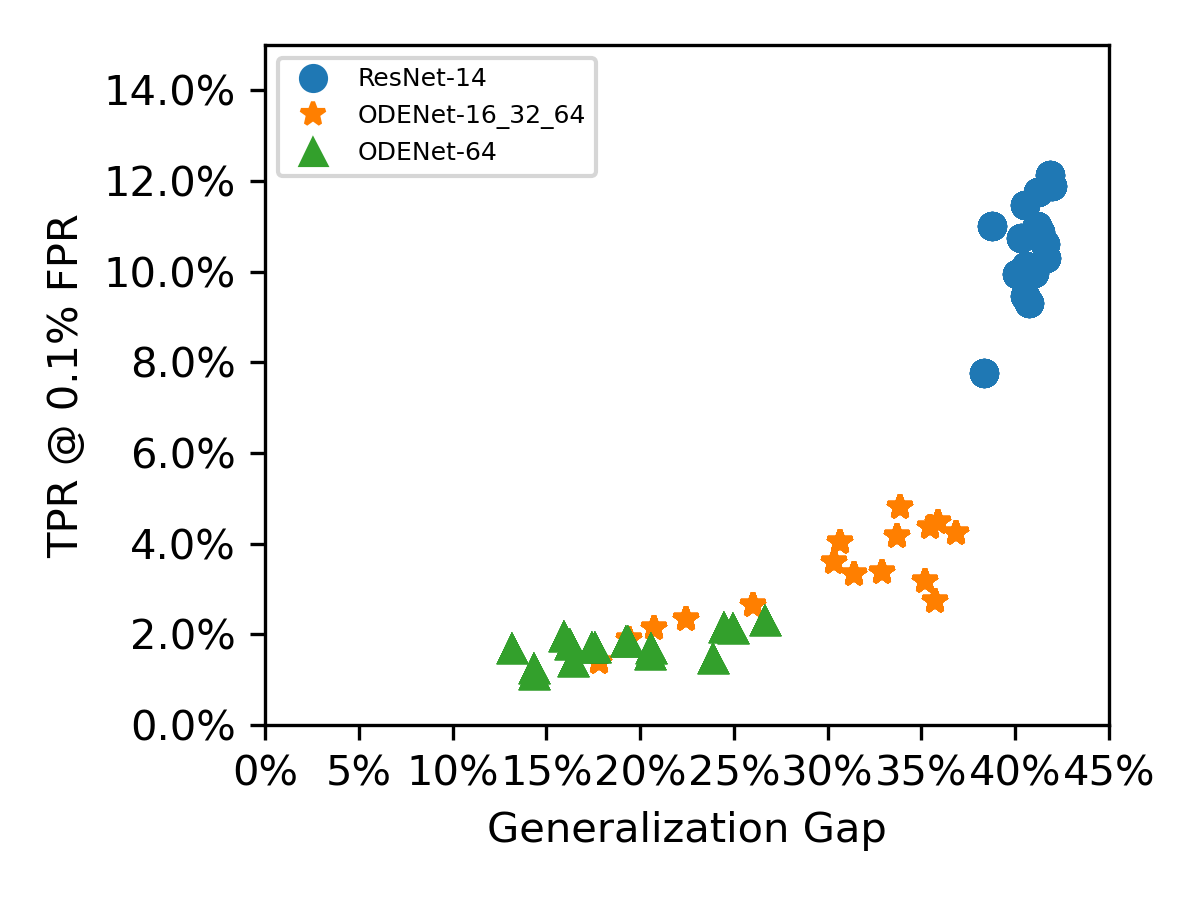}
    \caption{\textbf{Membership risks and overfitting} in CIFAR-100.}
    \label{fig:generalization_gap_cifar100}
\end{figure}

\topic{Impact of generalization gap in attack success (CIFAR-100).}
Figure~\ref{fig:generalization_gap_cifar100} illustrates the relationship 
between membership risks measured by LiRA 
and the generalization gap for models trained on CIFAR-100.
For ResNet-14, the TPR at 0.1\% FPR ranges from 8\%--14\%, 
even as the generalization gap remains consistently around 40\%.
This suggests that overfitting contributes to membership inference risks.
In contrast, NODEs (ODENet-16\_32\_64 and ODENet-64) show 
slightly lower generalization gaps, typically ranging from 19.08\%--30\%.
Their TPR at 0.1\% FPR values are also lower, 
fluctuating betweeen 0.5\%--5.5\%.
This suggests that NODEs are less prone to overfitting, 
thereby reducing their risks to membership inference attacks.
\smallskip

\topic{Membership risks of replace-then-finetune strategy.}
We analyze the computational resources required for training ResNet, NODEs, NSDEs and NSDEs with replace-then-finetune models. As shown in Table~\ref{tbl:computational_resources}, training NSDE-based models from scratch is computationally expensive. Table~\ref{tbl:mia-res-sde} demonstrates the results with membership inference risks of replace-then-finetune strategy. We evaluate NSDEs with replace-then-finetune models against six existing attacks on four datasets. Although only replacing the final blocks of a ResNet with SDENet, the replace-then-finetune models narrows the generation gap while achieving lower TPR at both 0.1\% and 1\% FPR across four dataset, indicating improved robustness compared to the ResNet models.  

\smallskip
\topic{Figure and tables for additional results.}
Here, we present full experimental results 
to support the findings reported in the main body of the manuscript: 
Tables~\ref{tbl:mia-blocks}--\ref{tbl:mia-defense}, 
where each table covers specific aspects. 
Table~\ref{tbl:mia-blocks} demonstrates the results 
with varying number of blocks of models in Figure~\ref{fig:node-risk-ablation} and~\ref{fig:nsde-risk-ablation}.
Table~\ref{tbl:mia-solvers} and~\ref{tbl:mia-variants} provide supporting data for Figure~\ref{fig:node-risk-ablation}, 
illustrating the results of the impact of solvers and variants for NODEs.
The detailed evaluation results of defenses, 
which provide an in-depth analysis corresponding to Figure~\ref{fig:nsde-vs-others}, 
are presented in Table~\ref{tbl:mia-defense}.

\begin{table}[htbp]
\centering
\caption{%
    \textbf{Computational resources required for training on CIFAR-10 and CIFAR-100.}
}
\label{tbl:computational_resources}
\adjustbox{max width=\linewidth}{
\begin{tabular}{@{}lcc@{}}
\toprule
\textbf{Model}  & \textbf{Time per epoch} & \textbf{Memory usage} 

\\ \midrule \midrule
    ResNet-14 & 10s & 941MB  \\ 
    NODE (ODENet-16-32-64)\_block1& 23s & 1305MB  \\
    NSDE (SDENet $k\!=\!2$)\_block1& 30s & 5051MB  \\
    NSDE (replace-then-finetune)& 15s & 1575MB  \\
    \bottomrule
\end{tabular}
}
\end{table}

\begin{table*}[t]
\centering
\caption{%
    \textbf{Membership risks of replace-and-finetune strategy ResNet+SDENet.}
}
\label{tbl:mia-res-sde}
\vspace{-1.em}
\adjustbox{max width=\linewidth}{
\begin{tabular}{@{}l | cccc | cccc | cccc | cccc@{}}
\toprule
 & \multicolumn{4}{c|}{\textbf{TPR @ 0.1\% FPR}} & \multicolumn{4}{c|}{\textbf{TPR @ 1\% FPR}} & \multicolumn{4}{c|}{\textbf{AUC}} & \multicolumn{4}{c}{\textbf{Inference acc.}} \\ \midrule
    \textbf{Method} &
    \textbf{F-M} & \textbf{C-10} & \textbf{C-100} & \textbf{T-I} & 
    \textbf{F-M} & \textbf{C-10} & \textbf{C-100} & \textbf{T-I} & 
    \textbf{F-M} & \textbf{C-10} & \textbf{C-100} & \textbf{T-I} & 
    \textbf{F-M} & \textbf{C-10} & \textbf{C-100} & \textbf{T-I} \\ \midrule \midrule
    Yeom et al.~\cite{yeom} & 
    0.01\% & 0.00\% & 0.07\% & 0.09\% & 
    0.06\% & 0.14\% & 1.10\% & 1.12\% & 
    0.558 & 0.555 & 0.646 & 0.638 &
    55.72\%& 55.81\% & 64.41\% & 61.22\% \\
    Shokri et al.~\cite{shokri} & 
    0.06\% & 0.05\% & 0.05\% & 0.12\% & 
    0.65\% & 0.52\% & 0.56\% & 1.04\% & 
    0.477 & 0.450 & 0.433 & 0.502 & 
    50.01\% & 50.00\% & 50.00\% & 50.20\% \\
    Song and Mittal~\cite{song2021scorebased} & 
    0.00\% & 0.01\% & 0.13\% & 0.13\% & 
    0.24\% & 0.24\% & 1.48\% & 1.19\% & 
    0.468 & 0.530 & 0.496 & 0.500 & 
    51.05\% & 53.29\% & 51.15\% & 50.36\% \\
    Watson et al.~\cite{watson2022on} & 
    0.00\% & 0.00\% & 0.07\% & 0.11\% &
    0.10\% & 0.15\% & 1.10\% & 1.17\% &
    0.550 & 0.555 & 0.646 & 0.537 &
    53.81\% & 55.81\% & 64.41\% & 52.83\% \\
    Carlini et al.~\cite{lira} & 
    1.39\% & 0.45\% & 1.65\% & 0.21\% & 
    4.48\% & 2.84\% & 7.51\% & 1.57\% &  
    0.576 & 0.600 & 0.726 & 0.535 & 
    53.27\% & 56.78\% & 67.37\% & 52.42\% \\ 
    Zarifzadeh et al.~\cite{zarifzadeh2023low} & 
    0.73\% & 1.41\% & 2.91\% & 0.26\% & 
    4.74\% & 5.68\% & 10.88\% &1.89\% & 
    0.544 & 0.607 & 0.743 & 0.561 & 
    54.80\% & 57.97\% & 68.67\% & 54.54\% \\
    \midrule

\end{tabular}
}
\end{table*}

\begin{table*}[t]
\centering
\caption{%
    \textbf{Comparing membership risks of NODEs and NSDEs.}
    We vary the number of blocks in \{1, 2, 3\} and use LiRA.
}
\label{tbl:mia-blocks}
\vspace{-1.em}
\adjustbox{max width=\linewidth}{
\begin{tabular}{@{}llcccccccccccc@{}}
\toprule
 &  & \multicolumn{2}{c}{\textbf{Acc. (C-10)}} & \multicolumn{2}{c}{\textbf{Acc. (C-100)}} & \multicolumn{2}{c}{\textbf{TPR @ 0.1\% FPR}} & \multicolumn{2}{c}{\textbf{TPR @ 1\% FPR}} & \multicolumn{2}{c}{\textbf{AUC}} & \multicolumn{2}{c}{\textbf{Inference acc.}} \\ \midrule
\textbf{Model} & 
    \textbf{\# of Blocks} &
    \textbf{Train} & \textbf{Test} & \textbf{Train} & \textbf{Test} &
    \textbf{C-10} & \textbf{C-100} & 
    \textbf{C-10} & \textbf{C-100} & 
    \textbf{C-10} & \textbf{C-100} & 
    \textbf{C-10} & \textbf{C-100} \\ \midrule \midrule
\multirow{3}{*}{ODENet-16\_32\_64} & 
    1 & 94.30\% & 84.41\% &82.05\% & 52.16\%& 1.01\% & 3.30\%&  4.22\% &12.11\%& 0.616 & 0.782& 57.65\% & 72.88\%\\
    & 2 & 96.39\% & 85.26\% & 93.40\% & 52.8\%&  1.57\% & 7.01\%& 5.95\% & 20.86\%& 0.638 & 0.848& 58.99\%  & 75.77\% \\
    & 3 & 96.82\% & 85.25\% & 94.71\%& 53.07\%& 1.66\% & 7.78\%& 6.17\% &22.38\%& 0.641 &0.852& 58.86\% & 76.32\% \\ \midrule
\multirow{3}{*}{ODENet-64} & 
    1 & 94.14\% & 84.36\% &70.93\%&51.85\% & 0.86\% &1.72\% &4.17\% &7.70\%& 0.639 & 0.741&59.57\% & 69.49\%\\
    & 2 & 97.07\% & 85.20\% &80.74\%&53.39\%& 1.25\% &2.41\% &5.43\% &10.59\%& 0.656 &0.788& 60.78\% & 64.33\%\\
    & 3 & 97.74\% & 85.82\% &87.26\%&53.78\%& 1.49\% & 3.29\%&6.05\% &13.71\%& 0.669 & 0.815&60.89\% &69.85\% \\ \midrule
\multirow{3}{*}{SDENet-16\_32\_64} & 
    1 & 81.89\% & 81.91\% &56.54\% & 50.64\%& 0.14\% & 0.43\%& 1.24\% & 3.00\%& 0.525 &0.612 & 51.89\% & 57.64\%\\
    & 2 & 82.41\% & 82.42\% & 57.20\%& 51.12\%& 0.14\% & 0.45\% &1.25\% & 2.99\%& 0.528&0.619 & 52.23\% & 58.57\%\\
    & 3 & 74.44\% & 74.96\% & 55.73\%& 48.27\%& 0.08\% &0.32\% &0.90\% & 2.58\%& 0.511&0.604& 50.00\% &57.43\%\\    \bottomrule
\end{tabular}
}
\end{table*}

\begin{table*}[t]
\centering
\caption{%
    \textbf{Membership risks of NODEs trained with different solvers.}
    We compare the LiRA success on them.
}
\label{tbl:mia-solvers}
\vspace{-1.em}
\adjustbox{max width=\linewidth}{
\begin{tabular}{@{}llcccccccccccc@{}}
\toprule
 &  & \multicolumn{2}{c}{\textbf{Acc. (C-10)}} & \multicolumn{2}{c}{\textbf{Acc. (C-100)}} & \multicolumn{2}{c}{\textbf{TPR @ 0.1\% FPR}} & \multicolumn{2}{c}{\textbf{TPR @ 1\% FPR}} & \multicolumn{2}{c}{\textbf{AUC}} & \multicolumn{2}{c}{\textbf{Inference acc.}} \\ \midrule
\textbf{Model} & 
    \textbf{Solver} &
    \textbf{Train} & \textbf{Test} & \textbf{Train} & \textbf{Test} &
    \textbf{C-10} & \textbf{C-100} & 
    \textbf{C-10} & \textbf{C-100} & 
    \textbf{C-10} & \textbf{C-100} & 
    \textbf{C-10} & \textbf{C-100} \\ \midrule \midrule
\multirow{3}{*}{ODENet-16\_32\_64} & 
    Euler & 94.30\% & 84.41\% &82.05\% & 52.16\%& 1.01\% & 3.30\%&  4.22\% &12.11\%& 0.616 & 0.782& 57.65\% & 72.88\%\\
    & Rk4 & 94.28\% & 84.39\% & 83.32\% & 52.19\%&  0.89\% & 3.70\%& 4.02\% & 13.07\%& 0.617 & 0.792& 57.78\%  & 73.05\% \\
    & Dopri5 & 94.41\% & 84.61\% & 80.20\%& 52.20\%& 0.68\% & 3.02\%& 3.40\% &11.20\%& 0.609 &0.777& 57.43\% & 64.34\% \\ \bottomrule
\end{tabular}
}
\end{table*}

\begin{table*}[t]
\centering
\caption{%
    \textbf{Membership risks of NODE variants.}
    We compare the LiRA success on these models in CIFAR-10.
}
\label{tbl:mia-variants}
\vspace{-1.em}
\adjustbox{max width=\linewidth}{
\begin{tabular}{@{}lccccccc@{}}
\toprule
\textbf{Model} & \textbf{Acc. (Train)} & \textbf{Acc. (Test)} & \textbf{Train-Test Acc. Gap} & \textbf{TPR @ 0.1\% FPR} & \textbf{TPR @ 1\% FPR} & \textbf{AUC} & \textbf{Inference acc.}

\\ \midrule \midrule
    SONODE & 92.35\% & 86.52\% & 5.9\% & 0.44\% & 2.62\% & 0.585 & 55.81\% \\
    ANODE(+64) & 89.09\% & 85.10\% & 4.0\% & 0.35\%& 2.15\% &  0.572 & 55.07\% \\ 
    ANODE(+16) & 87.60\% & 84.62\% & 3.0\% & 0.24\% & 1.76\% & 0.558 & 54.10\% \\
    HBNODE & 84.91\% & 82.11\% & 2.8\% & 0.24\% & 1.74\% & 0.545 & 53.10\% \\ \midrule
    ODENet-64 (Baseline) & 94.14\% & 84.36\% & 9.7\% & 0.86\%& 4.17\% & 0.639 & 59.57\%\\ \bottomrule
\end{tabular}
}
\end{table*}

\begin{table*}[t]
\centering
\caption{%
    \textbf{Effectiveness of defenses in CIFAR-10.}
    We use LiRA. %
    All defenses, except for SDENets, are applied to ResNet-14.
}
\label{tbl:mia-defense}
\vspace{-1.em}
\adjustbox{max width=\linewidth}{
\begin{tabular}{@{}lccccccc@{}}
\toprule
\textbf{Model} & \textbf{Acc. (Train)} & \textbf{Acc. (Test)} & \textbf{Train-Test Acc. Gap} & \textbf{TPR @ 0.1\% FPR} & \textbf{TPR @ 1\% FPR} & \textbf{AUC} & \textbf{Inference acc.}

\\ \midrule \midrule
    Memguard & 98.14\% & 85.97\% & 12.17\% & 0.00\% & 0.00\% & 0.580 & 54.48\% \\
    MMD+mix-up & 95.10\% & 87.39\% & 7.71\% & 2.03\%& 7.09\% &  0.628 & 57.95\% \\ 
    L2-Regularization & 98.84\% & 86.60\% & 12.24\% & 2.36\% & 7.81\% & 0.675 & 61.54\% \\
    L1-Regularization & 97.39\% & 86.48\% & 10.91\% & 1.06\% & 4.77\% & 0.641 & 59.47\% \\ 
    DP-SGD, $\epsilon$=8 & 73.49\% & 74.44\% & -0.95\% & 0.12\% & 1.18\% & 0.516 & 51.28\% \\ 
    \midrule
    ResNet-14 & 98.14\% & 85.97\% & 12.17\% & 3.96\%& 10.57\% & 0.679 & 61.15\%\\ 
    SDENet $\sigma\!=\!0$ & 89.09\% & 83.93\% & 5.16\% & 0.35\%& 2.02\% & 0.559 & 54.09\%\\ 
    SDENet $\sigma\!=\!2$ & 81.89\% & 81.91\% & -0.02\% & 0.14\%& 1.24\% & 0.525 & 51.89\%\\ 
    \bottomrule
\end{tabular}
}
\end{table*}

\begin{figure}[ht]
    \centering
    \includegraphics[width=0.8\linewidth]{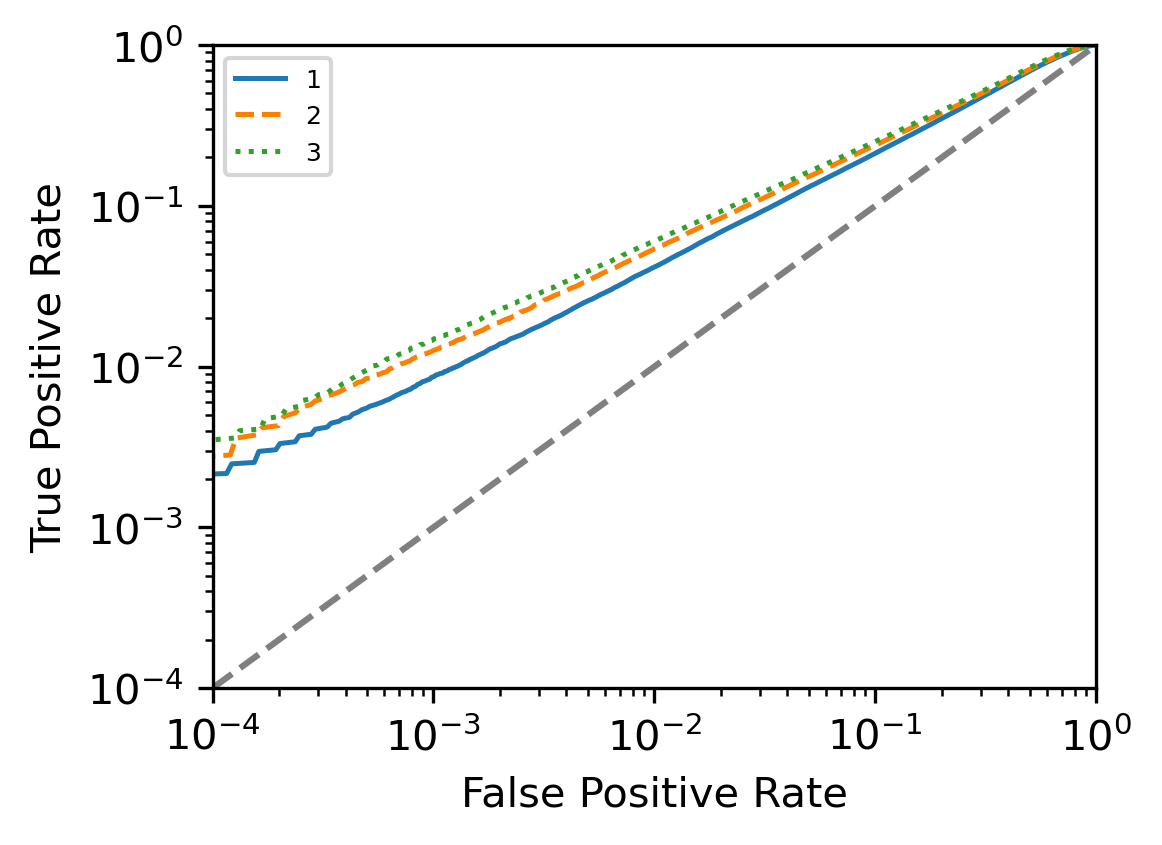}
    \caption{\textbf{Membership inference risks of ODENet-64 models.} We evaluate ODENet-64 models with varying numbers of NODE blocks against the attack of LiRA in CIFAR-10.}
    \label{fig:ode-64_blocks}
\end{figure}

\begin{figure*}[t]
    \centering
    \includegraphics[width=.3\linewidth]{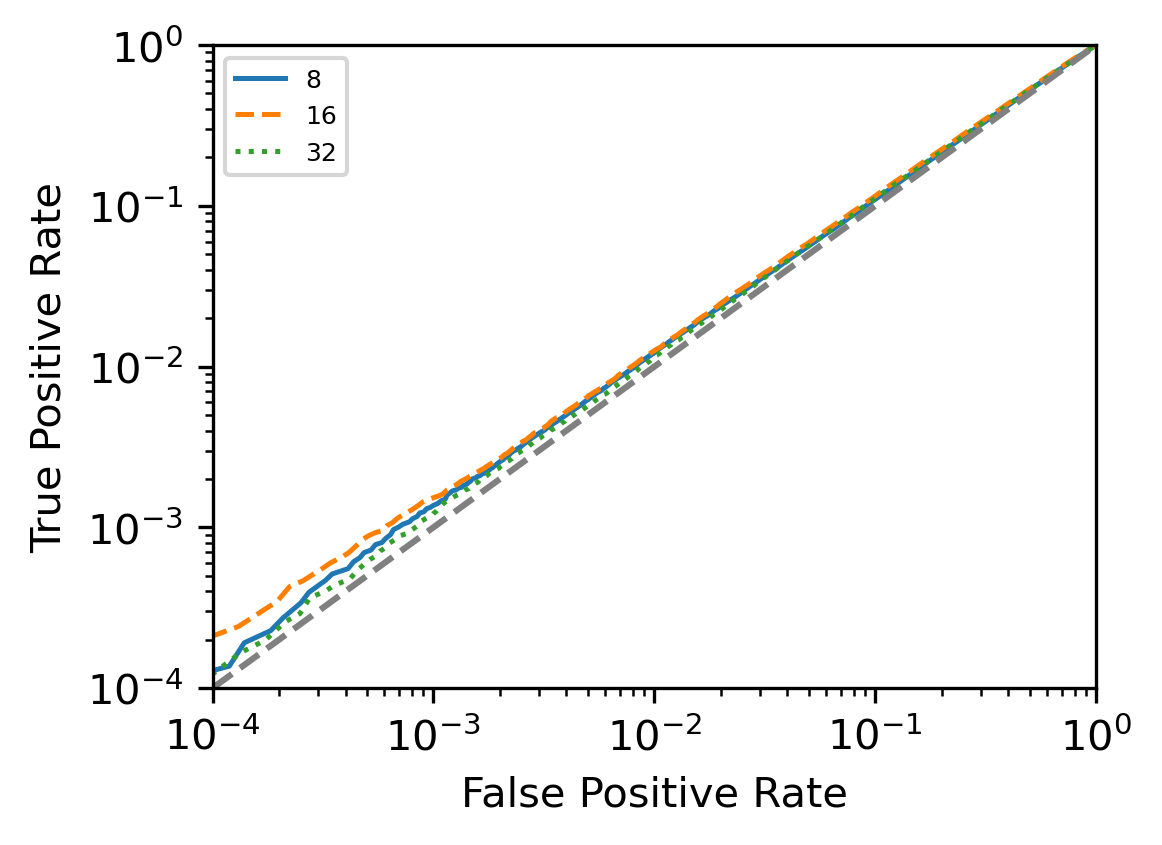}
    \includegraphics[width=.3\linewidth]{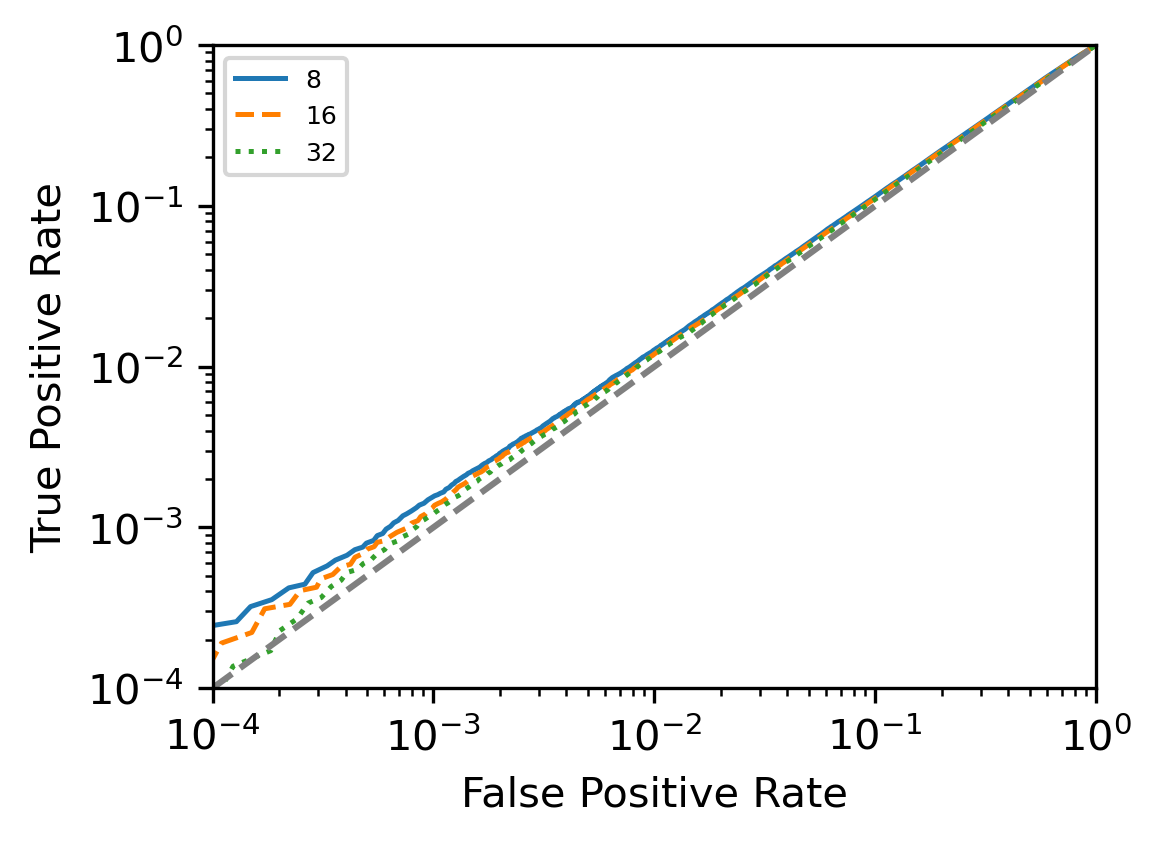}
    \includegraphics[width=.3\linewidth]{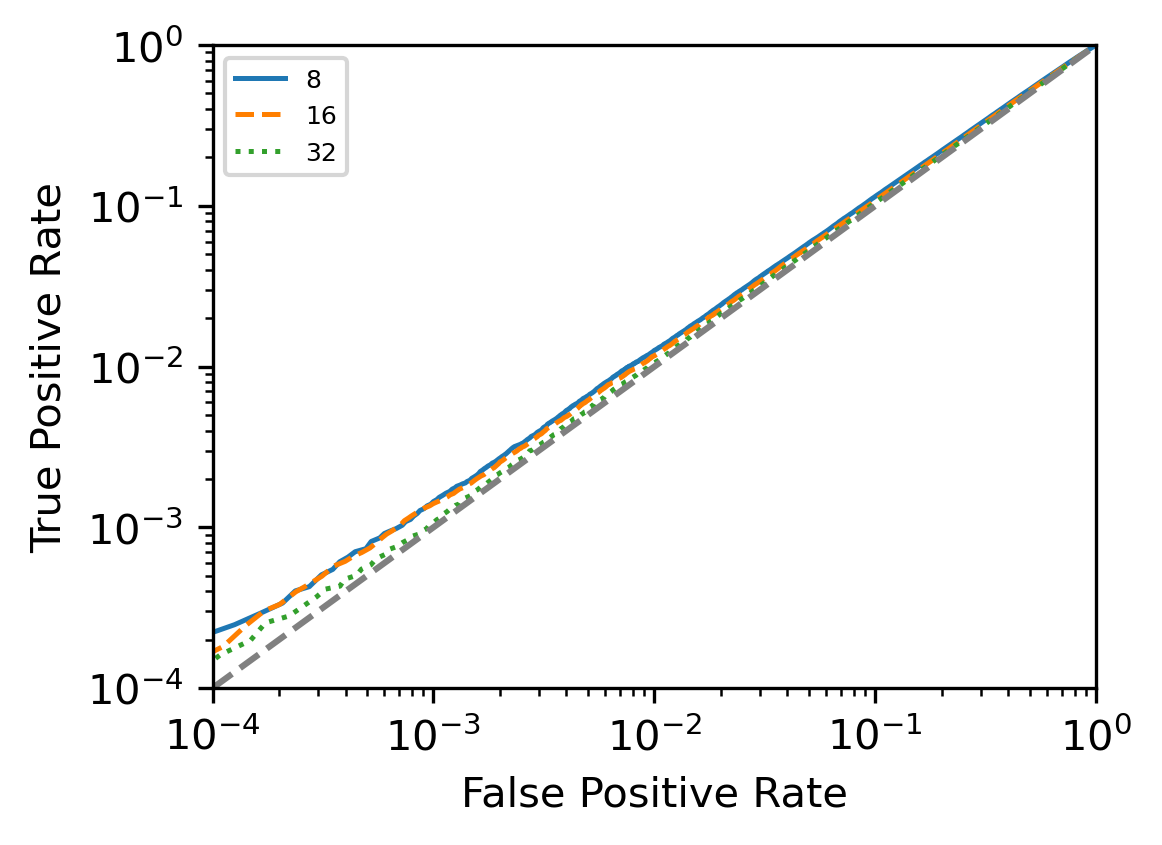}
    \includegraphics[width=.3\linewidth]{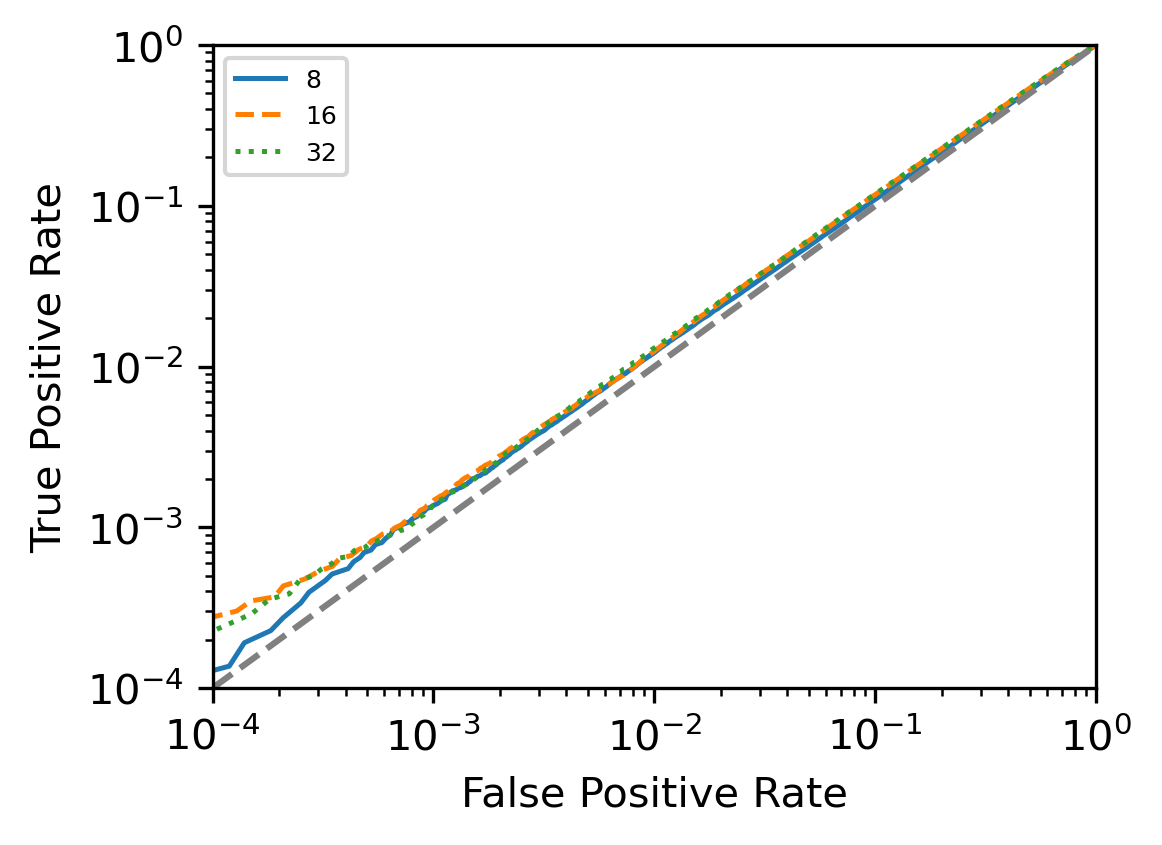}
    \includegraphics[width=.3\linewidth]{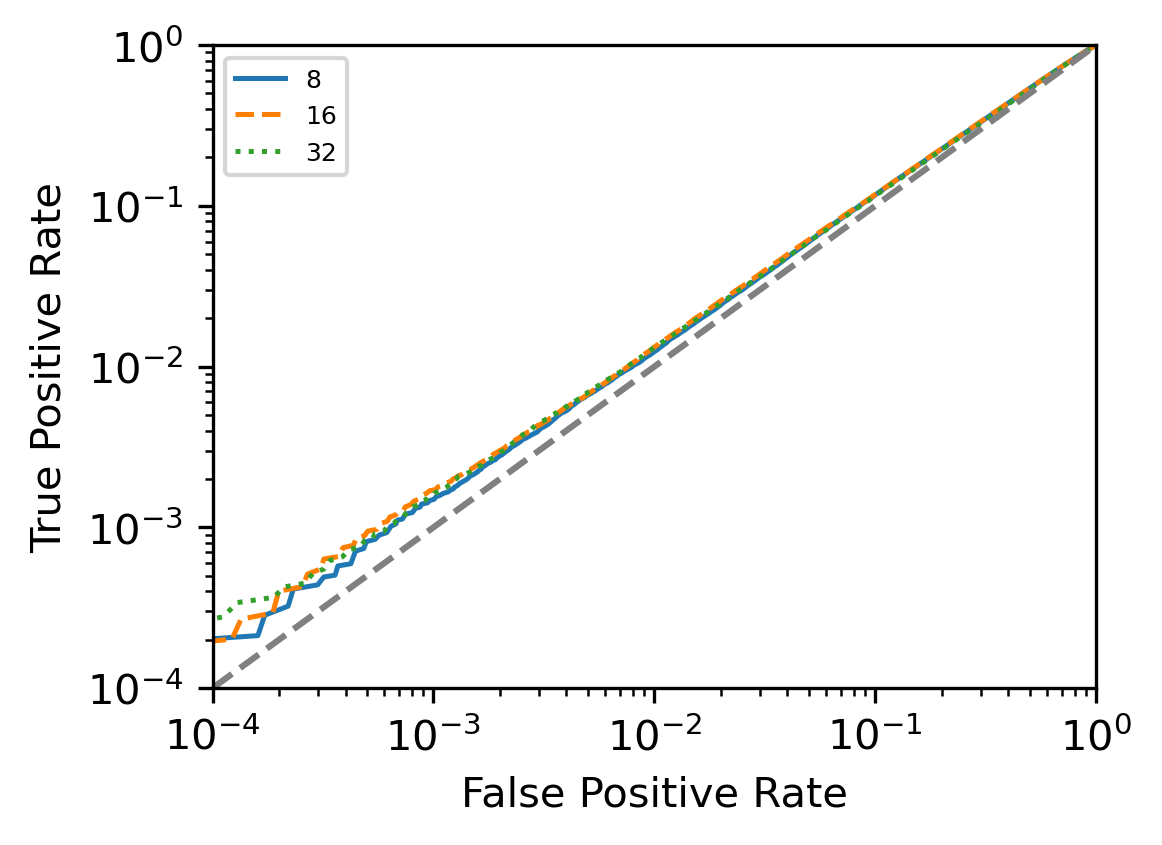}
    \includegraphics[width=.3\linewidth]{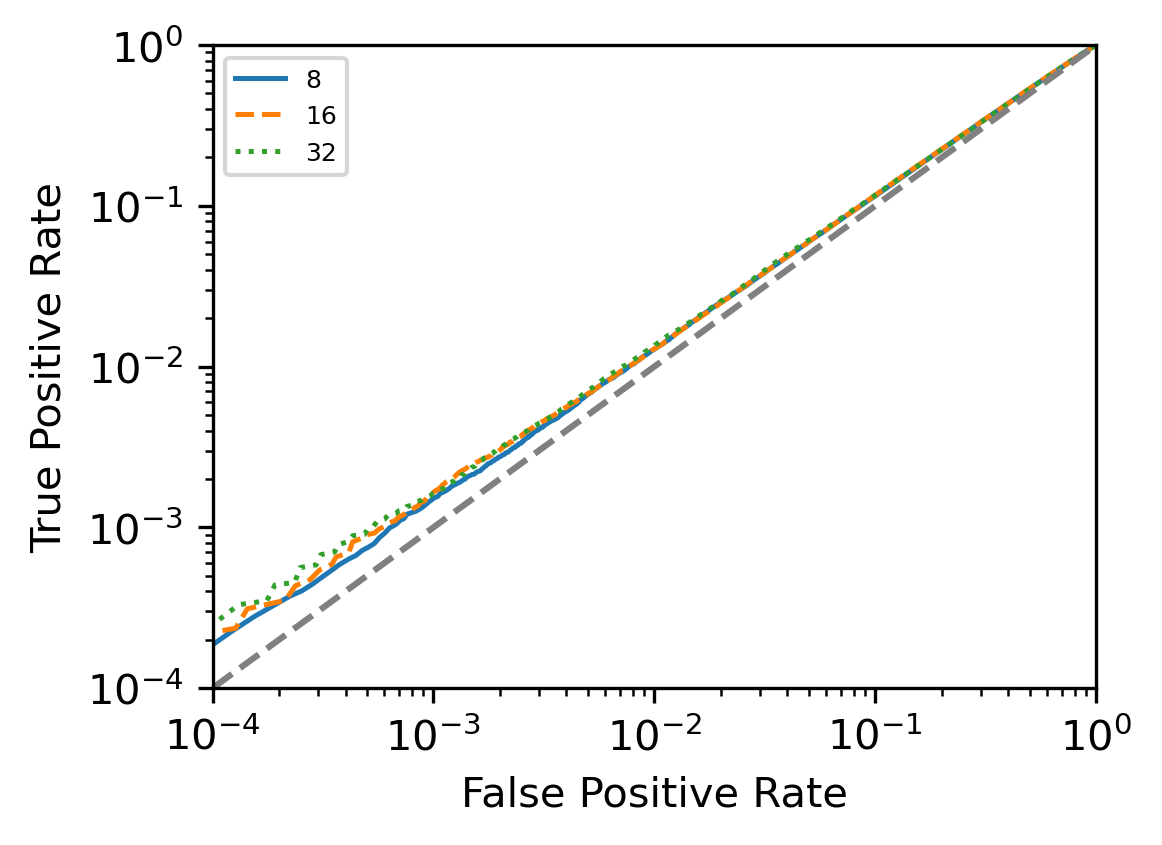}
    \vspace{-1.em}
    \caption{Membership risks of NSDEs train with different step-size ($1/s$). (Top row, left to right) We illustrate the impact of step-size ($1/s$) with the level of stochasticity $k\!=\!0.5$, and varying the integration Time $T\!=\!0.5, 0.25, 0.125$, respectively. (Bottom row, left to right) We illustrate the impact of step-size ($1/s$) with the level of noise intensity $\sigma\!=\!2$, and varying the integration Time $T\!=\!0.5, 0.25, 0.125$, respectively.}
\end{figure*}

\begin{figure*}[t]
    \centering
    \includegraphics[width=.23\linewidth]{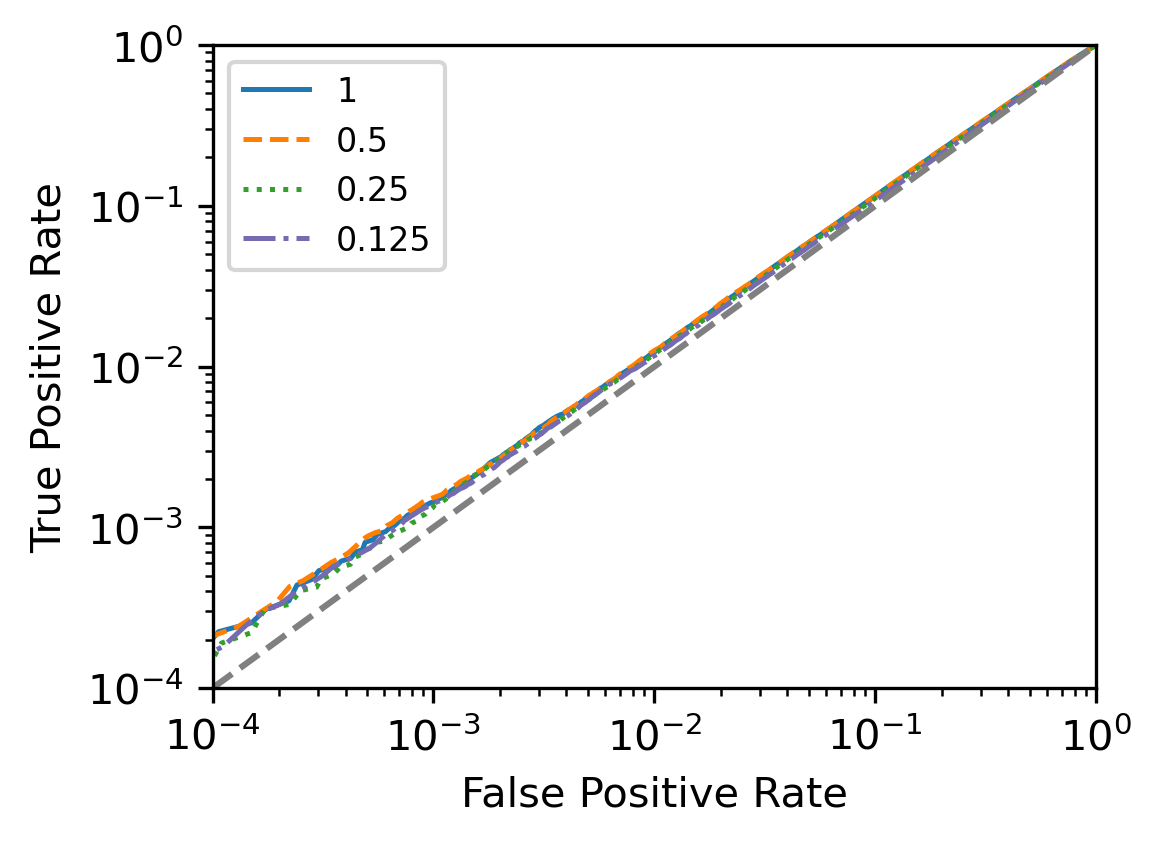}
    \includegraphics[width=.23\linewidth]{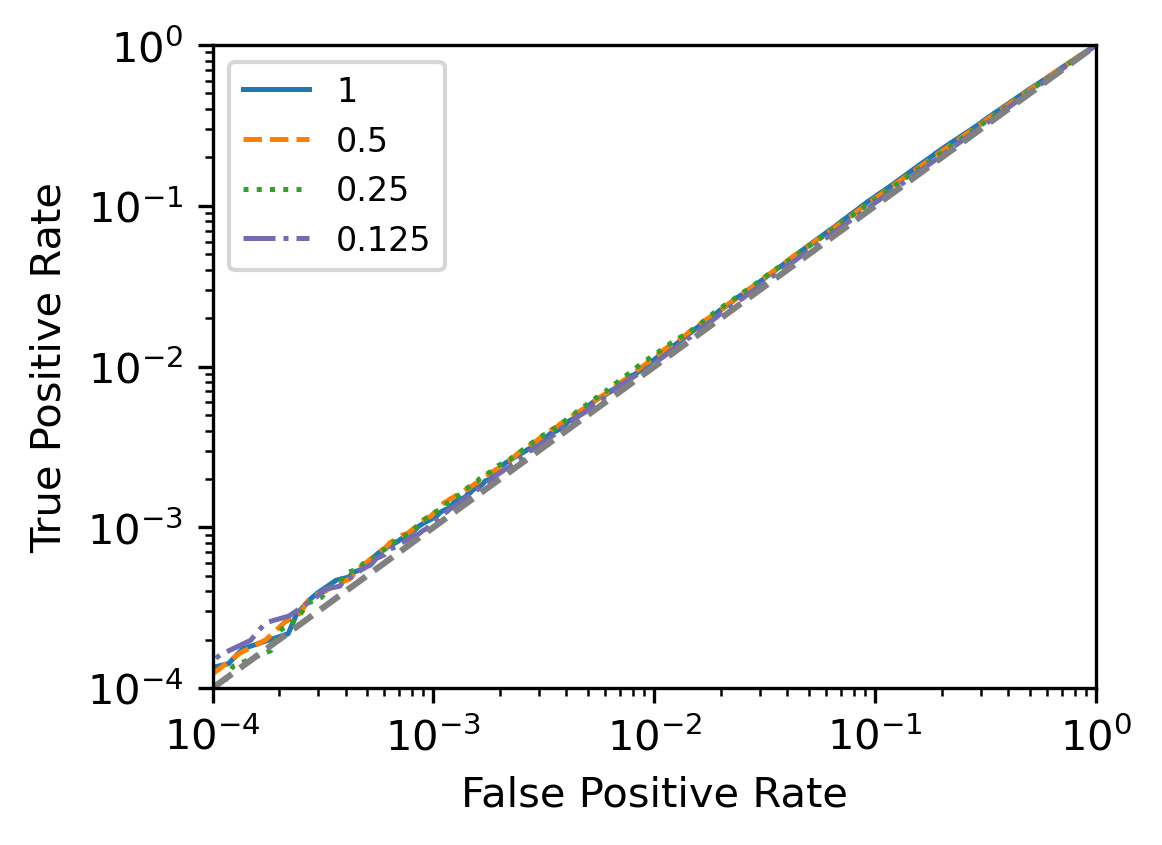}
    \includegraphics[width=.23\linewidth]{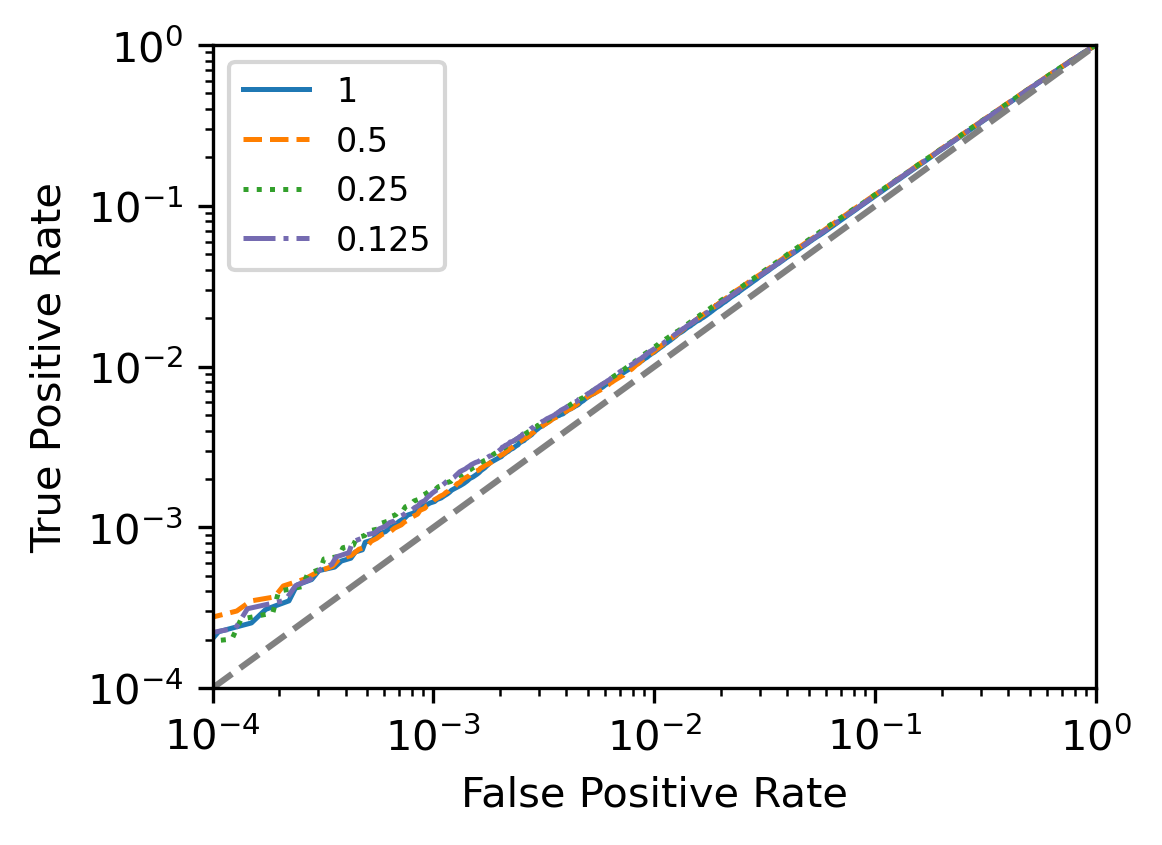}
    \includegraphics[width=.23\linewidth]{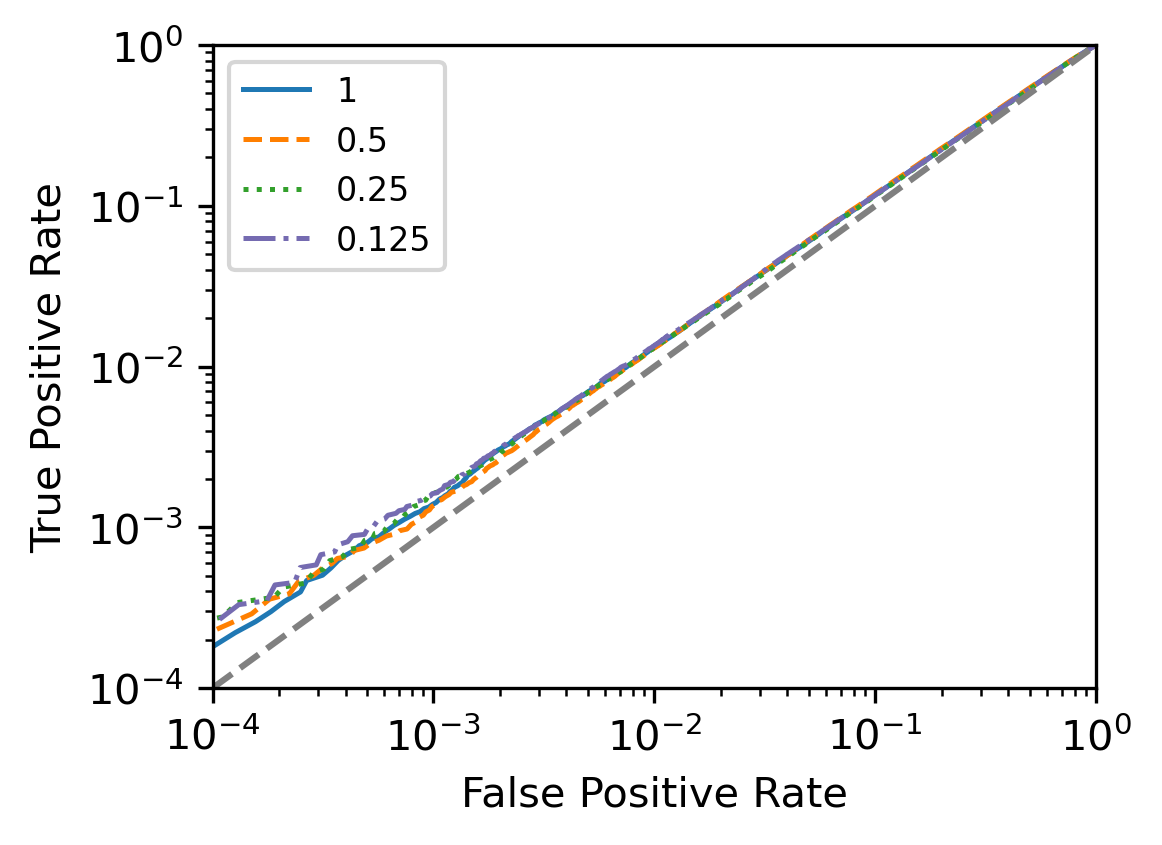}
    \vspace{-1.em}
    \caption{Membership risks of NSDEs train with different time interval ($T$). We illustrate the impact of time ($T$) with the level of stochasticity $k\!=\!0.5$ (left two) and noise intensity $\sigma\!=\!2$ (right two),  and varying $s\!=\!16, 32$, respectively.}
\end{figure*}

\begin{table*}[t]
\centering
\caption{%
    \textbf{Full evaluation results of different NSDE configurations.}
}
\label{tbl:nsde-ablation}
\vspace{-1.em}
\adjustbox{max width=\linewidth}{
\begin{tabular}{@{}lccccccccccc@{}}
\toprule
\textbf{Impact} & 
    \textbf{Step\_size} & \textbf{Time} & \textbf{Stochasticity} & \textbf{Noise} &
    \textbf{Train Acc.} & \textbf{Test Acc.} & \textbf{TPR @ 0.1\% FPR} & \textbf{TPR @ 1\% FPR} &
    \textbf{AUC} & 
    \textbf{Inference acc.} & \\ \midrule \midrule

\multirow{3}{*}{Step\_size} & 
    8 & 1 & 0.5 & 1.4& 83.27\% & 82.38\% & 0.15\% & 1.37\% &  0.527& 52.04\% \\
    & 16 & 1 & 0.5 & 2 & 81.89\% & 81.91\% & 0.14\% & 1.24\% & 0.525 &  51.89\%  \\
    & 32 & 1 & 0.5 & 2.8& 46.65\% & 35.86\% & 0.14\% & 1.29\% & 0.523 & 52.10\%   \\ \midrule
\multirow{3}{*}{Step\_size} & 
    8 & 0.5 & 0.5 & 2& 83.09\% & 82.52\% & 0.14\% & 1.27\% & 0.527 & 51.60\% \\
    & 16 & 0.5 & 0.5 & 2.8 & 81.61\% & 82.00\% & 0.15\% & 1.26\% & 0.527 &  51.99\%   \\
    & 32 & 0.5 & 0.5 & 4& 79.26\% & 80.92\% & 0.14\% & 1.15\% & 0.519 &  51.37\%  \\ \midrule
\multirow{3}{*}{Step\_size} & 
    8 & 0.25& 0.5& 2.8& 82.40\% & 82.05\% & 0.16\% & 1.27\% & 0.525 & 51.88\% \\
    & 16 & 0.25& 05& 4& 80.89\% & 81.21\% & 0.13\% & 1.20\% & 0.521 & 51.59\% \\
    & 32 & 0.25& 0.5& 5.7& 77.90\% & 80.12\% & 0.12\% & 1.19\% & 0.514 & 50.84\% \\ \midrule
\multirow{3}{*}{Step\_size} & 
    8 & 0.125 & 0.5 & 4& 80.61\% & 80.77\% & 0.14\% & 1.26\% & 0.522 & 51.49\% \\
    & 16 & 0.125 & 0.5 & 5.7 & 78.52\% & 79.66\% & 0.14\% & 1.18\% & 0.514 &  51.88\%   \\
    & 32 & 0.125 & 0.5 & 8& 75.01\% & 77.43\% & 0.11\% & 1.06\% & 0.507 & 50.41\%   \\ \midrule
\multirow{3}{*}{Step\_size} & 
    8 & 1& 0.71& 2& 81.66\% & 81.96\% & 0.17\% & 1.34\% & 0.524 & 51.67\%  \\
    & 16 & 1& 0.5& 2& 81.89\% & 81.91\% & 0.14\% & 1.24\% & 0.525 & 51.89\% & \\
    & 32 & 1& 0.35& 2& 81.64\% & 82.02\% & 0.15\% & 1.31\% & 0.527  & 51.78\% &\\ \midrule
\multirow{3}{*}{Step\_size} & 
    8 & 0.5& 0.5& 2& 83.09\% & 82.52\% & 0.14\% & 1.27\% & 0.527 & 51.60\% \\
    & 16 & 0.5& 0.35& 2& 83.09\% & 82.50\% & 0.14\% & 1.25\% & 0.528 & 52.25\% \\
    & 32 & 0.5& 0.25& 2& 82.98\% & 82.32\% & 0.14\% & 1.30\% & 0.530 & 52.03\% \\ \midrule
\multirow{3}{*}{Step\_size} & 
    8 & 0.25& 0.35& 2& 83.70\% & 82.46\% & 0.15\% & 1.26\% & 0.528 & 52.17\% \\
    & 16 & 0.25& 0.25& 2& 83.53\% & 82.49\% & 0.18\% & 1.33\% & 0.529 & 52.71\% \\
    & 32 & 0.25& 0.18& 2& 83.44\% & 82.52\% & 0.17\% & 1.32\% & 0.524 & 51.67\% \\ \midrule
\multirow{3}{*}{Step\_size} & 
    8 & 0.125 & 0.25 & 2 & 82.98\% & 81.95\% & 0.16\% & 1.29\% & 0.525 & 52.34\% \\
    & 16 & 0.125 & 0.18 & 2 & 82.91\% & 81.71\% & 0.17\% & 1.28\% & 0.527 &  51.55\%  \\
    & 32 & 0.125 & 0.13 & 2 & 82.88\% & 81.70\% & 0.17\% & 1.36\% & 0.526 & 52.34\% \\ \midrule
\multirow{4}{*}{Time} & 
    8 & 1 & 0.5 & 1.4& 83.27\% & 82.38\% & 0.15\% & 1.37\% & 0.528 & 52.04\% \\
    & 8 & 0.5 & 0.5 & 2 & 83.09\% & 82.52\% & 0.14\% & 1.27\% & 0.527 &  51.60\%   \\
    & 8 & 0.25 & 0.5 & 2.8& 82.40\% & 82.05\% & 0.16\% & 1.27\% & 0.525 & 51.88\%   \\ 
    & 8 & 0.125 & 0.5 & 4& 80.61\% & 80.77\% & 0.14\% & 1.26\% & 0.523 & 51.49\%   \\    
    \midrule
\multirow{4}{*}{Time} & 
    16 & 1 & 0.5 & 2& 81.89\% & 81.91\% & 0.14\% & 1.24\% & 0.525 & 51.89\% \\
    & 16 & 0.5 & 0.5 & 2.8 & 81.61\% & 82.00\% & 0.15\% & 1.26\% & 0.527 &  51.99\%  \\
    & 16 & 0.25 & 0.5 & 4& 80.89\% & 81.21\% & 0.13\% & 1.20\% & 0.521 & 51.59\%   \\ 
    & 16 & 0.125 & 0.5 & 5.7& 78.52\% & 79.66\% & 0.14\% & 1.18\% & 0.514 & 51.88\%  \\ \midrule
\multirow{4}{*}{Time} & 
    32 & 1 & 0.5 & 2.8& 46.65\% & 35.86\% & 0.14\% & 1.29\% & 0.523 & 52.10\% \\
    & 32 & 0.5 & 0.5 & 4 & 79.26\% & 80.92\% & 0.14\% & 1.15\% & 0.519 &  51.37\%   \\
    & 32 & 0.25 & 0.5 & 5.7& 77.90\% & 80.12\% & 0.12\% & 1.19\% & 0.514 & 50.84\%  \\ 
    & 32 & 0.125 & 0.5 & 8& 74.60\% & 77.22\% & 0.12\% & 1.01\% & 0.510 & 50.14\% &  \\ \midrule
\multirow{4}{*}{Time} & 
    8 & 1 & 0.71 & 2& 81.66\% & 81.96\% & 0.17\% & 1.34\% & 0.524 & 51.67\% \\
    & 8 & 0.5 & 0.5 & 2 & 83.09\% & 82.52\% & 0.14\% & 1.27\% & 0.527 &  51.60\%   \\
    & 8 & 0.25 & 0.35 & 2& 83.70\% & 82.46\% & 0.15\% & 1.26\% & 0.528 & 52.17\%  \\ 
    & 8 & 0.125 & 0.25 & 2& 82.98\% & 81.95\% & 0.16\% & 1.29\% & 0.525 & 52.34\%  \\    
    \midrule
\multirow{4}{*}{Time} & 
    16 & 1 & 0.5 & 2& 81.89\% & 81.91\% & 0.14\% & 1.24\% & 0.525 & 51.89\% \\
    & 16 & 0.5 & 0.35 & 2 & 83.09\% & 82.50\% & 0.14\% & 1.25\% & 0.528 &  0.523\%   \\
    & 16 & 0.25 & 0.25 & 2& 83.53\% & 82.49\% & 0.18\% & 1.33\% & 0.529 & 52.71\%  \\ 
    & 16 & 0.125 & 0.18 & 2& 82.91\% & 81.71\% & 0.17\% & 1.28\% & 0.527 & 51.6\%   \\ \midrule
\multirow{4}{*}{Time} & 
    32 & 1 & 0.35 & 2& 81.64\% & 82.02\% & 0.15\% & 1.31\% & 0.527 & 51.78\% \\
    & 32 & 0.5 & 0.25 & 2 & 82.98\% & 82.32\% & 0.14\% & 1.30\% & 0.530 &  52.03\%    \\
    & 32 & 0.25 & 0.18 & 2& 83.44\% & 82.52\% & 0.17\% & 1.32\% & 0.524 & 51.67\%   \\ 
    & 32 & 0.125 & 0.13 & 2& 82.88\% & 81.70\% & 0.17\% & 1.36\% & 0.526 & 52.34\%  \\ \midrule
\multirow{4}{*}{k} & 
    16 & 1 & 0.3 & 1.2& 83.56\% & 82.62\% & 0.15\% & 1.29\% & 0.530 & 51.29\% \\
    & 8 & 1 & 0.5 & 1.4 & 82.38\% & 82.38\% & 0.15\% & 1.37\% & 0.528 &  52.04\%   \\
    & 16 & 1 & 0.4 & 1.6& 82.58\% & 82.22\% & 0.17\% & 1.40\% &0.528  & 52.27\%   \\ 
    & 8 & 0.5 & 0.5 & 2& 83.09\% & 82.52\% & 0.14\% & 1.27\% & 0.527 & 51.60\% &  \\
    & 16 & 1 & 0.6 & 2.4 & 63.23\% & 57.12\% & 0.09\% & 1.09\% & 0.516 &  52.02\%   \\
    & 8 & 0.25 & 0.5 & 2.8& 82.40\% & 82.05\% & 0.16\% & 1.27\% & 0.525 & 51.88\% &  \\ 
    & 16 & 0.25 & 0.5 & 4& 80.89\% & 81.21\% & 0.13\% & 1.20\% & 0.521 & 51.59\% &  \\
    & 32 & 0.25 & 0.5 & 5.7 & 77.90\% & 80.12\% & 0.12\% & 1.19\% & 0.514 &  50.84\%  &  \\
    & 32 & 0.125 & 0.5 & 8& 75.01\% & 77.43\% & 0.11\% & 1.06\% & 0.507 & 50.41\% &  \\ 
    \midrule
\multirow{4}{*}{Stochasticity} & 
    0.3 & 1& 16& 1.2& 83.56\% & 82.62\% & 0.15\% & 1.29\% & 0.516 & 52.02\% & \\
    & 0.4 & 1& 16& 1.6& 82.58\% & 82.22\% & 0.17\% & 1.40\% & 0.528 & 52.27\% & \\
    & 0.5 & 1& 16& 2& 81.89\% & 81.91\% & 0.14\% & 1.24\% & 0.525 & 51.89\% &\\   
    & 0.6 & 1& 16& 2.4& 63.23\% & 57.12\% & 0.09\% & 1.09\% & 0.516 & 52.02\% &\\ \bottomrule
\end{tabular}
}
\end{table*}

\end{document}